\documentclass{article}
\usepackage[utf8]{inputenc}
\usepackage{amsmath,amscd,amssymb,amsfonts,amsthm,courier,relsize,bm}

\usepackage{bbm}
\usepackage{hyperref,enumerate,mathrsfs,mathtools,slashed}
\usepackage{tikz-cd}
\tikzcdset{ampersand replacement=\&}
\usepackage{xcolor}  
\hypersetup{
    colorlinks,
    linkcolor={red!50!black},
    citecolor={blue!70!black},
    urlcolor={blue!80!black}
}
\DeclareFontFamily{U}{mathx}{}
\DeclareFontShape{U}{mathx}{m}{n}{<-> mathx10}{}
\DeclareSymbolFont{mathx}{U}{mathx}{m}{n}
\DeclareMathAccent{\widehat}{0}{mathx}{"70}
\DeclareMathAccent{\widecheck}{0}{mathx}{"71}

\newtheorem{theorem}{Theorem}[section]

\newtheorem{proposition}[theorem]{Proposition}

\theoremstyle{definition}    

\theoremstyle{remark}

\newtheorem{remark}[theorem]{Remark}

\newtheorem{example}[theorem]{Example}

\newcommand{\ignore}[1]{}

\def\la{\ensuremath{\langle}}
\def\ra{\ensuremath{\rangle}}
\def\lp{\ensuremath{\left(}}
\def\rp{\ensuremath{\right)}}

\def\ad{\ensuremath{\textnormal{ad}}}

\def\t{\ensuremath{\mathfrak{t}}}
\def\n{\ensuremath{\mathfrak{n}}}

\def\h{\ensuremath{\mathfrak{h}}}

\def\gl{\ensuremath{\mathfrak{gl}}}
\def\so{\ensuremath{\mathfrak{so}}}
\def\o{\ensuremath{\mathfrak{o}}}

\def\A{\ensuremath{\mathcal{A}}}

\def\E{\ensuremath{\mathcal{E}}}
\def\F{\ensuremath{\mathcal{F}}}

\def\H{\ensuremath{\mathcal{H}}}

\def\L{\ensuremath{\mathcal{L}}}

\def\O{\ensuremath{\mathcal{O}}}

\def\R{\ensuremath{\mathcal{R}}}
\def\S{\ensuremath{\mathcal{S}}}
\def\T{\ensuremath{\mathcal{T}}}
\def\U{\ensuremath{\mathcal{U}}}
\def\V{\ensuremath{\mathcal{V}}}

\def\Cc{\ensuremath{\mathscr{C}}}

\def\Sc{\ensuremath{\mathscr{S}}}

\def\Mc{\ensuremath{\mathscr{M}}}

\def\Deltab{\ensuremath{\boldsymbol{\Delta}}}

\def\Deltab{\ensuremath{\boldsymbol{\Delta}}}

\def\bOmega{\ensuremath{\boldsymbol{\Omega}}}

\def\bC{\ensuremath{\mathbb{C}}}
\def\bR{\ensuremath{\mathbb{R}}}
\def\bZ{\ensuremath{\mathbb{Z}}}

\def\bQ{\ensuremath{\mathbb{Q}}}

\def\bA{\ensuremath{\mathbb{A}}}
\def\bD{\ensuremath{\mathbb{D}}}
\def\bS{\ensuremath{\mathbb{S}}}

\def\bX{\ensuremath{\mathbb{X}}}
\def\bE{\ensuremath{\mathbb{E}}}

\def\End{\ensuremath{\textnormal{End}}}

\def\Ker{\ensuremath{\textnormal{Ker}}}

\def\tr{\ensuremath{\textnormal{tr}}}

\title{Hypoellipticity of the Asymptotic Bismut Superconnection on Contact Manifolds}
\author{Jesus Sanchez Jr$^1$, Andres Franco Valiente$^2$}
\date{%
    $^1$Mathematics Department, Texas A\&M University\\%
    $^2$Leinweber Institute for Theoretical Physics, University of California, Berkeley\\[2ex]%
}

\begin{document}
\maketitle
\begin{abstract}
  Given a contact sub-Riemannian manifold one obtains a non-integrable splitting of the tangent bundle into the directions along the contact distribution and the Reeb field. We generalize the construction of the Bismut superconnection to this non-integrable setting and show that although singularities appear within the superconnection, if one extracts the finite part then the resulting operator is hypoelliptic. We find that the hypoellipticity also holds in the setting of two-step subRiemannian manifolds and produce a modification for arbitrary subRiemannian manifolds which always gives a hypoelliptic operator. A discussion of the explicit form of the operator on principal $\bS^1$-bundles is provided. The index theory is worked out on contact manifolds and a matrix twisting of the Clifford relations produces operators with non-trivial Fredholm index. We conclude with a possible relationship between our hypoelliptic operator and a constrained supersymmetric sigma model.
\end{abstract}
{
  \hypersetup{linkcolor=black}
  \tableofcontents
}
\section{Introduction}
On a Riemannian manifold $\lp M,g\rp$ the Dirac-type operators have proven to be the golden operators of study among the class of differential operators. They have demonstrated a wide range of applicability and have been at the center of connections relating analysis, topology, and geometry on $M$. For example, they play a fundamental role in the various proofs of the Atiyah-Singer index theorems, obstructions to positive scalar curvature, Connes noncommutative geometry program, and appear as essential terms within the nonlinear equations of gauge theory. Recently the subject of index theory and microlocal analysis have seen a huge expansion given by the work of van Erp \cite{Van}, van Erp-Yuncken \cite{Van-Yun}, and Mohsen \cite{Moh} into the realm of hypoelliptic operators on filtered manifolds. The techniques rely heavily on noncommutative harmonic analysis and the tangent groupoid machinery developed by Connes and his contemporaries and is perhaps the most recent success story of noncommutative geometry. In a recent paper of Dave-Haller \cite{DH}, a plethora of natural hypoelliptic operators have been supplied to a given filtered manifold, thus opening up the range of applicability of these new hypoelliptic index theorems as can be seen in as in the work of Goffeng \cite{Go}, Goffeng-Helffer \cite{GH}, and Goffeng-Kuzmin \cite{GK}. For a given filtration $\mathcal{F}$ of the tangent bundle $TM$ one can equip $\lp M,\F\rp$ with a sub-Riemannian metric $g_{\F}$, and a natural question arises: is there a canonical class of first order differential operators associated to $\lp M,\F,g_{\F}\rp$ which plays the role of the Dirac-type operators in sub-Riemannian geometry? This question has already been approached from the spectral triple perspective in the work of Hasselmann \cite{Has} as well as in the CR geometric setting in the work of Stadmüller \cite{S}, the verdict being that the naive choice of horizontal Dirac operator does recover the sub-Riemannian distance and has similar geometric properties like that of ordinary Dirac operators such as a Weitzenböck identity containing curvature terms but also have undesirable analytic properties such as having infinite dimensional eigenspaces on closed manifolds and lacking hypoellipticity.

\subsection{Main Results}
Our aim in this paper is to explore this question with the construction of a differential operator $\bD$ associated to a sub-Riemannian manifold $\lp M,\F,g_{\F}\rp$ equipped with a transverse distribution $\mathcal{T}$ which splits the tangent bundle $TM=\F\oplus\T$. The idea has inspiration from both mathematics and physics, the mathematical side coming from Bismut's construction of Levi-Civita superconnection in the proof of the local families index theorem \cite{Bis3} and the physics side coming from Witten's supersymmetric quantum mechanical proof of the Atiyah-Singer index theorem \cite{Wit}. We begin with an auxiliary transverse metric $g_{\T}$ and obtain an ordinary Riemannian metric $g_M=g_{\F}\oplus g_{\T}$ on $M$. Following Bismut, we apply an adiabatic scaling to $g_{\T}$ and consider the one-parameter family of metrics $g_{M,u}=g_{\F}\oplus u^{-1}g_{\T}$. Using a projection, one then obtains a pair of one-parameter families of metric-compatible connections $\nabla^{\F,u}$ and $\nabla^{\T,u}$ and equips the covariant Euclidean vector bundle $\lp\F,g_{\F},\nabla^{\F,u}\rp$ with a $\bZ_2$-graded Clifford module 
\[
    C\ell\lp\F^*,g^*_{\F}\rp\circlearrowright_c\lp E,h_E,\nabla^{E,u}\rp
\]
which is tensored with the Clifford module
\[
    C\ell\lp\T^*,ug^*_{\T}\rp\circlearrowright_{c^{\T}_u}\lp \wedge\T^*,\wedge ug^*_{\T},\nabla^{\T,u}\rp
\]
to obtain the Clifford module 
\[
    C\ell\lp T^*M,g^*_{M,u}\rp\circlearrowright_{m_u}\lp E\hat{\otimes}\wedge\T^*,h_E\otimes\wedge ug^*_{\T},\nabla^{E,u}\otimes I+I\otimes \nabla^{\T,u}\rp.
\]
Bismut then introduces a modified connection $\nabla^{\bE,u}$ for $E\otimes\wedge\T^*$ which is compatible with the Levi-Civita connection of $g_{M,u}$ and proceeds to investigate the behavior of the one-parameter family of Dirac operators
\[
    D_u:=m_u\circ \nabla^{E\otimes\wedge\T^*,u}
\]
in the limit as $u\rightarrow 0^+$. When the distribution $\F$ is the vertical distribution associated to a locally trivial fibration $\pi:M\rightarrow B$, then one has that the limit as $u\rightarrow 0^+$ exists
\[
    \bD^E :=\lim_{u\rightarrow 0^+}D_u
\]
and is the Bismut superconnection $\bA$, up to a zeroth order term. However, when $\F$ is a \emph{non-integrable} distribution, singularities will appear in the limit as $u\rightarrow 0^+$ thus initially warranting one to dismiss the non-integrable setting. In this paper, we show that if one extracts the finite part of the limit of $D_u$ as $u\rightarrow 0^+$, then there are situations in which the operator
\[
    \bD^E:=\lim^{F.P.}_{u\rightarrow 0^+}D_u
\]
is a \textbf{\emph{hypoelliptic operator}}. The simplest example of a non-trivial non-integrable filtered manifold is a contact manifold $\lp M,\theta\rp$ in which case the filtration is given by the kernel of the contact 1-form $\F=\operatorname{Ker}\theta$ and in this case we show that for \textbf{\emph{any choice of sub-Riemannian metric}}, there are natural Clifford modules which admit hypoelliptic \textbf{\emph{asymptotic Bismut superconnections}}. 
\begin{theorem}
    Suppose that $\lp M,\theta,g_{\F}\rp$ is a contact sub-Riemannian manifold. Then for an appropriate choice of $\bZ_2$-graded covariant Hermitian Clifford module
    \[
        \lp C\ell\lp\F^*,g^*_{\F}\rp,\nabla^{\F,u}\rp\circlearrowright_c\lp E,h_E,\nabla^{E,u}\rp
    \]
    the operator
    \begin{align*}
    \bD^E&:=\lim_{u\rightarrow 0^+}^{F.P.}D_u
\end{align*}
is hypoelliptic. If $M$ is closed, then $\Ker\bD^{E}$ is a finite dimensional space of smooth sections.
\end{theorem}
We give two proofs of this result. The first demonstrates that the square $\Deltab^E=\lp\bD^E\rp^2$ is a Rockland operator within the van-Erp Yuncken Heisenberg calculus of $\lp M,\theta\rp$ \cite{Van-Yun}. The awkward feature of this proof is that $\bD^E$ is \emph{not} Rockland within the van-Erp-Yuncken calculus but its square $\Deltab^E$ is. In the second proof, we utilize a natural grading on the bundle $E\otimes\wedge\T^*$ and use the graded Heisenberg calculus of Dave-Haller \cite{DH} to give a more refined analytic understanding of $\bD^E$.
\begin{theorem}
    Suppose that $\lp M,\theta,g_{\F}\rp$ is a contact sub-Riemannian manifold and we have a $\bZ_2$-graded covariant Hermitian Clifford module
    \[
        \lp C\ell\lp\F^*,g^*_{\F}\rp,\nabla^{\F,0}\rp\circlearrowright_c\lp E,h_E,\nabla^E\rp
    \]
    Then the operator
    \begin{align*}
    \bD^E=D^{\F,E\otimes\wedge[\theta]}+\varepsilon^{\T}_{\theta_{\T}}\nabla^{E\otimes\wedge[\theta]}_T+\frac{1}{4}c\lp\iota_T\Omega_{\F}\rp\iota^{\T}_T+\frac{1}{4}\varepsilon^{\T}\circ\tr Lg_{\F}
\end{align*}
is a graded Rockland operator of order $1$.
\end{theorem}
Our proof extends to the more general setting of \emph{arbitrary} two-step filtered subRiemannian manifolds $\lp M,\F,g_{\F}\rp$ equipped with a Euclidean distribution $\lp\T,g_{\T}\rp$ which is transverse $TM=\F\oplus\T$. In the presence of a Lie group $G$ acting on $M$ and preserving the geometric structures $\lp\F,g_{\F}\rp$ and $\lp\T,g_{\T}\rp$ the asymptotic Bismut superconnection $\bD^E$ will be a $G$-equivariant operator and its kernel will give a finite dimensional representation of $G$.\par
We then inquire on the generality of the construction, exploring the setting of arbitrary $m$-step filtered manifolds. We are able to show that the asymptotic Bismut superconnection has graded Heisenberg order $1$ and that a perturbation always yields a graded Rockland operator of order $1$.
\begin{theorem}
    Suppose that $\lp M,\F,g_{\F}\rp$ is an $m$-step filtered subRiemannian manifold and $\lp\T^j,g_{\T^j}\rp$ is a sequence of Euclidean distributions which are transverse $\F^j=\F^{j-1}\oplus\T^j$. Suppose that $\F^1$ is bracket generating. Then for a $\bZ_2$-graded covariant Hermitian Clifford module
    \[
        \lp C\ell\lp\F^*,g^*_{\F}\rp,\nabla^{\F,0}\rp\circlearrowright\lp E,h_E,\nabla^E\rp
    \]
    the operator
    \[
        \widetilde{\bD}^E:=D^{\F,E\otimes\wedge\T^*}+d^{\T}_{\nabla^E}+\lim^{F.P.}_{u\rightarrow 0^+}m_u\circ m_u\lp\omega_u\rp
    \]
   has graded Heisenberg order $1$ and satisfies the graded Rockland condition.
\end{theorem}
It is natural to inquire what one can expect to find within the kernel of $\bD^E$ and we do so for the case of principal $\bS^1$-bundles $\lp P,\theta\rp\circlearrowleft\bS^1\rightarrow \Sigma$ over a closed manifold with non-vanishing curvature.
\begin{theorem}
    Suppose that $\Sigma$ is a closed manifold, $\lp P,\theta\rp\circlearrowleft\bS^1\xrightarrow{\pi}\Sigma$ is a covariant principal $\bS^1$-bundle with non-vanishing curvature $d\theta=\pi^*\bOmega$, and $g_{\Sigma}$ is a Riemannian metric of $\Sigma$. Then if we have a $\bZ_2$-graded covariant Hermitian Clifford module
    \[
        \lp C\ell\lp T^*\Sigma,g^*_{\Sigma}\rp,\nabla^{g_{\Sigma}}\rp\circlearrowright_c\lp E,h,\nabla^E \rp
    \]
    the kernel of the asymptotic Bismut superconnection $\bD^{\pi^*E}$ for the metric $g_{\F}:=\pi^*g_{\Sigma}$ defined above is given by 
    \begin{itemize}
        \item The even component of $\Ker\bD^{\pi^*E}$ is isomorphic to the direct summand
        \[
            \hspace{-.8cm}\Ker^0D^E\oplus\left\{ \tau\in\Ker^1D^E\Big|\text{ $c\lp\bOmega\rp\tau\perp \Ker^1D^E$}\right\}\bigoplus_{m\neq0}\left\{\tau_m\in\Gamma^1\lp E\otimes L_m\rp\Big|\text{ $c\lp \bOmega\rp\tau_m=\frac{4i}{m}\lp D^{E\otimes L_m}\rp^2\tau_m$}\right\}
        \]
        \item The odd component of $\Ker\bD^{\pi^*E}$ is isomorphic to the direct summand
        \[
            \hspace{-.8cm}\Ker^1D^E\oplus\left\{\xi\in\Ker^0D^E\Big|\text{ $c\lp\bOmega\rp\xi\perp\Ker^0D^E$}\right\}\bigoplus_{m\neq0}\left\{ \xi_m\in\Gamma^0\lp E\otimes L_m\rp\Big|\text{ $c\lp\bOmega\rp\xi_m=\frac{4i}{m}\lp D^{E\otimes L_m}\rp^2\xi_m$}\right\}
        \]
    \end{itemize}
\end{theorem}
If one is given a geometric structure on a closed manifold $M$ and can naturally associate a hypoelliptic differential operator then one can hope that the dimension of the kernel of this operator provides an invariant for this geometric structure. The above calculation restricted to $\bS^2$ with Kähler structure $\lp\bS^2,g,\bOmega\rp$ show that there are simple examples where the dimension of the kernel of our operator is non-trivial as well as the $\bZ_2$-dimension. \par
As alluded to in the beginning of our introduction, we investigate the index theory of the operator $\bD^E$ in the case of contact manifolds and show that index \emph{always vanishes}.
\begin{theorem}
    Let $\lp M,\theta,g_{\F}\rp$ be a contact subRiemannian manifold which is closed. Then for an appropriate choice of $\bZ_2$-graded covariant Hermitian Clifford module
    \[  
    \lp C\ell\lp\F^*,g^*_{\F}\rp,\nabla^{\F,u}\rp\circlearrowright\lp E,h_E,\nabla^{E,u}\rp.
    \]
    the restriction of the asymptotic Bismut superconnection to the even sections 
    \[
        \bD^{E,+}:\Gamma^0\lp E\otimes\wedge[\theta]\rp\rightarrow \Gamma^1\lp E\otimes\wedge[\theta]\rp
    \]
    is Fredholm and its Fredholm index is zero.
\end{theorem}
One should contrast the above theorem with the explicit calculations done on principal $\bS^1$-bundles above, demonstrating that there are simple examples where the $\bZ_2$-dimension of the kernel of $\bD^E$ does \emph{not vanish}.\par
Taking inspiration from van Erp's trick in \cite{Van3} to construct hypoelliptic differential operators on contact manifolds with non-vanishing index, we consider further twisting the Clifford relations by matrix coefficients and generalize the construction of $\bD^E$ in the setting of contact subRiemannian manifolds equipped with a chosen even section $\gamma\in\Gamma\lp\End E\rp$ and denote the modified asymptotic Bismut superconnection by $\bD^E_{\gamma}$.
\begin{theorem}
    Let $\lp M,\theta,g_{\F}\rp$ be a contact subRiemannian manifold and consider a $\bZ_2$-graded covariant Hermitian Clifford module
    \[  
    \lp C\ell\lp\F^*,g^*_{\F}\rp,\nabla^{\F,0}\rp\circlearrowright\lp E,h_E,\nabla^E\rp.
    \]
    For an even element $\gamma\in\Gamma\lp\End E\rp$ such that $[\gamma,c\lp d\theta\rp]=0$ and for each $x\in M$ the spectrum of the operator section $\gamma c\lp d\theta\rp$ satisfies
    \[
        \textup{Spec}\lp\pm i\lp\gamma_x-1\rp c\lp d\theta_x\rp\rp \bigcap\textup{Spec}\lp\pi_{1}\lp\sigma_2\lp\nabla^E\big|_{\F}\rp^*\nabla^E\big|_{\F}\rp_x\rp=\emptyset
    \]
    where $\pi_{ 1}$ is the $\lambda= 1$ Schrodinger representation of $\t_{\F}M_x$. Then the operator
    \[
        \bD^E_{\gamma}\circlearrowright\Gamma\lp E\otimes\wedge[\theta]\rp
    \]
    is graded Rockland operator of order $1$ and is Fredholm if $M$ is closed.
\end{theorem}
We use the above theorem along with an explicit index formula of van Erp in \cite{Van3} to provide examples of even sections $\gamma$ for which $\bD^E_{\gamma}$ has non-vanishing index. These operators provide examples of \emph{first order differential operators with bundle coefficients and non-vanishing index on odd-dimensional closed manifolds}.\par
Our generalization of the Bismut superconnection to the case of filtered manifolds was largely motivated by the relationship between the Dirac operator and supersymmetric non-linear sigma models. In \cite{Wit}, Witten provides a supersymmetric quantum mechanics proof of the Atiyah-Singer index theorem and a construction of Dirac-type operators through the lenses of supersymmetry is also presented. The lesson one can take is that if one has a supersymmetric Lagrangian action functional $S$ for a particle, then the quantized Noether current $\hat{Q}$ known as a supercharge and quantum Hamiltonian $\hat{H}$ give rise to an operator system satisfying the supersymmetry algebra $\{\hat{Q},\hat{Q}^{\dagger}\}=2\hat{H}$ which is one of the basic ingredients of \emph{local index theory}. There is a natural configuration space which one can associate to a filtered manifold $\lp M,\F\rp$ with admits a representation of supersymmetry, but in general we do not have a candidate supersymmetric action functional $S_{\F}$ or a method for quantizing such a theory. In \cite{BisF}, Bismut presents a direct relationship between the Bismut superconnection and the formal localization of path integrals on loop spaces associated to Riemannian fibrations indicating that the Chern-character of the Bismut superconnection $\bA$ represents the operator side of the formal fibered path integrals of a supersymmetric action functional $S$. As explained above, the ingredients of the Bismut superconnection $\bA$ are provided on an arbitrary filtered manifold $\lp M,\F\rp$ once a transverse distribution $\T$ is chosen, and one is led to dodge the question of quantization for the action functional $S_{\F}$ by working directly with the Bismut superconnection construction. We provide a few details on this line of reasoning and the difficulties surrounding the path integral for $S_{\F}$.
\subsection{Structure of the paper}
This article is organized as follows: In section 1, we give the construction of the asymptotic Bismut superconnection for a manifold $M$ equipped with two Euclidean subbundles $\lp\F,g_{\F}\rp$ and $\lp\T,g_{\T}\rp$ which are transverse $TM=\F\oplus\T$. We include the Weitzenbock formula for the $\F$ component of $\bD^E$ and discuss the Clifford contracted curvature in this new setting. We further include the $G$-equivariant setting and end with a global formula for $\bD^E$ in the contact subRiemannian setting. In section 2, we demonstrate the hypoellipticity of our operators using the methods of noncommutative harmonic analysis as developed in \cite{Van-Yun} and \cite{DH}. We give a few indications that the construction is sensitive to the choice of Clifford compatible connections $\nabla^{E,u}$ and inquire on the generality of the construction, proving that hypoellipticity is guaranteed in the two-step setting and provide a modification which always grants hypoellipticity on arbitrary filtered manifolds. We further investigate what lies within the kernel in the case of principal $\bS^1$-bundles over a closed manifold with non-vanishing curvature. In section 3 we give directions for further study mainly pointing out the relation to index theory and demonstrate that one can twist the Clifford relations in the $\T$-directions with matrix coefficients to obtain first order hypoelliptic operators with bundle coefficients and non-vanishing index. We end with a sigma model interpretation of the asymptotic Bismut superconnection.

\subsection*{Acknowledgments}
The first author is thankful to Xiang Tang and Yanli Song for the insightful discussions and support throughout this project as well as Stefan Haller and Magnus Goffeng for their helpful comments and suggestions. The first author is also grateful to Guoliang Yu for the generous hospitality and support during part of the writing of this paper. The first author was supported by NSF grants 1952551, 1952557, and 2247322.
\subsection*{Notation}
All vector spaces $V$ and algebras $A$ will be assumed to be $\bZ_2$-graded and all tensor products will be $\bZ_2$-graded tensor products. If a grading on a vector space $V$ is not given then we assign to $V$ the trivial grading $V^0=V$ and $V^1=\{\boldsymbol{0}\}$. If $A$ is a $\bZ_2$-graded algebra with $\bZ_2$-graded representation $\rho_A:A\rightarrow \End M$ and $B$ is a $\bZ_2$-graded algebra with $\bZ_2$-graded representation $\rho_B:B\rightarrow \End N$ then when tensoring the representations $\rho_A$ and $\rho_B$ we will take the $\bZ_2$-graded version of the tensoring
\[
    \rho_A\otimes\rho_B\lp a\otimes b\rp\lp m\otimes n\rp:=\lp-1\rp^{|m||b|}\rho_A\lp a\rp\lp m\rp\otimes \rho_B\lp n\rp
\]
which induces a representation $\rho_A\otimes\rho_B:A\otimes B\rightarrow\End M\otimes N$ which is also $\bZ_2$-graded.
\section{Construction of the Asymptotic Bismut Superconnection}
When one is working with a locally trivial fibration $\pi:M\rightarrow B$ equipped with an Ehresmann connection $\nabla$ and vertical fiber metric $g_{\F}$
\begin{displaymath}
  \begin{tikzcd}[column sep=2em]
   F \arrow{r}{ } \&  \lp M,g_{\F},\nabla\rp  \arrow{d}{\pi}\\
  \& B
  \end{tikzcd}
\end{displaymath}
one has the vertical distribution $\F:=\Ker\pi_*\leqslant TM$ and horizontal distribution $\T\leqslant TM$ determined by the Ehresmann connection $\nabla$. In \cite{Bis3} a smoothly varying family of Levi-Civita connections $\nabla^{\F}$ is constructed for $\lp\F,g_{\F}\rp$ and for an associated $\bZ_2$-graded covariant Hermitian Clifford module
\[
    \lp C\ell\lp\F^*,g^*_{\F}\rp,\nabla^{\F}\rp\circlearrowright\lp E,h_E,\nabla^E\rp
\]
the \textbf{\emph{Bismut superconnection}} is constructed as an odd first order differential operator $\bA$ acting on $\Gamma\lp E\otimes\wedge\T^*\rp$. Our aim in this section is to generalize the construction of $\bA$ to a smooth manifold $M$ equipped with an arbitrary subbundle $\F\leqslant TM$ and arbitrary complementary distribution $\T$ satisfying $TM=\F\oplus\T$. We let $n_1$ denote the rank of $\F$ and $n_2$ the rank of $\T$.
\subsection{The Asymptotic Bismut Superconnection}
We begin with a smooth manifold $M$ equipped with a smooth distribution $\F\leqslant TM$ and a Euclidean $\F$-fiber metric $g_{\F}$, and do not assume that $\F$ is bracket generating. We further equip $M$ with a Euclidean distribution $\lp\T,g_{\T}\rp$ which is transverse $TM=\F\oplus\T$. We consider the one-parameter family of Riemannian metrics
\[
    g_{M,u}:=g_{\F}\oplus u^{-1}g_{\T},\quad u>0
\]
and obtain two connections 
\[
    \nabla^{\F,u}:=P_{\F}\nabla^{g_{M,u}}P_{\F},\quad \nabla^{\T,u}:=P_{\T}\nabla^{g_{M,u}}P_{\T}.
\]
We define the tensors
\begin{align*}
    \iota_U\Omega_{\F}\lp Y,Z\rp&:=g_{\T}\lp U^{\T},[Z,Y]^{\T}\rp,\quad U\in\Gamma\lp TM\rp,\quad Y,Z\in\Gamma\lp\F\rp\\
    \iota_X\Omega_{\T}\lp U,V\rp&:=g_{\F}\lp X^{\F},[V,U]^{\F}\rp,\quad X\in\Gamma\lp TM\rp,\quad U,V\in\Gamma\lp\T\rp
\end{align*}
and observe that each of them can be considered as one-forms on $M$ with values in skew-adjoint operators by using the metrics $g_{\F}$ on $\F$ and $g_{\T}$ on $\T$
\[
    \so\Omega_{\F}\in\Omega^1\lp M,\so\lp\F,g_{\F}\rp\rp,\quad \so\Omega_{\T}\in\Omega^1\lp M,\so\lp\T,g_{\T}\rp\rp.
\]
We record the variation of the connections $\nabla^{\F,u}$ and $\nabla^{\T,u}$ below.
\begin{proposition}
    There exists $u$-independent connections $\nabla^{\F,0}$ and $\nabla^{\T,0}$ such that
    \[
        \nabla^{\F,u}=\nabla^{\F,0}+\frac{u^{-1}}{2}\so\Omega_{\F},\quad \nabla^{\T,u}=\nabla^{\T,0}+\frac{u}{2}\so\Omega_{\T}.
    \]
    Moreover, the connection $\nabla^{\F,0}$ only depends on $g_{\F}$ and the choice of splitting $TM=\F\oplus\T$. The same holds for $\nabla^{\T,0}$.
\end{proposition}
\begin{proof}
   First, observe that by Koszul we have for $X,Y,Z\in\Gamma\lp\F\rp$ that
\[
    g_{\F}\lp\nabla^{\F,u}_XY,Z\rp=g_{M,u}\lp\nabla^{g_{M,u}}_XY,Z\rp
\]
does \emph{not} vary with $u$.
If $T\in\Gamma\lp \T\rp$ and $Y,Z\in\Gamma\lp \F\rp$ then Koszul gives
\begin{align*}
    2g_{\F}\lp\nabla^{\F,u}_TY,Z\rp&=2g_{M,u}\lp\nabla^{g_{M,u}}_TY,Z\rp\\
    &=Tg_{M,u}\lp Y,Z\rp-g_{M,u}\lp Y,[T,Z]\rp-Zg_{M,u}\lp T,Y\rp+g_{M,u}\lp T,[Z,Y]\rp\\
    &\quad +Yg_{M,u}\lp Z,T\rp-g_{M,u}\lp Z,[Y,T]\rp\\
    &=Tg_{\F}\lp Y,Z\rp-g_{\F}\lp Y,[T,Z]^{\F}\rp+u^{-1}g_{\T}\lp T,[Z,Y]\rp\\
    &\quad -g_{\F}\lp Z,[Y,T]^{\F}\rp.
\end{align*}
Unraveling the definition of $\so\Omega_{\F}$ reveals that the connection
\[
    \nabla^{\F,0}=\nabla^{\F,u}-\frac{u^{-1}}{2}\so\Omega_{\F}
\]
is indeed $u$-independent. Moreover, we see from the Koszul identity that 
\begin{align*}
    2g_{\F}\lp\nabla^{\F,0}_XY,Z\rp&=Xg_{\F}\lp Y,Z\rp-g_{\F}\lp Y,[X,Z]\rp-Zg_{\F}\lp X,Y\rp+g_{\F}\lp X,[Z,Y]\rp\\
    &\quad +Yg_{\F}\lp Z,X\rp-g_{\F}\lp Z,[Y,X]\rp,\quad X,Y,Z\in\Gamma\lp\F\rp\\
    2g_{\F}\lp\nabla^{\F,0}_TY,Z\rp&=Tg_{\F}\lp Y,Z\rp-g_{\F}\lp Y,[T,Z]^{\F}\rp\\
    &\quad -g_{\F}\lp Z,[Y,T]^{\F}\rp,\quad T\in\Gamma\lp\T\rp,\quad X,Y\in\Gamma\lp\F\rp.
\end{align*}
Thus we see that $\nabla^{\F,0}$ only depends on $\lp\F,g_{\F}\rp$ and the choice of splitting $TM=\F\oplus\T$.
Similarly for $T,U,V\in\Gamma\lp\T\rp$, the quantity 
\[
    g_{\T}\lp \nabla^{\T,u}_TU,V\rp
\]
does not vary with $u$. If $X\in\Gamma\lp\F\rp$ and $U,V\in\Gamma\lp\T\rp$, then Koszul gives
\begin{align*}
    2u^{-1}g_{\T}\lp\nabla^{\T,u}_XU,V\rp&=2g_{M,u}\lp\nabla^{g_{M,u}}_XU,V\rp\\
    &=Xg_{M,u}\lp U,V\rp-g_{M,u}\lp U,[X,V]\rp-Vg_{M,u}\lp X,U\rp+g_{M,u}\lp X,[V,U]\rp\\
    &\quad +Ug_{M,u}\lp V,X\rp-g_{M,u}\lp V,[U,X]\rp\\
    &=u^{-1}Xg_{\T}\lp U,V\rp-u^{-1}g_{\T}\lp U,[X,V]^{\T}\rp+g_{\F}\lp X,[V,U]^{\F}\rp\\
    &\quad -u^{-1}g_{\T}\lp V,[U,X]^{\T}\rp.
\end{align*}
The definition of $\so\Omega_{\F}$ demonstrates the $u$-independence of the connection  
\[
    \nabla^{\T,0}:=\nabla^{\T,u}-\frac{u}{2}\so\Omega_{\T}.
\]
The Koszul identity gives us an explicit global formula for $\nabla^{\T,0}$
\begin{align*}
    2g_{\T}\lp\nabla^{\T,0}_TU,V\rp&=Tg_{\T}\lp U,V\rp-g_{\T}\lp U,[T,V]\rp-Vg_{\T}\lp T,U\rp+g_{\T}\lp T,[V,U]\rp\\
    &\quad +Ug_{\T}\lp V,T\rp-g_{\T}\lp V,[U,T]\rp,\quad T,U,V\in\Gamma\lp\T\rp\\
    2g_{\T}\lp\nabla^{\T,0}_XU,V\rp&=Xg_{\T}\lp U,V\rp-g_{\T}\lp U,[X,V]^{\T}\rp\\
    &\quad -g_{\T}\lp V,[U,X]^{\T}\rp,\quad X\in\Gamma\lp\F\rp,\quad U,V\in\Gamma\lp\T\rp.
\end{align*}
\end{proof}
 We now consider the shape one-form
\[
    \nabla^{g_{M,u}}=\nabla^{\F,u}\oplus\nabla^{\T,u}+S_u,\quad S_u\in\Omega^1\lp M,\so\lp TM,g_{M,u}\rp\rp
\]
and the induced one-form with values in two-forms
\[
    \iota_X\omega_u\lp Y,Z \rp:=g_{M,u}\lp \iota_XS_uY,Z\rp,\quad X,Y,Z\in\Gamma\lp TM\rp.
\]
The Koszul identity provides an explicit computation of $\omega_u$.
\begin{proposition}
\label{DSF}
    For $X,Y,Z\in\Gamma\lp\F\rp$ and $T,U,V\in\Gamma\lp\T\rp$ we have the identities
    \begin{align*}
        &\iota_X\omega_u\lp Y,Z\rp=0,\quad\iota_X\omega_u\lp U,V\rp=0,\quad \iota_X\omega_u\lp Y,V\rp=-\frac{1}{2}L_Vg_{\F}\lp X,Y\rp-\frac{u^{-1}}{2}\iota_V\Omega_{\F}\lp X,Y\rp\\
        &\iota_T\omega_u\lp U,V\rp=0\quad \iota_T\omega_u\lp Y,Z\rp=0,\quad\iota_T\omega_u\lp U,Z\rp=-\frac{u^{-1}}{2}L_Zg_{\T}\lp T,U\rp-\frac{1}{2}\iota_Z\Omega_{\T}\lp T,U\rp
    \end{align*}
\end{proposition}
The Euclidean fiber metric $g_{\F}$ on $\F$ allows one to construct the Clifford algebra bundle $C\ell\lp\F^*,g^*_{\F}\rp$. We further consider a $\bZ_2$-graded covariant Hermitian vector bundle $\lp E,h_E,\nabla^E\rp\rightarrow M$ equipped with a covariant Clifford action
\[
    \lp C\ell\lp \F^*,g^*_{\F}\rp,\nabla^{\F,0}\rp\circlearrowright_c\lp E,h_E,\nabla^E\rp
\]
If we define the one-parameter family of connection on $E$ by
\[
    \nabla^{E,u}:=\nabla^E+\frac{u^{-1}}{4}c\lp\Omega_{\F}\rp
\]
then we obtain a one-parameter family of covariant Clifford modules
\[
    \lp C\ell\lp \F^*,g^*_{\F}\rp,\nabla^{\F,u}\rp\circlearrowright_c\lp E,h_E,\nabla^{E,u}\rp.
\]
Note that the connections $\nabla^{E,u}$ remain $h_E$-compatible and preserve the $\bZ_2$-grading of $E$. The $u$-dependent Clifford algebra bundle $C\ell\lp\T^*,ug^*_{\T}\rp$ comes equipped with a natural $u$-dependent $\bZ_2$-graded Clifford bundle
\[
    \lp C\ell\lp\T^*,ug^*_{\T}\rp,\nabla^{\T,u}\rp\circlearrowright_{c^{\T}_u}\lp\wedge\T^*,\wedge ug^*_{\T},\nabla^{\T,u}\rp,
\]
where the $\bZ_2$-grading on $\wedge\T^*$ is given by the degree of forms mod $2$. The $u$-variation of the connections $\nabla^{\T,u}$ on $\wedge\T^*$ can be expressed as
\[
    \nabla^{\T,u}_X\alpha=\nabla^{\T,0}_X\alpha+\frac{u}{4}\text{ad}_{\iota_X\Omega_{\T}}\alpha,\quad X\in\Gamma\lp TM\rp,\quad \alpha\in\Gamma\lp\wedge\T^*\rp
\]
where $\text{ad}_{\iota_X\Omega_{\T}}$ denotes the commutator with $\iota_X\Omega_{\T}\in\Gamma\lp\wedge^2\T^*\rp\simeq\Gamma\lp C\ell^2\lp\T^*,g^*_{\T}\rp\rp$ in the Clifford algebra bundle $C\ell\lp\T^*,g^*_{\T}\rp$. The operator section for $\psi\in\T^*$ with $g_{\T}$ dual $\psi=g_{\T}\lp f,\cdot\rp$ is given at $u>0$ by
\[
    c^{\T}_u\lp \psi\rp:=\varepsilon^{\T}_{\psi}-u\iota^{\T}_{f}\circlearrowright\Gamma\lp\wedge\T^*\rp
\]
where $\varepsilon^{\T}$ is wedge and $\iota^{\T}$ is contraction.
Using the isomorphism of Clifford algebra bundles
\[
    C\ell\lp T^*M,g^*_{M,u}\rp\simeq  C\ell\lp\F^*,g^*_{\F}\rp\otimes C\ell\lp\T^*,ug^*_{\T}\rp
\]
we obtain a covariant $\bZ_2$-graded Clifford bundle
\[
    \lp C\ell\lp T^*M,g^*_{M,u}\rp,\nabla^{\F,u}\oplus \nabla^{\T,u}\rp\circlearrowright_{m_u}\lp E\otimes\wedge\T^*,h_E\otimes\wedge ug^*_{\T},\nabla^{E,u}\otimes I+I\otimes\nabla^{\T,u}\rp.
\]
Here $m_u:C\ell\lp T^*M,g^*_{M,u}\rp\rightarrow \End\lp E\otimes\wedge\T^*\rp$ is our notation for the Clifford action. When the shape $1$-form $S_u$ is non-zero, the connection $\nabla^{\F,u}\oplus\nabla^{\T,u}$ is not the Levi-Civita connection for the metric $g_{M,u}$. Bismut provides a remedy to this situation in $\cite{BisF}$ by perturbing the Clifford module connection by
\[
    \nabla^{\bE,u}:=\nabla^{E,u}\otimes I+I\otimes\nabla^{\T,u}+\frac{1}{2}m_u\omega_u.
\]  
The \textbf{\emph{Asymptotic Bismut Superconnection}} is defined as the finite part of the limit of the one-parameter family of differential operators
\[
    \bD^E:=\lim^{F.P}_{u\rightarrow 0^+}m_u\circ\nabla^{\bE,u}.
\]
\begin{remark}
    In the case of $\lp\F,g_{\F}\rp\rightarrow M$ is oriented and spin, if $\lp\Sc_{\F},h_{\Sc}\rp\rightarrow M$ is a choice of spinor bundle and $\nabla^{\Sc,0}$ is the spin Levi-Civita connection corresponding to $\nabla^{\F,0}$ then
    \[
        \nabla^{\Sc,u}=\nabla^{\Sc,0}+\frac{u^{-1}}{4}c\lp\Omega_{\F}\rp
    \]
    is the spin Levi-Civita connection corresponding to $\nabla^{\F,u}$. This explains our choice of adding $\frac{u^{-1}}{4}c\lp\Omega_{\F}\rp$ to a fixed $\nabla^E$.\par\qed
\end{remark}
Given our global formulas obtained above for $\omega_u$ we can write down a global formula for $\bD^E$. If $\{ e_i\}_{i=1}^{n_1}$ is a local orthonormal framing of $\lp\F,g_{\F}\rp$ with dual framing $\{\varphi^i\}_{i=1}^{n_1}$, and $\{f_{\mu}\}_{\mu=1}^{n_2}$ is a local orthonormal of $\lp\T,g_{\T}\rp$ with dual framing $\{\psi^{\mu}\}_{\mu=1}^{n_2}$ we can define the following operators
\begin{align*}
    D^{\F,E\otimes\wedge\T^*}&:=c\lp\varphi^i\rp\lp\nabla^E_{e_i}\otimes I+I\otimes\nabla^{\T,0}_{e_i}\rp\\
    \varepsilon^{\T}\circ \tr Lg_{\F}&:=\delta^{ij}\varepsilon^{\T}_{\psi^{\mu}}L_{f_{\mu}}g_{\F}\lp e_i,e_j\rp=-c\lp\varphi^i\rp c\lp\varphi^j\rp\varepsilon^{\T}_{\psi^{\mu}}L_{f_{\mu}}g_{\F}\lp e_i,e_j\rp\\
    \iota^{\T}\circ c\lp\Omega_{\F}\rp&:=\frac{1}{2}c\lp\varphi^i\rp c\lp\varphi^j\rp\iota^{\T}_{f_{\mu}}\iota_{f_{\mu}}\Omega_{\F}\lp e_i,e_j\rp\\
    d^{\T}_{\nabla^E}&:=\varepsilon^{\T}_{\psi^{\mu}}\lp\nabla^E_{f_{\mu}}\otimes I+I\otimes\nabla^{\T,0}_{f_{\mu}}\rp\\
    c\circ \varepsilon^{\T}_{\Omega_{\T}}&:=\frac{1}{2}c\lp\varphi^i\rp\varepsilon^{\T}_{\psi^{\mu}}\varepsilon^{\T}_{\psi^{\nu}}\iota_{e_i}\Omega_{\T}\lp f_{\mu},f_{\nu}\rp\\
    c\circ \tr Lg_{\T}&:=\delta^{\mu\nu}c\lp\varphi^i\rp L_{e_i}g_{\T}\lp f_{\mu},f_{\nu}\rp=-\lp \varepsilon^{\T}_{\psi^{\mu}}c\lp\varphi^i\rp\iota^{\T}_{f_{\nu}}+\iota^{\T}_{f_{\mu}}c\lp\varphi^i\rp\varepsilon^{\T}_{\psi^{\nu}}\rp L_{e_i}g_{\T}\lp f_{\mu},f_{\nu}\rp
\end{align*}
which one can check are independent of the local choices of local orthonormal framings. We then have the following global formula for the asymptotic Bismut super connection.
\begin{proposition}
\label{GFD}
    For a given $\bZ_2$-graaded covariant Hermitian Clifford module 
    \[
        \lp C\ell\lp\F^*,g^*_{\F}\rp,\nabla^{\F,0}\rp\circlearrowright_c\lp E,h_E,\nabla^E\rp
    \]
    the asymptotic Bismut superconnection is given globally as
    \[
        \bD^E=D^{\F,E\otimes\wedge\T^*}+d^{\T}_{\nabla^E}+\frac{1}{4}\iota^{\T}\circ c\lp\Omega_{\F}\rp+\frac{1}{4}\varepsilon^{\T}\circ\tr Lg_{\F}-\frac{1}{2}c\circ\varepsilon^{\T}_{\Omega_{\T}}+\frac{1}{4}c\circ\tr Lg_{\T}
    \]
\end{proposition}
\begin{proof}
    We first split the Clifford contraction into the $\F$ and $\T$ components
\[
    m_u\circ\nabla^{\bE,u}=m_u\big|_{\F}\circ\nabla^{\bE,u}+m_u\big|_{\T}\circ\nabla^{\bE,u}
\]
If $\{ e_i\}_{i=1}^{n_1}$ is a local orthonormal framing of $\lp\F,g_{\F}\rp$ with dual framing $\{\varphi^i\}_{i=1}^{n_1}$, and $\{f_{\mu}\}_{\mu=1}^{n_2}$ is a local orthonormal of $\lp\T,g_{\T}\rp$ with dual framing $\{\psi^{\mu}\}_{\mu=1}^{n_2}$ then the $\F$-contraction locally takes the form
\begin{align*}
    &m_u\big|_{\F}\circ\nabla^{\bE,u}=c\lp\varphi^i\rp\nabla^{\bE,u}_{e_i}\\
    &=c\lp\varphi^i\rp\lp\nabla^E_{e_i}\otimes I+I\otimes\nabla^{\T,0}_{e_i}\rp+\frac{u^{-1}}{4}c\lp\varphi^i\rp\iota_{e_i}c\lp\Omega_{\F}\rp+\frac{u}{4}c\lp\varphi^i\rp\text{ad}_{\iota_{e_i}\Omega_{\T}}+\frac{1}{2}c\lp\varphi^i\rp m_u\lp\iota_{e_i}\omega_u\rp\\
    &=c\lp\varphi^i\rp\lp\nabla^E_{e_i}\otimes I+I\otimes\nabla^{\T,0}_{e_i}\rp+\frac{u}{4}c\lp\varphi^i\rp\text{ad}_{\iota_{e_i}\Omega_{\T}}\\
    &\hspace{5cm}+\frac{1}{2}c\lp\varphi^i\rp c\lp\varphi^j\rp\lp\varepsilon^{\T}_{\psi^{\mu}}-u\iota^{\T}_{f_{\mu}}\rp\lp -\frac{1}{2}L_{f_{\mu}}g_{\F}\lp e_i,e_j\rp -\frac{u^{-1}}{2}\iota_{f_{\mu}}\Omega_{\F}\lp e_i,e_j\rp\rp
\end{align*}
where in the second to third equality we used that $\iota_X\Omega_{\F}=0$ if $X\in\Gamma\lp\F\rp$. In the limit as $u\rightarrow 0^+$ the finite part gives
\begin{align*}
    &\lim^{F.P.}_{u\rightarrow 0^+}m_u\big|_{\F}\circ\nabla^{\bE,u}=\\
    &\quad c\lp\varphi^i\rp\lp\nabla^E_{e_i}\otimes I+I\otimes\nabla^{\T,0}_{e_i}\rp-\frac{1}{4}c\lp\varphi^i\rp c\lp\varphi^j\rp\varepsilon^{\T}_{\psi^{\mu}}L_{f_{\mu}}g_{\F}\lp e_i,e_j\rp+\frac{1}{4}c\lp\varphi^i\rp c\lp\varphi^j\rp\iota^{\T}_{f_{\mu}}\iota_{f_{\mu}}\Omega_{\F}\lp e_i,e_j\rp
\end{align*}
giving the global formula
\[
    \lim^{F.P.}_{u\rightarrow 0^+}m_u\big|_{\F}\circ\nabla^{\bE,u}=D^{\F,E\otimes\wedge\T^*}+\frac{1}{4}\varepsilon^{\T}\circ \tr Lg_{\F}+\frac{1}{2}\iota^{\T}\circ c\lp\Omega_{\F}\rp
\]
Similarly, the $\T$-contraction is locally of the form
\begin{align*}
    m_u\big|_{\T}\nabla^{\bE,u}&=\lp\varepsilon^{\T}_{\psi^{\mu}}-u\iota^{\T}_{f_{\mu}}\rp\lp\nabla^E_{f_{\mu}}\otimes I+I\otimes\nabla^{\T,0}_{f_{\mu}}\rp+\frac{u^{-1}}{4}\lp\varepsilon^{\T}_{\psi^{\mu}}-u\iota^{\T}_{f_{\mu}}\rp \iota_{f_{\mu}}c\lp\Omega_{\F}\rp\\&\hspace{3cm}+\frac{u}{4}\lp\varepsilon^{\T}_{\psi^{\mu}}-u\iota^{\T}_{f_{\mu}}\rp\iota_{f_{\mu}}\text{ad}\Omega_{\T}
    +\frac{1}{2}\lp\varepsilon^{\T}_{\psi^{\mu}}-u\iota^{\T}_{f_{\mu}}\rp m_u\lp\iota_{f_{\mu}}\omega_u\rp\\
    &=\lp\varepsilon^{\T}_{\psi^{\mu}}-u\iota^{\T}_{f_{\mu}}\rp\lp\nabla^E_{f_{\mu}}\otimes I+I\otimes\nabla^{\T,0}_{f_{\mu}}\rp+\frac{u^{-1}}{4}\lp\varepsilon^{\T}_{\psi^{\mu}}-u\iota^{\T}_{f_{\mu}}\rp \iota_{f_{\mu}}c\lp\Omega_{\F}\rp\\
    &\hspace{2cm}\frac{1}{2}\lp\varepsilon^{\T}_{\psi^{\mu}}-u\iota^{\T}_{f_{\mu}}\rp c\lp\varphi^i\rp\lp\varepsilon^{\T}_{\psi^{\nu}}-u\iota^{\T}_{f_{\nu}}\rp \lp \frac{u^{-1}}{2} L_{e_i}g_{\T}\lp f_{\mu},f_{\nu}\rp+\frac{1}{2}\iota_{e_i}\Omega_{\T}\lp f_{\mu},f_{\nu}\rp\rp
\end{align*}
which implies that the finite part as $u\rightarrow 0^+$ is given as
\begin{align*}
    \lim^{F.P.}_{u\rightarrow0^+}&m_u\big|_{\T}\circ\nabla^{\bE,u}=\varepsilon^{\T}_{\psi^{\mu}}\lp\nabla^E_{f_{\mu}}\otimes I+I\otimes\nabla^{\T,0}_{f_{\mu}}\rp-\frac{1}{4}\iota^{\T}_{f_{\mu}}c\lp\iota_{f_{\mu}}\Omega_{\F}\rp\\
    &+\frac{1}{2}\varepsilon^{\T}_{\psi^{\mu}}c\lp\varphi^i\rp\lp\varepsilon^{\T}_{\psi^{\nu}}\frac{1}{2}\iota_{e_i}\Omega_{\T}\lp f_{\mu},f_{\nu}\rp-\frac{1}{2}\iota^{\T}_{f_{\nu}}L_{e_i}g_{\T}\lp f_{\mu},f_{\nu}\rp\rp-\frac{1}{2}\iota^{\T}_{f_{\mu}}c\lp\varphi^i\rp\varepsilon^{\T}_{\psi^{\nu}}\frac{1}{2}L_{e_i}g_{\T}\lp f_{\mu},f_{\nu}\rp
\end{align*}
thus giving the global formula
\[
    \lim^{F.P.}_{u\rightarrow 0^+}m_u\big|_{\T}\circ\nabla^{\bE,u}=d^{\T}_{\nabla^E}-\frac{1}{4}\iota^{\T}\circ c\lp\Omega_{\F}\rp-\frac{1}{2}c\circ\varepsilon^{\T}_{\Omega_{\T}}+\frac{1}{4}c\circ\tr Lg_{\T}
\]
\end{proof}
When the distribution $\F$ is the vertical distribution of a locally trivial fibration equipped with Ehresmann connection $\nabla$ and vertical fiber metric $g_{\F}$
\begin{displaymath}
  \begin{tikzcd}[column sep=2em]
   F \arrow{r}{ } \&  \lp M,g_{\F},\nabla\rp  \arrow{d}{\pi}\\
  \& B
  \end{tikzcd}
\end{displaymath}
one can take as transverse distribution $\T=\Ker\nabla\simeq\pi^*TB$ and transverse metric $g_{\T}=\pi^*g_B$ for some auxiliary Riemannian metric $g_B$ on $B$. In this case the distribution $\F$ is integrable, thus $\Omega_{\F}\equiv0$. The distribution $\T$ will not be integrable in general, however one does have that 
\[
    \nabla^{\T,0}=\pi^*\nabla^{g_B}
\]
where $\nabla^{g_B}$ is the Levi-Civita connection of $g_B$. We now compare our definition of asymptotic Bismut superconnection with the original definition of Bismut.
\begin{proposition}
     Let $\pi:\lp M,g_{\F},\nabla\rp\rightarrow B$ be a locally trivial fibration with Ehresmann connection $\nabla$ and vertical fiber metric $g_{\F}$, and let $g_B$ be a choice of Riemannian metric for $B$. Then for $\lp\F,g_{\F}\rp=\lp\Ker\pi_*,g_{\F}\rp$, $\lp\T,g_{\T}\rp=\lp\Ker\nabla,\pi^*g_B\rp$, and a $\bZ_2$-graded covariant Hermitian Clifford module
     \[
        \lp C\ell\lp\F^*,g^*_{\F}\rp,\nabla^{\F,0}\rp\circlearrowright\lp E,h_E,\nabla^E\rp
     \]
     then
     \[
        \bD^E=\bA^E-\frac{1}{4}c\circ\varepsilon^{\T}_{\Omega_{\T}}
     \]
     where $\bA^E$ is the Bismut supperconnection.
\end{proposition}
\begin{proof}
    Given that $\nabla^{\T,0}=\pi^*\nabla^{g_B}$, we have that
    \[
        \nabla^{\T,u}=\pi^*\nabla^{g_B}+\frac{u}{2}\so\Omega_{\T}.
    \]
    If $\omega^{\nabla^{\T,u}}_u$ denotes the dual shape-form for $\nabla^{\F,u}\oplus\nabla^{\T,u}$ and $\omega^{\pi^*\nabla^{g_B}}_u$ denotes the dual shape-form for $\nabla^{\F,u}\oplus\pi^*\nabla^{g_B}$ then we have the relation
    \[
        \omega^{\nabla^{\T,u}}_u=\omega^{\pi^*\nabla^{g_B}}_u-\frac{1}{2}\Omega_{\T}
    \]
    Because $\Omega_{\F}\equiv0$, there is no modification to the connection $\nabla^E$. Thus the asymptotic Bismut connection is
    \[
        \bD^E=\lim^{F.P.}_{u\rightarrow 0^+}m_u\circ\lp\nabla^E\otimes\lp\pi^*\nabla^{g_B}+\frac{u}{4}\ad_{\Omega_{\T}}\rp+\frac{1}{2}m_u\omega^{\nabla^{\T,u}}_u\rp
    \]
    and the same computation as in \ref{GFD} show that $\ad_{\Omega_{\T}}$ does not contribute in the limit as $u\rightarrow 0^+$. Thus we have
    \begin{align*}
        \bD^E&=\lim^{F.P.}_{u\rightarrow0^+}m_u\circ\lp\nabla^E\otimes\pi^*\nabla^{g_B}+\frac{1}{2}m_u\omega^{\pi^*\nabla^{g_B}}_u-\frac{1}{4}m_u\Omega_{\T}\rp\\
        &=\bA^E-\frac{1}{4}c\circ\varepsilon^{\T}_{\Omega_{\T}}
    \end{align*}
\end{proof}
In our analysis above, if we used $\nabla^{\T,0}$ instead of $\nabla^{\T,u}$ then the definition of $\bD^E$ would agree with Bismut's definition on the nose in the situation of covariant Riemnannian fibrations.\par
In general, the simplest examples of covariant Clifford modules one can cook up come from the usual recipes in Riemannian geometry. Namely when $\lp\F,g_{\F}\rp$ is oriented and admits a spin structure one can take $\lp\Sc_{\F},h_{\Sc},\nabla^{\Sc,0}\rp$ where $\nabla^{\Sc,0}$ is the spin Levi-Civita connection corresponding to $\nabla^{\F,0}$. When $\F$ admits an almost complex structure $J_{\F}$ compatible with $g_{\F}$ and $\nabla^{\F,0}$ one can use $\lp\wedge^{0,*}\F^*,\wedge^{0,*}g_{\F},\nabla^{\F,0}\rp$. And if no additional structure on $\lp\F,g_{\F}\rp$ is available one can always use $\lp \wedge\F^*,\wedge g^*_{\F},\nabla^{\F,0}\rp$. Moreover, once one is given a fixed covariant Clifford module, then the one can apply the operation of twisting. Suppose that $\lp B,h_B,\nabla^B\rp\rightarrow M$ is a $\bZ_2$-graded covariant Hermitian vector bundle. Then for a given $\bZ_2$-graded covariant Hermitian Clifford module
\[
    \lp C\ell\lp\F^*,g^*_{\F}\rp,\nabla^{\F,0}\rp\circlearrowright\lp E,h_E,\nabla^E\rp
\]
we can twist the bundle by the covariant data and obtain
\begin{align*}
    \lp C\ell\lp\F^*,g^*_{\F}\rp,\nabla^{\F,0}\rp\circlearrowright\lp E\otimes B,h_{E}\otimes h_B,\nabla^E\otimes I+I\otimes \nabla^B\rp
\end{align*}
and obtain twisted versions of the asymptotic Bismut superconnections, denoted $\bD^{E\otimes B}$. This operation of twisting provides plenty of examples to work with in this new setting.\par
For completeness we include the Weitzenbock formula for the horizontal Dirac operator $D^{\F,E\otimes\wedge\T^*}$. For simplicity of exposition we will omit the $\wedge\T^*$-variables on $E$ and let $\Delta^{\F,E}$ denote the square of $D^{\F,E}$. For $\{e_i\}_{i=1}^{n_1}$ a local orthonormal framing of $\lp\F,g_{\F}\rp$ with dual framing $\{\varphi^i\}_{i=1}^{n_1}$ we define the operators
\begin{align*}
    c\circ\nabla^{E}_{\widetilde{\Omega}_{\F}}&:=\frac{1}{2}c\lp\varphi^i\rp c\lp\varphi^j\rp\nabla^E_{\widetilde{\Omega}_{\F}\lp e_i,e_j\rp},\quad \widetilde{\Omega}_{\F}\lp e_e,e_j\rp:=[e_j,e_i]^{\T}\\
    c\circ K^E&:=\frac{1}{2}c\lp\varphi^i\rp c\lp\varphi^j\rp K^E\lp e_i,e_j\rp,\quad K^E:=\lp\nabla^E\rp^2
\end{align*}
which one verifies are independent of the choice of local orthonormal framing and thus define global differential operators.
\begin{proposition}
\label{WF}
    The Weitzenbock formula for $\Delta^{\F,E}$ is given by
    \[
        \Delta^{\F,E}=\lp\nabla^E\Big|_{\F}\rp^*\nabla^E\Big|_{\F}-c\circ\nabla^E_{\widetilde{\Omega}_{\F}}+c\circ K^E.
    \] 
\end{proposition}
\begin{proof}
    If $\{e_i\}_{i=1}^{n_1}$ is a local orthonormal framing of $\lp\F,g_{\F}\rp$ with dual framing $\{\varphi^i\}_{i=1}^{n_1}$ then we have
\begin{align*}
    \Delta^{\F,E}&=c\lp \varphi^i\rp\nabla^E_{e_i}c\lp \varphi^j\rp\nabla^E_{e_j}\\
    &=-\delta^{ij}\nabla^{E}_{e_i}\nabla^E_{e_j}+c\lp\varphi^i\rp c\lp\nabla^{\F,0}_{e_i}\varphi^j\rp\nabla^E_{e_j}+c\lp\varphi^i\rp c\lp\varphi^j\rp_{i<j} K^E\lp e_i,e_j\rp+c\lp \varphi^i\rp c\lp \varphi^j\rp_{i<j}\nabla^E_{[e_i,e_j]}.
\end{align*}
The curvature term can be expressed as a Clifford contraction
\[
    c\circ K^E=\frac{1}{2}c\lp\varphi^i\rp c\lp\varphi^j\rp K^E\lp e_i,e_j\rp.
\]
Next, we see that
\[
    c\lp\varphi^i\rp c\lp\nabla^{\F,0}_{e_i}\varphi^j\rp\nabla^E_{e_j}=-c\lp\varphi^i\rp c\lp\varphi^k\rp_{i<k}\varphi^j\lp\nabla^{\F,0}_{e_i}e_k-\nabla^{\F,0}_{e_k}e_i\rp\nabla^E_{e_j}+\varphi^j\lp\nabla^{\F,0}_{e_i}e_i\rp\nabla^E_{e_j}
\]
Using the restricted torsion free identity
\[
    \nabla^{\F,0}_{e_i}e_k-\nabla^{\F,0}_{e_k}e_i=[e_i,e_k]^{\F}
\]
as well as the Lie bracket identity
\[
    [e_i,e_j]=[e_i,e_j]^{\F}-\widetilde{\Omega}_{\F}\lp e_i,e_j\rp
\]
we see that
\[
    c\lp\varphi^i\rp c\lp\nabla^{\F,0}_{e_i}\varphi^j\rp\nabla^E_{e_j}+c\lp\varphi^i\rp c\lp\varphi^j\rp_{i<j}\nabla^E_{[e_i,e_j]}=\delta^{ij}\nabla^E_{\nabla^{\F,0}_{e_i}e_j}-c\circ\nabla^E_{\widetilde{\Omega}_{\F}}
\]
Thus, all together we have
\[
    \Delta^{\F,E}=-\delta^{ij}\lp \nabla^{E}_{e_i}\nabla^{E}_{e_j}-\nabla^{E}_{\nabla^{\F,0}_{e_i}e_j}\rp-c\circ\nabla^E_{\widetilde{\Omega}_{\F}}+c\circ K^E
\]
The horizontal covariant derivative 
\begin{align*}
    \nabla^E\Big|_{\F}:\Gamma\lp E\rp&\rightarrow \Gamma\lp\F^*\otimes E\rp\\
    s&\mapsto \nabla^Es\Big|_{\F}
\end{align*}
can be composed with it's formal adjoint (here the adjoint relation is defined using the Hermitian inner product of $E$ and the Riemannian metric $g_M=g_{\F}\oplus g_{\T}$) and is given locally as
\[
    \lp\nabla^E\Big|_{\F}\rp^*\nabla^E\Big|_{\F}=-\delta^{ij}\lp\nabla^E_{e_i}\nabla^E_{e_j}-\nabla^E_{\nabla^{\F,0}_{e_i}e_j}\rp
\]
Thus we have the global Weitzenbock formula
\[
    \Delta^{\F,E}=\lp\nabla^E\Big|_{\F}\rp^*\nabla^E\Big|_{\F}-c\circ\nabla^E_{\widetilde{\Omega}_{\F}}+c\circ K^E
\]
\end{proof}

\subsection{The Clifford Contracted Curvature}
Let $R^{\F,0}\in\Omega^2\lp M,\so\lp\F,g_{\F}\rp\rp$ denote the curvature of the connection $\nabla^{\F,0}$ and let 
\[
    c\lp R^{\F,0}\rp:=\tau^{-1}R^{\F,0}\in\Omega^2\lp M,C\ell^2\lp\F,g^*_{\F}\rp\rp
\]
where $\tau:C\ell^2\lp\F^*,g^*_{\F}\rp\xrightarrow{\sim}\so\lp\F,g_{\F}\rp$ is the usual isomorphism of Lie algebras. If $\{ e_i\}_{i=1}^{n_1}$ is a local orthonormal framing of $\lp\F,g_{\F}\rp$ with dual framing $\{\varphi^i\}_{i=1}^{n_1}$ then locally one has
\[
    c\lp R^{\F,0}\rp=\frac{1}{4}g_{\F}\lp R^{\F,0}e_i,e_j\rp c\lp\varphi^i\rp c\lp\varphi^j\rp
\]
Observe then that there exists a unique $F^{E/S}\in\Omega^2\lp M,\End_{C\ell}\lp E\rp\rp$ for which
\[
    K^E=c\lp R^{\F,0}\rp+F^{E/S}.
\]
We are interested in computing the Clifford contraction of $c\lp R^{\F,0}\rp$. It is classical calculation that when $\lp\F,g_{\F}\rp=\lp TM,g_M\rp$, then the Clifford contracted curvature is one quarter of the scalar curvature $k$ of $g$. To this end, we wish to evaluate
\[
    \frac{1}{2}c\lp\varphi^k\rp c\lp\varphi^l\rp c\lp R^{\F,0}\rp\lp e_k,e_l\rp =\frac{1}{8}c\lp \varphi^k\rp c\lp\varphi^l\rp c\lp\varphi^i\rp c\lp\varphi^j\rp g_{\F}\lp R^{\F,0}\lp e_k,e_l\rp e_i,e_j\rp
\]
The fact that $\nabla^{\F,0}$ is $g_{\F}$-compatible implies that the $\F$-tensor
\[
    Rm^{\F,0}\lp X,Y,Z,W\rp:=g_{\F}\lp R^{\F,0}\lp X,Y\rp Z,W\rp,\quad W,X,Y,Z\in\Gamma\lp\F\rp
\]
defines an operator section $\R^{\F,0}\in\Gamma\lp\End\wedge^2\F^*\rp$. Unfortunately this operator section may no longer be symmetric, the lack of symmetry can be traced back to the failure of the Bianchi identity 
\[
    R^{\F,0}\lp X,Y\rp Z+R^{\F,0}\lp Z,X\rp Y+R^{\F,0}\lp Y,Z\rp X=0,\quad \forall X,Y,Z\in\Gamma\lp\F\rp.
\]
This implies that the usual cancellations that appear for the Clifford contractions of the Riemann curvature tensor will not occur for the tensor $Rm^{\F,0}$. Thus we are tasked with understanding the deviation of the algebraic Bianchi identity
\[
    g_{\F}\lp R^{\F,0}\lp X,Y\rp Z+R^{\F,0}\lp Z,X\rp Y+R^{\F,0}\lp Y,Z\rp X,W\rp,\quad X,Y,Z,W\in\Gamma\lp\F\rp
\]
for the connection $\nabla^{\F,0}$.
\begin{proposition}
\label{ABia}
    For $X,Y,Z,W\in\Gamma\lp\F\rp$ one has
    \begin{align*}
    g_{\F}&\lp  R^{\F,0}\lp X,Y\rp Z+R^{\F,0}\lp Z,X\rp Y+R^{\F,0}\lp Y,Z\rp X,W\rp=\\
    &\hspace{3cm}\frac{1}{2}\lp L_{\widetilde{\Omega}_{\F}\lp Y,Z\rp}g_{\F}\lp X,W\rp+L_{\widetilde{\Omega}_{\F}\lp X,Y\rp}g_{\F}\lp Z,W\rp+L_{\widetilde{\Omega}_{\F}\lp Z,X\rp}g_{\F}\lp Y,W\rp\rp
\end{align*}
\end{proposition}
\begin{proof}
    Using the constrained torsion identity one evaluates that the above sum is
\begin{align*}
    R^{\F,0}\lp X,Y\rp Z+R^{\F,0}\lp Z,X\rp Y+R^{\F,0}\lp Y,Z\rp X&=\nabla^{\F,0}_X[Y,Z]^{\F}-\nabla^{\F,0}_{[Y,Z]}X+\nabla^{\F,0}_Z[X,Y]^{\F}-\nabla^{\F,0}_{[X,Y]}Z\\
    &\quad+\nabla^{\F,0}_Y[Z,X]^{\F}-\nabla^{\F,0}_{[Z,X]}Y
\end{align*}
We can rewrite the first difference in the above sum using the $u$-varying connections $\nabla^{g_{M,u}}$ as
\begin{align*}
    &\nabla^{\F,0}_X[Y,Z]^{\F}-\nabla^{\F,0}_{[Y,Z]}X=P^{\F}\lp\nabla^{g_{M,u}}_X[Y,Z]^{\F}-\nabla^{g_{M,u}}_{[Y,Z]}X\rp+\frac{u^{-1}}{2}\iota_X\so\Omega_{\F}[Y,Z]^{\F}-\frac{u^{-1}}{2}\iota_{[Y,Z]}\so\Omega_{\F}X\\
    &\hspace{3cm}=[X,[Y,Z]]^{\F}-P^{\F}\nabla^{g_{M,u}}_X[Y,Z]^{\T}+\frac{u^{-1}}{2}\iota_X\so\Omega_{\F}[Y,Z]^{\F}-\frac{u^{-1}}{2}\iota_{[Y,Z]}\so\Omega_{\F}X
\end{align*}
We obtain then
\begin{align*}
    &R^{\F,0}\lp X,Y\rp Z+R^{\F,0}\lp Z,X\rp Y+R^{\F,0}\lp Y,Z\rp X=\\
    &\quad[X,[Y,Z]]^{\F}-P^{\F}\nabla^{g_{M,u}}_X[Y,Z]^{\T}+\frac{u^{-1}}{2}\iota_X\so\Omega_{\F}[Y,Z]^{\F}-\frac{u^{-1}}{2}\iota_{[Y,Z]}\so\Omega_{\F}X\\
    &\quad+[Z,[X,Y]]^{\F}-P^{\F}\nabla^{g_{M,u}}_Z[X,Y]^{\T}+\frac{u^{-1}}{2}\iota_Z\so\Omega_{\F}[X,Y]^{\F}-\frac{u^{-1}}{2}\iota_{[X,Y]}\so\Omega_{\F}Z\\
    &\quad\quad +[Y,[Z,X]]^{\F}-P^{\F}\nabla^{g_{M,u}}_Y[Z,X]^{\T}+\frac{u^{-1}}{2}\iota_Y\so\Omega_{\F}[Z,X]^{\F}-\frac{u^{-1}}{2}\iota_{[Z,X]}\so\Omega_{\F}Y
\end{align*}
We see that the sum involving iterated Lie brackets vanishes due to the Jacobi identity. Moreover, when taking the inner product with $W$, the result must be independent of $u$. Using Proposition \ref{GET}, we find that 
\begin{align*}
    g_{\F}&\lp  R^{\F,0}\lp X,Y\rp Z+R^{\F,0}\lp Z,X\rp Y+R^{\F,0}\lp Y,Z\rp X,W\rp=\\
    &\hspace{4cm}\frac{1}{2}\lp L_{\widetilde{\Omega}_{\F}\lp Y,Z\rp}g_{\F}\lp X,W\rp+L_{\widetilde{\Omega}_{\F}\lp X,Y\rp}g_{\F}\lp Z,W\rp+L_{\widetilde{\Omega}_{\F}\lp Z,X\rp}g_{\F}\lp Y,W\rp\rp
\end{align*}
\end{proof}
For a fixed $j$ we consider the following sum
\[
    c\lp \varphi^k\rp c\lp\varphi^l\rp c\lp\varphi^i\rp  g_{\F}\lp R^{\F,0}\lp e_k,e_l\rp e_i,e_j\rp
\]
and note that the sum over all $k,l,i$ takes the form
\begin{align*}
    &c\lp \varphi^k\rp c\lp\varphi^l\rp c\lp\varphi^i\rp  g_{\F}\lp R^{\F,0}\lp e_k,e_l\rp e_i,e_j\rp=\\
    &c\lp \varphi^k\rp c\lp\varphi^l\rp c\lp\varphi^i\rp ^{k<l<i} \lp g_{\F}\lp R^{\F,0}\lp e_k,e_l\rp e_i,e_j\rp+ g_{\F}\lp R^{\F,0}\lp e_i,e_k\rp e_l,e_j\rp+g_{\F}\lp R^{\F,0}\lp e_l,e_i\rp e_k,e_j\rp\rp-\\
    &c\lp \varphi^k\rp c\lp\varphi^l\rp c\lp\varphi^i\rp^{k<l<i} \lp g_{\F}\lp R^{\F,0}\lp e_l,e_k\rp e_i,e_j\rp+ g_{\F}\lp R^{\F,0}\lp e_i,e_l\rp e_k,e_j\rp+g_{\F}\lp R^{\F,0}\lp e_k,e_i\rp e_l,e_j\rp\rp+\\
    &\delta_{li}c\lp \varphi^k\rp c\lp\varphi^l\rp c\lp\varphi^i\rp g_{\F}\lp R^{\F,0}\lp e_k,e_l\rp e_i,e_j\rp+
    \delta_{ki}c\lp \varphi^k\rp c\lp\varphi^l\rp c\lp\varphi^i\rp g_{\F}\lp R^{\F,0}\lp e_k,e_l\rp e_i,e_j\rp
\end{align*}
Using (\ref{ABia}) we have

\begin{align*}
    &c\lp \varphi^k\rp c\lp\varphi^l\rp c\lp\varphi^i\rp ^{k<l<i} \lp g_{\F}\lp R^{\F,0}\lp e_k,e_l\rp e_i,e_j\rp+ g_{\F}\lp R^{\F,0}\lp e_i,e_k\rp e_l,e_j\rp+g_{\F}\lp R^{\F,0}\lp e_l,e_i\rp e_k,e_j\rp\rp-\\
    &c\lp \varphi^k\rp c\lp\varphi^l\rp c\lp\varphi^i\rp ^{k<l<i} \lp g_{\F}\lp R^{\F,0}\lp e_l,e_k\rp e_i,e_j\rp+ g_{\F}\lp R^{\F,0}\lp e_i,e_l\rp e_k,e_j\rp+g_{\F}\lp R^{\F,0}\lp e_k,e_i\rp e_l,e_j\rp\rp\\
    &=\frac{1}{2}c\lp \varphi^k\rp c\lp\varphi^l\rp c\lp\varphi^i\rp^{k<l<i}\lp L_{\widetilde{\Omega}_{\F}\lp e_l,e_i\rp}g_{\F}\lp e_k,e_j\rp+L_{\widetilde{\Omega}_{\F}\lp e_k,e_l\rp}g_{\F}\lp e_i,e_j\rp+L_{\widetilde{\Omega}_{\F}\lp e_i,e_k\rp}g_{\F}\lp e_l,e_j\rp\rp-\\
    &\frac{1}{2}c\lp \varphi^k\rp c\lp\varphi^l\rp c\lp\varphi^i\rp ^{k<l<i}\lp L_{\widetilde{\Omega}_{\F}\lp e_k,e_i\rp}g_{\F}\lp e_l,e_j\rp+L_{\widetilde{\Omega}_{\F}\lp e_l,e_k\rp}g_{\F}\lp e_i,e_j\rp+L_{\widetilde{\Omega}_{\F}\lp e_i,e_l\rp}g_{\F}\lp e_k,e_j\rp\rp\\
    &=\frac{1}{2}c\lp \varphi^k\rp c\lp\varphi^l\rp c\lp\varphi^i\rp^{k\neq l\neq i}\lp L_{\widetilde{\Omega}_{\F}\lp e_k,e_i\rp}g_{\F}\lp e_l,e_j\rp\rp\\
    &=-\frac{1}{2}c\lp \varphi^k\rp c\lp\varphi^i\rp c\lp\varphi^l\rp^{k\neq i\neq l}\lp L_{\widetilde{\Omega}_{\F}\lp e_k,e_i\rp}g_{\F}\lp e_l,e_j\rp\rp
\end{align*}
If we now re-incorprate $c\lp\varphi^j\rp$ and sum over $j$ then we obtain
\begin{align*}
    &c\lp \varphi^k\rp c\lp\varphi^l\rp c\lp\varphi^i\rp c\lp\varphi^j\rp^{k\neq i\neq l}  g_{\F}\lp R^{\F,0}\lp e_k,e_l\rp e_i,e_j\rp\\
    &\hspace{5cm}=c\lp \varphi^k\rp c\lp\varphi^l\rp c\lp\varphi^i\rp c\lp\varphi^j\rp^{k\neq i\neq l} \lp L_{\widetilde{\Omega}_{\F}\lp e_k,e_i\rp}g_{\F}\lp e_l,e_j\rp\rp
\end{align*}
However the Lie derivative of $g_{\F}$ remains symmetric, thus for fixed $k,i$ when summing over $l\neq i,k$ and all $j\neq i,k$ the skew symmetry of $c\lp\varphi^l\rp c\lp\varphi^j\rp$ combined with the symmetry of $L_{\widetilde{\Omega}_{\F}\lp e_k,e_l\rp}g_{\F}\lp e_l,e_j\rp$ gives us cancellation and the sum simplifies to
\begin{align*}
    &c\lp \varphi^k\rp c\lp\varphi^i\rp c\lp\varphi^l\rp c\lp\varphi^j\rp^{k\neq i\neq l}\lp L_{\widetilde{\Omega}_{\F}\lp e_k,e_i\rp}g_{\F}\lp e_l,e_j\rp\rp=
    -c\lp\varphi^k\rp c\lp\varphi^i\rp \delta^{lj}_{l\neq k,i}L_{\widetilde{\Omega}_{\F}\lp e_k,e_i\rp}g_{\F}\lp e_l,e_j\rp\\
    &\quad+c\lp\varphi^k\rp c\lp\varphi^i\rp c\lp\varphi^l\rp c\lp\varphi^k\rp L_{\widetilde{\Omega}_{\F}\lp e_k,e_i\rp}g_{\F}\lp e_l,e_k\rp+c\lp\varphi^k\rp c\lp\varphi^i\rp c\lp\varphi^l\rp c\lp\varphi^i\rp L_{\widetilde{\Omega}_{\F}\lp e_k,e_i\rp}g_{\F}\lp e_l,e_i\rp
\end{align*}
For a tensor $T\in\Gamma\lp\F^{*,\otimes^k}\rp$ and distinct numbers $a<b\in\{1,\cdots,k\}$ we define
\begin{align*}
    S_{ab}T\lp\cdots,v_a,\cdots,v_b\cdots\rp&:=\frac{1}{2}\lp T\lp\cdots,v_a,\cdots,v_b\cdots\rp+T\lp\cdots,v_b,\cdots,v_a\cdots\rp\rp\\
    A_{ab}T\lp\cdots,v_a,\cdots,v_b\cdots\rp&:=\frac{1}{2}\lp T\lp\cdots,v_a,\cdots,v_b\cdots\rp-T\lp\cdots,v_b,\cdots,v_a\cdots\rp\rp\\
    \tr_{ab}T\lp v_1,\cdots,v_{k-2}\rp&:=\delta^{ij}T\lp v_1,\cdots,e_i,\cdots,e_j,\cdots,v_{k-2}\rp
\end{align*}
and note that $T=S_{ab}T+A_{ab}T$.
We define the $\F$-two-forms
\begin{align*}
    \tr_{34}^{\perp}L_{\widetilde{\Omega}_{\F}}g_{\F}\lp X,Y\rp&:=\tr_{34}L_{\widetilde{\Omega}_{\F}\lp X,Y\rp}g_{\F}\lp\cdot,\cdot\rp\big|_{\bR\cdot\{ X,Y\}^{\perp}}\\
    \tr^{\perp}_{23}A_{14}L_{\widetilde{\Omega}_{\F}}g_{\F}\lp X,Y\rp&:=\tr_{23}A_{14}L_{\widetilde{\Omega}_{\F}\lp X,\cdot\rp}g_{\F}\lp\cdot,Y\rp\big|_{\bR\cdot\{ X,Y\}^{\perp}}
\end{align*}
and observe that we have
\begin{align*}
    &-\frac{1}{2}c\lp \varphi^k\rp c\lp\varphi^i\rp c\lp\varphi^l\rp c\lp\varphi^k\rp^{k\neq i\neq l}L_{\widetilde{\Omega}_{\F}\lp e_k,e_i\rp}g_{\F}\lp e_l,e_k\rp\\
    &\hspace{7cm}-\frac{1}{2}c\lp\varphi^k\rp c\lp\varphi^i\rp c\lp\varphi^l\rp c\lp\varphi^i\rp^{k\neq i\neq l}L_{\widetilde{\Omega}_{\F}\lp e_k,e_i\rp} g_{\F}\lp e_l,e_i\rp\\
    &\hspace{5cm}=c\lp \varphi^k\rp c\lp\varphi^i\rp c\lp\varphi^l\rp c\lp\varphi^k\rp^{k\neq i\neq l}L_{\widetilde{\Omega}_{\F}\lp e_i,e_k\rp}g_{\F}\lp e_k,e_l\rp\\
    &\hspace{5cm}=-c\lp\varphi^i\rp c\lp \varphi^l\rp^{i\neq l}\tr^{\perp}_{23}A_{14}L_{\widetilde{\Omega}_{\F}}g_{\F}\lp e_i,e_l\rp\\
    &\hspace{5cm}=-2c\circ\tr^{\perp}_{23}A_{14}L_{\widetilde{\Omega}_{\F}}g_{\F}
\end{align*}
and 
\[
    \frac{1}{2}c\lp\varphi^k\rp c\lp\varphi^i\rp\delta^{lj}_{l\neq k,i}L_{\widetilde{\Omega}_{\F}\lp e_k,e_i\rp}g_{\F}\lp e_l,e_j\rp=c\circ\tr^{\perp}_{34}L_{\widetilde{\Omega}_{\F}}g_{\F}
\]
If we define $Ric^{\F,0}\in\Gamma\lp\F^*\otimes\F^*\rp$ by
\[
    Ric^{\F,0}\lp X,Y\rp:=\delta^{ij}g_{\F}\lp R^{\F,0}\lp X,e_i\rp e_j,Y\rp
\]
then we have 
\begin{align*}
    &\delta_{li}c\lp \varphi^k\rp c\lp\varphi^l\rp c\lp\varphi^i\rp c\lp\varphi^j\rp g_{\F}\lp R^{\F,0}\lp e_k,e_l\rp e_i,e_j\rp+
    \delta_{ki}c\lp \varphi^k\rp c\lp\varphi^l\rp c\lp\varphi^i\rp c\lp\varphi^j\rp g_{\F}\lp R^{\F,0}\lp e_k,e_l\rp e_i,e_j\rp\\
    &\hspace{5cm}=-2c\lp\varphi^k\rp c\lp\varphi^j\rp Ric^{\F,0}\lp e_k,e_j\rp\\
    &\hspace{5cm}=2\tr  SRic^{\F,0}-4c\circ ARic^{\F,0}
\end{align*}
where $SRic^{\F,0}$ and $ARic^{\F,0}$ denote the symmetric and anit-symmetric part of $Ric^{\F,0}$.
Putting all of the terms together we find 
\[
    c\circ c\lp R^{\F,0}\rp=\frac{1}{8}c\circ\tr^{\perp}_{34}L_{\widetilde{\Omega}_{\F}}g_{\F}-\frac{1}{4}c\circ\tr^{\perp}_{23}A_{14}L_{\widetilde{\Omega}_{\F}}g_{\F}+\frac{1}{4}\tr SRic^{\F,0}-\frac{1}{2}c\circ ARic^{\F,0}
\]
which gives the explicit Weitzenbock identity as
\begin{align*}
    &\Delta^{\F,E}=\lp\nabla^E\Big|_{\F}\rp^*\nabla^E\Big|_{\F}-c\circ\lp\nabla^E_{\widetilde{\Omega}_{\F}}\rp\\
    &\hspace{2cm}+c\circ F^{E/S}\Big|_{\F}+\frac{1}{8}c\circ\tr^{\perp}_{34}L_{\widetilde{\Omega}_{\F}}g_{\F}-\frac{1}{4}c\circ\tr^{\perp}_{23}A_{14} L_{\widetilde{\Omega}_{\F}}g_{\F}+\frac{1}{4}\tr SRic^{\F,0}-\frac{1}{2}c\circ ARic^{\F,0}
\end{align*}
Observe that a great simplification occurs when one has
\[
    L_{\widetilde{\Omega}_{\F}\lp X,Y\rp}g_{\F}\lp Z,W\rp+L_{\widetilde{\Omega}_{\F}\lp Z,X\rp}g_{\F}\lp Y,W\rp+L_{\widetilde{\Omega}_{\F}\lp Y,Z\rp}g_{\F}\lp X,W\rp=0
\]
for all $X,Y,Z,W\in\Gamma\lp\F\rp$. In this case the curvature tensor $Rm^{\F,0}$ satisfies all of the algebraic Bianchi identities, the operator section $\R^{\F,0}$ is symmetric, and the tensor $Ric^{\F,0}$ is symmetric. Thus the Weitzenbock identity takes on the more familiar form
\[
    \Delta^{\F,E}=\lp\nabla^E\Big|_{\F}\rp^*\nabla^E\Big|_{\F}-c\circ\nabla^E_{\widetilde{\Omega}_{\F}}+\frac{1}{4}\tr Ric^{\F,0}+ c\circ F^{E/S}.
\]
In particular we obtain a simple proposition in the low-rank setting.
\begin{proposition}
    Suppose $\F$ has rank $n_1=2$. Then for any sub-Riemannian metric $g_{\F}$ and transverse distribution $\T$ to $\F$, $g_{\F}$ satisfies
    \[
        L_{\widetilde{\Omega}_{\F}\lp X,Y\rp}g_{\F}\lp Z,W\rp+L_{\widetilde{\Omega}_{\F}\lp Z,X\rp}g_{\F}\lp Y,W\rp+L_{\widetilde{\Omega}_{\F}\lp Y,Z\rp}g_{\F}\lp X,W\rp=0
    \]
    for all $X,Y,Z,W\in\Gamma\lp\F\rp$. 
\end{proposition}

\subsection{The Case of a Contact Manifold}
Let $\lp M^{2n+1},\theta\rp$ be a contact manifold and $\F=\operatorname{Ker}\theta\leqslant TM$ the corresponding contact distribution equipped with sub-Riemannian metric $g_{\F}$. The contact 1-form comes with a unique vector field $T$ satisfying
\[
    \iota_T\theta=1,\quad \iota_Td\theta=0,
\]
which is called the \emph{Reeb vector field}. Let $\T$ be the smooth subbundle of $TM$ spanned by $T$ and $g_{\T}$ be the Euclidean fiber metric on $\T$ which gives the Reeb field length $1$. Because $\T$ is spanned by $T$ we see that $\theta_{\T}:=\theta\big|_{\T}$ spans $\T^*$ and we obtain the isomorphism $\wedge\T^*\simeq\wedge[\theta]$ and let $g_{\wedge[\theta]}$ denote the canonical metric for the trivial bundle $\wedge[\theta]$. A calculation reveals that in this setting one has
\begin{align*}
    \nabla^{\T,u}&=\nabla^{\T,0}=d^{\T},\quad u>0\\
    \Omega_{\T}&=0\\
    \quad L_Xg_{\T}&=0,\quad X\in\Gamma\lp\F\rp\\
    \iota_T\Omega_{\F}&=d\theta\\
    d^{\T}_{\nabla^E}&=\varepsilon^{\T}_{\theta_{\T}}\lp\nabla^E_T\otimes I+I\otimes d^{\T}_T\rp=\varepsilon^{\T}_{\theta_{\T}}\nabla^{E\otimes\wedge[\theta]}_T
\end{align*}
where $d^{\T}$ denotes the canonical flat connection of the trivial bundle $\T\simeq M\times\bR$. The global formula for the asymptotic Bismut superconnection becomes
\[
    \bD^E=D^{\F,E\otimes\wedge[\theta]}+\varepsilon^{\T}_{\theta_{\T}}\nabla^{E\otimes\wedge[\theta]}_T+\frac{1}{4}c\lp d\theta\rp\iota^{\T}_T+\frac{1}{4}\varepsilon^{\T}\circ\tr Lg_{\F}
\]
\subsection{The Equivariant Setting}
Suppose that a Lie group $G$ acts on $M$ and preserves the geometric structures $\lp\F,g_{\F}\rp$ and $\lp\T,g_{\T}\rp$. 
\begin{proposition}
    The Lie group $G$ preserves $\nabla^{\F,0}$ and $\nabla^{\T,0}$ as well as the forms $\Omega_{\F}$, $\Omega_{\T}$ and $\omega_u$.
\end{proposition}
\begin{proof}
    To see this simply note that $G$ preserves the splitting $TM=\F\oplus\T$ as well as the metrics $g_{M,u}=g_{\F}\oplus u^{-1}g_{\T}$. Thus it preserves the Levi-Civita connections $\nabla^{g_{M,u}}$ as well as the projected connections $\nabla^{\F,u}$ and $\nabla^{\T,u}$. For $g\in G$ the action takes the form
    \begin{align*}
        g\cdot\nabla^{\F,u}&=g\cdot\nabla^{\F,0}+\frac{u^{-1}}{2}g\cdot\so\Omega_{\F}\\
        g\cdot\nabla^{\T,u}&=g\cdot\nabla^{\T,0}+\frac{u}{2}g\cdot\so\Omega_{\T}.
    \end{align*}
    This implies that $G$ fixes the constant and $u$-dependent terms separately. Furthermore, $G$ must fix the shape one-form 
    \[
        S_u=\nabla^{g_{M,u}}-\nabla^{\F,u}\oplus\nabla^{\T,u}
    \]
    and thus preserves $\omega_u$.
\end{proof}
If we are further given a $G$-equivariant $\bZ_2$-graded covariant Hermitian Clifford module
\[
    G\circlearrowright\lp C\ell\lp\F^*,g^*_{\F}\rp,\nabla^{\F,0}\rp\circlearrowright G\circlearrowright\lp E,h_E,\nabla^E\rp
\]
then the one-parameter family of associated Dirac operators $D_u$ are $G$-equivariant. Our general calculation in Proposition \ref{GFD} shows that $D_u$ can be given a $u$-Laurent expansion
\[
    D_u=u^{-1}D_{-1}+\bD^E+uD_1
\]
and the $G$-equivariance implies that each component must be $G$-equivariant. Thus for all $g\in G$ we have
\[
    g\cdot\bD^E=\bD^E\cdot g.
\]
An example worth mentioning in this case is when a Lie group $G$ acts by contactomorphisms on a contact manifold $\lp M,\theta\rp$ and also preserves a sub-Riemannian metric $g_{\F}$ for $\F=\Ker\theta$. Then it preserves both $\lp\F,g_{\F}\rp$ and $\lp\T,g_{\T}\rp$, where $\lp\T,g_{\T}\rp$ is the transverse distribution determined by the Reeb field $T$ and assigning $g_{\T}\lp T,T\rp\equiv 1$. For example, if $G$ is a compact Lie group acting by contactomorphisms $G\circlearrowright\lp M,\theta\rp$, then by averaging along $G$ we can always construct a $G$-invariant sub-Riemannian metric $g_{\F}$.

\section{The Heisenberg Calculus and Hypoellipticity}
On a filtered manifold $\lp M,\F\rp$ there is a microlocal calculus, called the \textbf{\emph{Heisenberg Calculus}} which one can use to investigate problems of hypoellipticity when dealing with non-elliptic operators on $M$ such as a sub-Laplacian. For differential operators acting on the sections of a complex vector bundle $E\rightarrow M$, the description of the Heisenberg calculus is fairly straightforward and we will follow the description given in \cite{DH} and \cite{Van-Yun}. We will first assume that the filtered manifold $\lp M,\F\rp$ is a contact manifold $\lp M,\theta\rp$ with filtration $\F=\Ker\theta$ and provide a proof of hypoellipticity of $\bD^E$ and discuss the subtleties in choice of $u$-varying connection $\nabla^{E,u}$. We then increase the complexity to the setting of two-step filtrations $\F$ of $TM$ and prove that hypoellipticity of $\bD^E$ holds in this more general setting, in the process providing a second proof of the hypoellipticity of $\bD^E$ on contact manifolds. We then proceed to the general setting of arbitrary filtrations and give a redefinition of the asymptotic Bismut superconnection which properly generalizes our original definition. We also provide a close cousin of the asymptotic Bismut superconnection which is \emph{always hypoelliptic}. We end this section by computing the kernel of the asymptotic Bismut superconnection on principal $\bS^1$-bundles with non-vanishing curvature. 
\subsection{The Rockland Condition}
We begin with a general filtered manifold $\lp M,\F\rp$ of step size $m$
\[  
    \F=\F^1\leqslant\F^2\leqslant\cdots\leqslant\F^m=TM
\]
satisfying $[\Gamma\lp\F^i\rp,\Gamma\lp\F^j\rp]\subseteq\Gamma\lp\F^{i+j}\rp$, and we let $n_j$ denote the rank of $\F^j/\F^{j-1}$.
To describe the calculus Heisenberg calculus associated to the above filtration, first note that one can form the vector bundles
\[
    \t^j_{\F}M:=\F^j/\F^{j-1},\quad \t_{\F}M:=\bigoplus_{j=1}^m\t^j_{\F}M.
\]
At each $x\in M$ there is a natural bilinear form on $\t^j_{\F}M_x\otimes\t^k_{\F}M_x$ with values in $\t^{j+k}_{\F}M_x$ defined as
\begin{align*}
    [\cdot,\cdot]_x:\t^j_{\F}M_x\times\t^k_{\F}M_x\rightarrow \t^{k+j}_{\F}M_x\\
    \lp X_x,Y_x\rp\rightarrow [\tilde{X},\tilde{Y}]_x+\F^{j+k-1}_x
\end{align*}
where $\tilde{X},\tilde{Y}$ denote local sections of $\F^j$ and $\F^k$ which extend $X_x$ and $Y_x$. If we extend the bracket to $\t_{\F}M_x$ then one checks that this bracket satisfies the Jacobi identity. By defining the grading degree of $\t^j_{\F}M_x$ to be $j$, one obtains a \textbf{\emph{graded $m$-step nilpotent Lie algebra}}. We denote the \textbf{\emph{osculating group at $x$}} by $T_{\F}M_x:=\exp\t_{\F}M_x$ which is defined to be the simply connected exponential of $\t_{\F}M_x$. In this case 
\[
    \t_{\F}M:=\bigsqcup_{x\in M}\t_{\F}M_x,\quad T_{\F}M:=\bigsqcup_{x\in M} T_{\F}M_x,
\]
become bundles of Lie algebras and Lie groups, respectively. In general, the local bracket structure of $\t_{\F}M$ and $T_{\F}M$ \emph{cannot be locally trivialized}. Suppose we are given a rank $r$ vector bundle $E\rightarrow M$ and a differential operator $L:\Gamma\lp E\rp\rightarrow\Gamma\lp E\rp$. Along a neighborhood $U$ of $M$ for which each $\t^j_{\F}M$ admits a framing by projected vector fields $\{ X^j_{j_k}\}_{j_k=1}^{n_j}$, and $E$ for which trivializes, we can express $L$ as a finite sum
\[
    L=\sum_{\alpha_1,\cdots,\alpha_m} a_{\alpha_1,\cdots,\alpha_m}\lp \vec{X}^1\rp^{\alpha_1}\cdots \lp \vec{X}^m\rp^{\alpha_m},
\]
where $\alpha_j=\lp\alpha_{j_1},\cdots,\alpha_{j_{n_j}}\rp\in\lp\bZ_{\geq0}\rp^{n_j}$, $a_{\alpha_1,\cdots,\alpha_m}\in C^{\infty}\lp U,M_r\lp \bC\rp\rp$, and 
\[
    \lp \vec{X}^j\rp^{\alpha_j}:=\lp X^j_1\rp^{\alpha_{j_1}}\cdots\lp X^j_{n_j}\rp^{\alpha_{j_{n_j}}}.
\]
The operator $L$ has \textbf{\emph{$\F$-filtered differential order $l$}} if for arbitrary local choices of vector field trivializations of $\t^j_{\F}M$ and local trivializations of $E$, for any non-vanishing coefficient $a_{\alpha_1,\cdots,\alpha_m}$ one has
\[
    \sum_{j=1}^{m}j\cdot|\alpha_j|\leq l.
\]
For a collection of multi-indices $\lp\alpha_1,\cdots,\alpha_m\rp$ we will define their $\F$-homogeneous length as
\[
    |\lp\alpha_1,\cdots,\alpha_m\rp|_{\F}:=\sum_{j=1}^{m}j\cdot|\alpha_j|.
\]
For a given $x\in M$ one can freeze the coefficients of $a_{\alpha_1,\cdots,\alpha_m}$ and form the assignment
\[
    \lp \vec{X}^1\rp^{\alpha_1}\cdots \lp \vec{X}^m\rp^{\alpha_m}\Rightarrow \lp [\vec{X}^1]_x\rp^{\alpha_1}\cdots \lp [\vec{X}^{m-1}]_x\rp^{\alpha_m}\in\U\lp\t_{\F}M_x\rp
\]
and thus upon tensoring obtain an element
\[  
    \widecheck{\sigma}_l\lp L\rp_x:=\sum_{|\lp\alpha_1,\cdots,\alpha_m\rp|_{\F}=l} a_{\alpha_1,\cdots,\alpha_m}\lp x\rp\lp [\vec{X}^1]_x\rp^{\alpha_1}\cdots \lp [\vec{X}^{m-1}]_x\rp^{\alpha_m}\in\U\lp\t_{\F}M_x\rp\otimes M_r\lp\bC\rp.
\]
One checks that the above element transforms like a section of $\U\lp\t_{\F}M\rp\otimes\End\pi^*E$ and thus gives one a section of said bundle. For $L\in\text{Diff}\lp M,E\rp$ we let $|L|^{H.O.}\in\bZ_{\geq0}$ be the minimal integer for which in any local framing of $\t_{\F}M$ and local trivialization of $E$, the Heisenberg order of all terms appearing within the local expansion of $L$ are bounded by $|L|^{H.O.}$. If $E$ is equipped with a connection $\nabla^E$ then by definition we have
\[
    X\in\Gamma\lp\F^j\rp\Rightarrow |\nabla^E_X|^{H.O.}\leq j,\quad \widecheck{\sigma}_j\lp\nabla^E_X\rp_x=[X]^j_x\otimes I_{E_x}
\]
where $[X]^j_x\in\t^j_{\F}M_x$ denotes the projection of $X_x$. The filtered symbol has the important \textbf{\emph{composition property}}: If $L_1$ has filtered order $l_1$ and $L_2$ has filtered order $l_2$, then the composition $L_1L_2$ has filtered order $l_1+l_2$ and
\[
    \widecheck{\sigma}_{l_1+l_2}\lp L_1L_2\rp=\widecheck{\sigma}_{l_1}\lp L_1\rp\widecheck{\sigma}_{l_2}\lp L_2\rp
\]
The main use of this new notion of differential order is the following theorem relating the representation theory of the osculating Lie groups $T_{\F}M$ to the hypoellipticity of differential operators. To set up notation, if $x\in M$ we let $\Pi_x:T_{\F}M_x\rightarrow U\lp H \rp$ denote a unitary representation and $\pi_x:\t_{\F}M_x\rightarrow \End H^{\infty}$ the induced Lie algebra representation on the smooth vectors $H^{\infty}$. If $V$ is a fixed complex vector space then we can inflate the Lie algebra representation 
\[
    \pi^V_x:\t_{\F}M_x\otimes \gl\lp V\rp\rightarrow \gl\lp H^{\infty}\otimes V\rp
\]
by defining 
\[
    \pi^V_x\lp X\otimes S \rp:=\pi_x\lp X\rp\otimes I_{ V}+I_{H^{\infty}}\otimes S
\]
which extends to an algebra homomorphism
\[
    \pi^V_x:\U\lp \t_{\F}M_x\rp\otimes \End V\rightarrow \End\lp H^{\infty}\otimes V\rp.
\]
When expressing $\pi^V_x\lp X\otimes I_V\rp$ and $\pi^V_x\lp 1\otimes S\rp$ we will usually drop the $\otimes I_V$ and $I_{H^{\infty}}\otimes$ and simply express $\pi^V_x\lp X\otimes I_V\rp=\pi_x\lp X\rp$ and $\pi^V_x\lp 1\otimes S\rp=S$. We also will often identify elements of $\U\lp\t_{\F}M_x\rp$ with the algebra of left-invariant differential operators on $T_{\F}M_x$.
\begin{theorem}[van Erp-Yunken \cite{Van-Yun}, Dave-Haller \cite{DH}]
\label{OGRC}
    Suppose that $\lp M,\F\rp$ is an $m$-step filtered manifold and $E\rightarrow M$ is a complex vector bundle with differential operator $L$ with Heisenberg order $l$. Then $L$ is hypoelliptic if for every $x\in M$ and every irreducible unitary representation $\Pi_x:T_{\F}M_x\rightarrow U(H)$, the induced operator on smooth vectors of the inflated representation $\pi^{E_x}_x:\mathfrak{t}_{\F}M_x\otimes \gl\lp E_x\rp\rightarrow \gl\lp H^{\infty}\otimes E_x\rp$ 
    \[
        \pi^{E_x}_x\lp\widecheck{\sigma}_p\lp L\rp_x\rp:H^{\infty}\otimes E_x\rightarrow H^{\infty}\otimes E_x
    \]
    is injective. If $M$ is closed, then $\Ker L$ is a finite dimensional space of smooth sections.
\end{theorem}
One says that a differential operator $L\in \text{Diff}\lp M,E\rp$ is \textbf{\emph{Rockland}} if it satisfies the above criteria for hypoellipticity. For a general filtered manifold the above theorem reduces the determination of hypoellipticity to the study of the action of the cosymbol of the operator on the irreducible unitary representations of the simply connected nilpotent Lie group $T_{\F}M_x$, for each $x\in M$.
\subsection{Hypoellipticity on Contact Manifolds}
We now restrict ourselves to contact manifolds $\lp M,\theta\rp$. Fix a point $x\in M$ and recall that Darboux's theorem guarantees the existence of a neighborhood $U$ of $x$ and a \textbf{\emph{contact framing}} of $\F\big|_U$ of the form $\{ Q_j,P_j\}_{j=1}^n$ which satisfies
\[
    [Q_j,Q_k]=[P_j,P_k]=0,\quad [Q_j,P_k]=\delta_{jk} T.
\]
If we assume that $U$ is chosen so that $E\rightarrow M$ trivializes along $U$, $E\big|_U\simeq \underline{\bC^r}_U$, then any differential operator $L\in\text{Diff}\lp M,E\rp$ will take the form of a finite sum of operators
\[
    L=\sum_{|\alpha|+|\beta|+\gamma}a_{\alpha\beta\gamma}Q_i^{\alpha_i}P_j^{\beta_j}T^{\gamma},\quad a_{\alpha\beta\gamma}\in C^{\infty}\lp U,M_r\lp\bC\rp\rp.
\]
Given that the order of vector fields tangent to $\F$ are $1$ and the order of vector fields tangent to $\T\simeq TM/\F$ are $2$, in this case one says that $L$ has \textbf{\emph{Heisenberg order $\leq l$}} if for all non-vanishing coefficients $a_{\alpha\beta\gamma}$ above one has $|\alpha|+|\beta|+2\gamma\leq l$. The \textbf{\emph{Heisenberg co-symbol of order $l$}} of $L$ at $x$ is defined by taking the highest order components i.e. the components for which $|\alpha|+|\beta|+2\gamma=l$ within the local form of $L$, and one evaluates the matrix components at $x$. The cosymbol then is
\[
    \widecheck{\sigma}_l\lp L\rp_x:=\sum_{|\alpha|+|\beta|+2\gamma=l}a_{\alpha\beta\gamma}\lp x\rp Q_i^{\alpha_i}P^{\beta_j}_jT^{\gamma}.
\]
By definition, if $E$ is given a connection $\nabla^E$ then we have
\begin{align*}
    X\in\Gamma\lp\F\rp\Rightarrow |\nabla^E_X|^{H.O.}=1,\quad \widecheck{\sigma}_1\lp\nabla^E_X\rp_x=X_x\otimes I_{E_x},\\
    Z\in\Gamma\lp\T\rp\Rightarrow |\nabla^E_Z|^{H.O.}=2,\quad \widecheck{\sigma}_2\lp\nabla^E_Z\rp_x=Z_x\otimes I_{E_x}.
\end{align*}
It is here that the simplicity of the contact assumtion becomes apparent: For all $x\in M$ the osculating group is isomorphic to the Heisenberg group $T_{\F}M_x\simeq\H$ which is the smooth manifold $\H=\bC^n\times \bR$ with group structure 
\[
    \lp \vec{z},t\rp\cdot \lp\vec{w},s\rp:=\lp\vec{z}+\vec{w},t+s+\frac{1}{2}\text{Im}\overline{\vec{z}}\cdot \vec{w}\rp.
\]
As a basis for the Lie algebra $\h$ of $\H$ one can take the left-invariant vector fields
\[
    Q_j=\frac{\partial}{\partial x_j}-\frac{y_j}{2}\frac{\partial}{\partial t},\quad P_j=\frac{\partial}{\partial y_j}+\frac{x_j}{2}\frac{\partial}{\partial t},\quad T=\frac{\partial}{\partial t},
\]
which satisfy the commutation relations
\[
    [Q_j,Q_k]=[P_j,P_k]=0,\quad [Q_j,P_k]=\delta_{jk} T
\]
The Heisenberg group comes equipped with a contact form
\[
    \theta_{\H}:=dt-\frac{1}{2}\delta^{jk}\lp x_jdy_k-y_jdx_k \rp
\]
for which the vector fields $\{ Q_j,P_j\}_{j=1}^n$ give rise to a global contact framing of $\F$ and $T$ is the Reeb vector field of $\theta_{\H}$. \par
The irreducible representations of $\H$ fall into two types: a $2n$-parameter family of one-dimensional complex representations and a one-parameter family of infinite dimensional representations. For $\vec{\xi}=\lp\xi_1,\cdots,\xi_{2n}\rp\in\bR^{2n}\setminus{\vec{0}}$, the one-dimensional representations are given at the Lie algebraic level as
\[
    \pi_{\vec{\xi}}\lp Q_j\rp=2\pi i\xi_j,\quad \pi_{\vec{\xi}}\lp P_j\rp=2\pi i\xi_{n+j},\quad \pi_{\vec{\xi}}\lp T\rp=0.
\]
The infinite dimensional representations are parametrized by $\lambda\in\bR\setminus{0}$ and are the standard Schrödinger representations of $\H$ acting on $L^2\lp\bR^n_q\rp$. The smooth vectors of these representations are the Schwartz class functions and the representation at the Lie algebraic level is given by
\[
    \pi_{\lambda}\lp Q_j\rp= \frac{\lambda}{i} q_j,\quad
    \pi_{\lambda}\lp P_j\rp=\frac{\partial}{\partial q_j},\quad
    \pi_{\lambda}\lp T\rp=i\lambda.
\]
A quick look at our operator will show that it is \emph{not} a Rockland operator. The operator $D^{\F,E\otimes\wedge\T^*}$ will only differentiate along the $\F$-directions and thus locally will have Heisenberg order $1$, where as the terms with no derivatives will be order zero. Thus the highest order component of $\bD^{E}$ will be the component involving differentiation along the Reeb directions $T$ and its cosymbol at $x\in M$ is the operator
\[
    \widecheck{\sigma}_2\lp\bD^{E}\rp_x=\varepsilon^{\T}_{\theta_{\T_x}}T_x.
\]
However the element $T_x\in\mathfrak{t}_{\F}M_x$ is zero within any of the one-dimensional representations of $\H$. Our first method for determining the hypoellipticity of $\bD^E$ will be to look at its square
\[
  \Deltab^{E}:=\lp\bD^E\rp^2. 
\]
\begin{theorem}
\label{OGH}
    Suppose that $\lp M,\theta,g_{\F}\rp$ is a contact sub-Riemannian manifold and we have a $\bZ_2$-graded covariant Hermitian Clifford module
    \[
        \lp C\ell\lp\F^*,g^*_{\F}\rp,\nabla^{\F,0}\rp\circlearrowright\lp E,h_E,\nabla^E\rp
    \]
    Then the operator
    \begin{align*}
    \bD^E=D^{\F,E\otimes\wedge[\theta]}+\varepsilon^{\T}_{\theta_{\T}}\nabla^{E\otimes\wedge[\theta]}_T+\frac{1}{4}c\lp d\theta\rp\iota^{\T}_T+\frac{1}{4}\varepsilon^{\T}\circ\tr Lg_{\F}
\end{align*}
is hypoelliptic. If $M$ is closed, then $\Ker\bD^E$ is a finite dimensional space of smooth sections.
\end{theorem}
\begin{proof}
    We first make the observation that a differential operator $A$ is hypoelliptic if and only if $A^2$ is hypoelliptic. Turning to our operator $\bD^E$ we study its square $\Deltab^E$. Observe that by our global formula for $\bD^E$ we have
    \begin{align*}
        &\Deltab^E=\lp D^{\F,E\otimes\wedge[\theta]}\rp^2+
        \left[D^{\F,E\otimes\wedge[\theta]},\varepsilon^{\T}_{\theta_{\T}}\nabla^{E\otimes\wedge[\theta]}_T+\frac{1}{4}c\lp d\theta\rp\iota^{\T}_T+\frac{1}{4}\tr L_Tg_{\F}\varepsilon^{\T}_{\theta}\right]\\&\hspace{8cm}+\left[\varepsilon^{\T}_{\theta_{\T}}\nabla^{E\otimes\wedge[\theta]}_T,\frac{1}{4}c\lp d\theta\rp\iota^{\T}_T\right]+\frac{1}{16}\tr L_Tg_{\F}c\lp d\theta\rp
    \end{align*}
    where the commutators above are supercommutators of odd operators.
    Observe that the distribution $\F$ is left invariant under the flow of $T$ which implies that for any $X\in\Gamma\lp\F\rp$ one has $[T,X]\in\Gamma\lp \F\rp$. The presence of $\nabla^{E\otimes\wedge[\theta]}_T$ in the second commutator tells us that the co-symbol of $\Deltab^E$ at a given $x\in M$ will be at least of order $2$. Moreover, because $\theta_{\T}\in\Gamma\lp\T^*\rp$ is parallel, we have
    \[
        D^{\F,E\otimes\wedge[\theta]}\varepsilon^{\T}_{\theta_{\T}}=-\varepsilon^{\T}_{\theta_{\T}}D^{\F,E\otimes\wedge[\theta]}\Rightarrow \left[D^{\F,E\otimes\wedge[\theta]},\varepsilon^{\T}_{\theta_{\T}}\nabla^0_T\right]=-\varepsilon^{\T}_{\theta_{\T}}\left[D^{\F,E\otimes\wedge[\theta]},\nabla^{E\otimes\wedge[\theta]}_T\right].
    \]
    
    Thus all terms within the commutator involving $D^{\F,E\otimes\wedge[\theta]}$ have Heisenberg order $\leq1$ and we can discard the commutator involving $D^{\F,E\otimes\wedge[\theta]}$ and the zeroth order term $\tr L_Tg_{\F}c\lp d\theta\rp$. By Proposition \ref{WF} applied to $D^{\F,E\otimes\wedge[\theta]}$ we have
    \begin{align*}
        \Delta^{\F,E\otimes\wedge[\theta]}&=\lp\nabla^{E\otimes\wedge[\theta]}\Big|_{\F}\rp^*\nabla^{E\otimes\wedge[\theta]}\Big|_{\F}-c\circ\lp\nabla^{E\otimes\wedge[\theta]}_{\widetilde{\Omega}_{\F}}\rp +c\circ K^E
    \end{align*}
    We can simplify one operator above as
    \[
        c\circ\lp \nabla^{E\otimes\wedge[\theta]}_{\widetilde{\Omega}_{\F}}\rp=c\lp d\theta\rp\nabla^{E\otimes\wedge[\theta]}_T
    \]
    thus giving
    \begin{align*}
        \Delta^{\F,E\otimes\wedge[\theta]}&=\lp\nabla^{E\otimes\wedge[\theta]}\Big|_{\F}\rp^*\nabla^{E\otimes\wedge[\theta]}\Big|_{\F}-c\lp d\theta\rp\nabla^{E\otimes\wedge[\theta]}_T+c\circ K^E
    \end{align*}
   Collecting all terms within $\Deltab^E$ we see that at $x\in M$ the cosymbol is given by
   \begin{align*}
        \widecheck{\sigma}_2\lp \Deltab^E\rp_x&=\widecheck{\sigma}_2\lp \lp\nabla^{E\otimes\wedge[\theta]}\Big|_{\F}\rp^*\nabla^{E\otimes\wedge[\theta]}\Big|_{\F}-c\lp d\theta\rp\nabla^{E\otimes\wedge[\theta]}_T+\left[\varepsilon^{\T}_{\theta_{\T}}\nabla^{E\otimes\wedge[\theta]}_T,\frac{1}{4}c\lp d\theta\rp\iota^{\T}_T\right]\rp_x\\
        &=\widecheck{\sigma}_2\lp\lp\nabla^{E\otimes\wedge[\theta]}\Big|_{\F}\rp^*\nabla^{E\otimes\wedge[\theta]}\Big|_{\F}-\frac{3}{4}c\lp d\theta\rp\nabla^{E\otimes\wedge[\theta]}_T\rp_x
   \end{align*}
   where we used
   \[
    \widecheck{\sigma}_2\lp\left[\varepsilon^{\T}_{\theta_{\T}}\nabla^{E\otimes\wedge[\theta]}_T,\frac{1}{4}c\lp d\theta\rp\iota^{\T}_T\right]\rp_x=\widecheck{\sigma}_2\lp\frac{1}{4}c\lp d\theta\rp\nabla^{E\otimes\wedge[\theta]}_T\rp_x
   \]
    Now for a fixed $x\in M$ let $\{ u_j,v_j\}_{i=j}^n$ be an orthonormal basis of $\lp\F_x,g_{\F_x}\rp$ which diagonalizes $d\theta_x$
    \[
        d\theta_x\lp u_j,v_k\rp=\delta_{jk}\alpha_j,\quad \alpha_j>0.
    \]
    Let $\{ U_j,V_j\}_{i=j}^n$ be a local orthonormal framing of $\lp\F,g_{\F}\rp$ which extends the framing $\{ u_j,v_j\}_{j=1}^n$. Then at $x$ we have
    \[
        \left[U_j,U_k\right]_x=\left[V_j,V_k\right]_x=0,\quad \left[U_j,V_k\right]_x=-\delta_{jk}\alpha_j T_x+E_x,
    \]
    where $E_x\in\F_x$. The local form of the laplacian
    \[
        \lp\nabla^E\Big|_{\F}\rp^*\nabla^E\Big|_{\F}=-\delta^{jk}\lp\nabla^E_{U_j}\nabla^E_{U_k}-\nabla^E_{\nabla^{\F,0}_{U_j}U_k}+\nabla^E_{V_j}\nabla^E_{V_k}-\nabla^E_{\nabla^{\F,0}_{V_j}V_k}\rp
    \]
    demonstrates that the cosymbol of $\Deltab^E$ at $x$ is 
    \begin{align*}
        \widecheck{\sigma}_2\lp \Deltab^E\rp_x&=-\delta^{jk}\lp u_ju_k+v_jv_k\rp-\frac{3}{4}c\lp d\theta_x\rp T_x\\
        &=-\delta^{jk}\lp u_ju_k+v_jv_k-\frac{3}{4}\alpha_jc\lp\varphi_{u_j}\rp c\lp\varphi_{v_k}\rp T_x\rp
    \end{align*}
    We now must check that the above operator is injective in every irreducible unitary representation of $\mathfrak{t}_{\F}M_x$. First note that if we define $e_j=\frac{1}{\sqrt{\alpha_j}}v_j$ and $f_j=\frac{1}{\sqrt{\alpha_j}}u_j$, then we have
    \[
        \left[e_j,e_k\right]=\left[f_j,f_k\right]=0,\quad\left[e_j,f_k\right]=\delta_{jk}T_x,
    \]
    where the commutators are now taken according to the pointwise Lie algebra structure of $\t_{\F}M_x$.
    The cosymbol takes the form
    \[
        \widecheck{\sigma}_2\lp\Deltab^E\rp_x=\delta^{jk}\alpha_j\lp- e_je_k-f_jf_k-\frac{3}{4}c\lp\varphi_{u_j}\varphi_{v_k}\rp T_x\rp,
    \]
    and for $\vec{\xi}\in\bR^2\setminus{\vec{0}}$, the inflated finite dimensional representation gives
    \[
        \pi^{E\otimes\wedge[\theta]_x}_{\vec{\xi}}\lp\widecheck{\sigma}_2\lp\Deltab^E\rp_x\rp=(2\pi)^2\lp \sum_{i=j}^n\alpha_j\rp|\vec{\xi}|^2,
    \]
    which is injective on $E\otimes\wedge[\theta]_x$ for any $\vec{\xi}\neq 0$. For $\lambda\in\bR\setminus{0}$ we have
    \[
        \pi_{\lambda}^{E\otimes\wedge[\theta]_x}\lp\widecheck{\sigma}_2\lp\Deltab^E\rp_x\rp=\delta^{jk}\alpha_j\lp-\frac{\partial}{\partial q_j}\frac{\partial}{\partial q_k}+\lambda^2q_jq_k-\frac{3\lambda i}{4}c\lp \varphi_{u_j}\varphi_{v_k}\rp \rp.
    \]
    Observe that the operator $c\lp\varphi_{u_j}\varphi_{v_j}\rp$ squares to $-1$ and is skew-adjoint and thus has eigenvalues $\pm i$. With respect to a diagonalizing basis for $c\lp\varphi_{u_j}\varphi_{v_j}\rp$, the $j$-th term in the above sum is
    \[
        \alpha_j\lp -\frac{\partial^2}{\partial q_j^2}+\lambda^2q^2\pm\frac{3}{4}\lambda\rp
    \]
    The operator $-\frac{\partial^2}{\partial q_j^2}+\lambda^2q^2$ is bounded below by $|\lambda|$ and thus we have for a Schwartz class function $\psi$ with values in $E\otimes\wedge[\theta]_x$ one has
    \[
        \int_{\bR^n}h_{E\otimes\wedge[\theta]_x}\lp\alpha_j\lp-\frac{\partial^2}{\partial q_j^2}+\lambda^2q_j^2-\frac{3\lambda i}{4}c_{\F}\lp \varphi_{u_j}\varphi_{v_j}\rp\hat{\otimes}I\rp\psi,\psi\rp dq\geq \frac{|\lambda|\alpha_j}{4}\int_{\bR^n}h_{E\otimes\wedge[\theta]_x}\lp\psi,\psi\rp dq
    \]
    where $h_{E\otimes\wedge[\theta]}=h_{E}\otimes g_{\wedge[\theta]}$. Summing over $j$ gives
    \[
        \int_{\bR^n}h_{E\otimes\wedge[\theta]_x}\lp\pi_{\lambda}^{E\otimes\wedge[\theta]_x}\widecheck{\sigma}_2\lp\Deltab^E\rp_x\psi,\psi\rp dq\geq\frac{|\lambda|}{4}\lp\sum_{j=1}^n\alpha_j\rp\int_{\bR^n}h_{E\otimes\wedge[\theta]_x}\lp\psi,\psi\rp dq.
    \]
    Thus $\widecheck{\sigma}_2\lp\Deltab^E\rp_x$ is injective in every inflated irreducible unitary representation of $T_{\F}M_x$ for every $x\in M$. Thus $\Deltab^{\bE}$ is a Rockland operator and by Theorem \ref{OGRC} above, the operator $\Deltab^E$ is hypoelliptic implying that $\bD^E$ is hypoelliptic as well.\par
\end{proof}
\subsection{A subtlety within the choice of connection}
When working with spinors, the natural choice of connections $\nabla^{\Sc,u}$ grants a Rockland operator $\Deltab^{\Sc_{\F}}$ and the choice generalized to an arbitrary Clifford module. One can inquire about the dependence of the Rockland condition on the choice on Clifford compatible connection. We investigate the issue for the Clifford module given by $\wedge\F^*$ for in this case their are three natural choices of Clifford compatible connection and we will see that our choice is the only candidate which satisfies the Rockland condition. Recall that $\wedge\F^*$ is a \emph{bimodule} for the algebra bundle $C\ell\lp\F^*,g^*_{\F}\rp$ and we denote the left and right actions with superscripts $c^L$ and $c^R$. For the first case we take our covariant Clifford module to be $\wedge\F^*$ with the natural induced connection $\nabla^{\wedge\F^*,u}$ coming from $\nabla^{\F,u}$. Let us see the problem with this connection on the three dimensional Heisenberg group. 
\begin{example}
    Let us see what we have computed in the case of the $3$ dimensional Heisenberg group $\lp \H,\theta_{\H},g_{\F}\rp$. We let $Q,P$ be the standard contact vector fields satisfying
    \[
        [Q,P]=T,\quad g_{\F}\lp Q,P\rp=0,\quad g_{\F}\lp Q,Q\rp=g_{\F}\lp P,P\rp=1.
    \]
    We let $\{\varphi_Q,\varphi_P,\theta_{\H}\}$ denote the dual framing to $\{ Q,P,T\}$ and $\theta_{\H,\T}$ the restriction of $\theta_{\H}$ to $\T$. In this case, for $u>0$ if $A^u$ denotes the connection matrix for $\nabla^{\F,u}$ relative to the framing $\{ Q,P\}$ of $\F$ then
    \begin{align*}
        A^u&=\theta\otimes \begin{bmatrix}
            0 & u^{-1}/2 \\
            -u^{-1}/2 & 0
        \end{bmatrix},\\
        \omega_u&=\varphi_Q\otimes \frac{1}{2}u^{-1}\varphi_P\wedge\theta-\varphi_P\otimes\frac{1}{2}u^{-1}\varphi_Q\wedge\theta.
    \end{align*}
    In which case the one-parameter family of Dirac operators $D_u$ acting on $\Gamma\lp\wedge\F^*\otimes\wedge\T^*\rp$ is 
    \begin{align*}
        D_u=c^L&\lp\varphi_Q\rp\lp Q+\frac{1}{4}u^{-1}c^L\lp\varphi_P\rp\lp \varepsilon^{\T}_{\theta_{\H,\T}}-u\iota^{\T}_{T}\rp\rp+c^L\lp\varphi_P\rp\lp P-\frac{1}{4}u^{-1}c^L\lp\varphi_Q\rp\lp\varepsilon^{\T}_{\theta_{\H,\T}}-u\iota^{\T}_{T}\rp\rp\\
        &+\lp\varepsilon^{\T}_{\theta_{\H,\T}}-u\iota^{\T}_{T}\rp\lp T-\frac{1}{4}u^{-1}\lp c^L\lp\varphi_Q\rp c^L\lp\varphi_P\rp-c^R\lp\varphi_P\rp c^R\lp\varphi_Q\rp\rp\rp.
    \end{align*}
    Collecting the terms constant in $u$ gives the limiting operator $\bD=\frac{d}{du}\big|_{u=0}\lp uD_u\rp$ below
    \[
        \bD=c^L\lp\varphi_Q\rp Q+c^L\lp \varphi_P\rp P +\varepsilon^{\T}_{\theta_{\H,\T}}T-\frac{1}{4}\lp c^L\lp \varphi_Q\rp c^L\lp\varphi_P\rp+c^R\lp\varphi_P\rp c^R\lp \varphi_Q\rp\rp\iota^{\T}_{T},
    \]
    which when squared gives
    \[
      \Deltab=-Q^2-P^2+c^L\lp \varphi_Q\varphi_P\rp T-\frac{1}{4}\lp c^L\lp \varphi_Q\rp c^L\lp\varphi_P\rp+c^R\lp\varphi_P\rp c^R\lp \varphi_Q\rp\rp T+\O_{\Psi^{\theta}}\lp 1\rp. 
    \]
     The unfortunate feature of the above operator is that when we restrict to forms in $\wedge\F^*$ with odd degree the term
     \[
        c^L_{\F}\lp \varphi_Q\rp c^L_{\F}\lp\varphi_P\rp+c^R_{\F}\lp\varphi_P\rp c^R_{\F}\lp \varphi_Q\rp,
     \]
     vanishes, and thus at any $x\in \H$, relative to the splitting 
     \[
        \wedge\F^*\otimes\wedge\T^*=\wedge^{even}\F^*\otimes\wedge\T^*\oplus \wedge^{odd}\F^*\otimes\wedge\T^*
     \]
     the symbol of $\Deltab$ takes the block diagonal form
     \[
        \widecheck{\sigma}_2\lp\Deltab\rp_x=\begin{bmatrix}
            -Q^2-P^2+\frac{1}{2}c^L_{\F}\lp\varphi_Q\varphi_P\rp T & 0 \\
            0 & -Q^2-P^2+ c^L_{\F}\lp\varphi_Q\varphi_P\rp T 
        \end{bmatrix}, 
     \]
     and the lower block is \emph{not} Rockland in the Heisenberg calculus.\par\qed
\end{example}
While the above calculation is quite depressing, we can take a step back and see what went wrong. If $A^{\F}\in\Omega^1\lp P_O\lp\F,g_{\F}\rp,\o(2n)\rp$ denotes the connection $1$-form, then we have that the induced connection $1$-form on $\wedge\F^*$ is given by 
\[
    A^{\wedge\F^*}=\frac{1}{2}[A^{\F},\cdot],
\]
where the bracket is with respect to Clifford multiplication on $C\ell\lp\bR^n\rp$. This brought in the pesky $c^R_{\F}\lp\varphi_P\rp c^R_{\F}\lp\varphi_Q\rp$ above. It is here that we have to remember that there is a bit of freedom in our choice of covariant Clifford structure $\nabla^{E,u}$ as well as perturbation $\nabla^{\bE,u}$ for $E\otimes\wedge\T^*$. We see here that there is another natural choice for $\nabla^{\bE,u}$. This comes from the observation that the above connection $\nabla^u$ on $\wedge\F^*\hat{\otimes}\wedge\T^*$ is \emph{not} the Levi-Civita connection on $\wedge\F^*\hat{\otimes}\wedge\T^*\simeq\wedge T^*M$. To obtain the Levi-Civita connection we would need to add in the term $-\frac{1}{2}\omega^R_u$ to $\nabla^u$. However, in this case the addition of extra right multiplication does not remedy the issue of being Rockland as we show below.
\begin{example}
    If we instead add in the right multiplication by $-\frac{1}{2}\omega_u^R$, we obtain for $u>0$
    \begin{align*}
        D_u=c^L&\lp\varphi_Q\rp\lp Q+\frac{1}{4}u^{-1}\lp c^L\lp\varphi_P\rp\lp \varepsilon^{\T,L}_{\theta_{\H,\T}}-u\iota^L_T\rp-\lp\varepsilon^{\T,R}_{\theta_{\H,\T}}-u\iota^{\T,R}_T\rp c^R\lp\varphi_P\rp\rp \rp\\
        &+c^L\lp\varphi_P\rp\lp P-\frac{1}{4}u^{-1}\lp c^L\lp\varphi_Q\rp\lp\varepsilon^{\T,L}_{\theta_{\H,\T}}-u\iota^{\T,L}_T\rp-\lp\varepsilon^{\T,R}_{\theta_{\H,\T}}-u\iota^{\T,R}_T\rp c^R\lp\varphi_Q\rp\rp\rp\\
        &\quad+\lp\varepsilon^{\T,L}_{\theta_{\H,\T}}-u\iota^{\T,L}_T\rp\lp T-\frac{1}{4}u^{-1}\lp c^L\lp\varphi_Q\rp c^L\lp\varphi_P\rp-c^R\lp\varphi_P\rp c^R\lp\varphi_Q\rp\rp\rp.
    \end{align*}
    Taking the finite part gives
    \begin{align*}
        \bD=c^L&\lp\varphi_Q\rp\lp Q-\frac{1}{4}\lp c^L\lp\varphi_P\rp\iota^{\T,L}_T-\iota^{\T,R}_T c^R\lp\varphi_P\rp\rp \rp
        +c^L\lp\varphi_P\rp\lp P+\frac{1}{4}\lp c^L\lp\varphi_Q\rp\iota^{\T,L}_T-\iota^{\T,R}_T c^R\lp\varphi_Q\rp\rp\rp\\
        &\quad+\varepsilon^{\T,L}_{\theta_{\H,\T}}T-\frac{1}{4}\lp c^L\lp \varphi_Q\rp c^L\lp\varphi_P\rp-c^R\lp\varphi_P\rp c^R\lp \varphi_Q\rp\rp\iota^{\T,L}_T\\
        &=c^L\lp \varphi_Q\rp Q+c^L\lp\varphi_P\rp P+\varepsilon^{\T,L}_{\theta_{\H,\T}}T\\
        &\quad-\frac{1}{4}\lp c^L\lp \varphi_Q\rp c^L\lp\varphi_P\rp+c^R\lp\varphi_P\rp c^R\lp \varphi_Q\rp\rp\iota^{\T,L}_T+\frac{1}{4}\lp c^L\lp\varphi_Q\rp\iota^{\T,R}_Tc^R\lp\varphi_P\rp-c^L\lp\varphi_P\rp\iota^{\T,R}_Tc^R\lp\varphi_Q\rp \rp
    \end{align*}
    which when squared gives 
    \begin{align*}
        \Deltab=&-Q^2-P^2+c^L\lp\varphi_Q\varphi_P\rp T-\frac{1}{4}\lp c^L\lp \varphi_Q\rp c^L\lp\varphi_P\rp+c^R\lp\varphi_P\rp c^R\lp \varphi_Q\rp\rp T\\
        &\quad +\frac{1}{4}[\varepsilon^{\T,L}_{\theta_{\H,\T}}, c^L\lp\varphi_Q\rp\iota^{\T,R}_Tc^R\lp\varphi_P\rp-c^L\lp\varphi_P\rp\iota^{\T,R}_Tc^R\lp\varphi_Q\rp]T+\O_{\Psi^{\theta}}\lp1\rp\\
        &=-Q^2-P^2+c^L\lp\varphi_Q\varphi_P\rp T-\frac{1}{4}\lp c^L\lp \varphi_Q\rp c^L\lp\varphi_P\rp+c^R\lp\varphi_P\rp c^R\lp \varphi_Q\rp\rp T\\
        &\quad +\frac{(-1)^H}{4}\lp c^L\lp\varphi_Q\rp c^R\lp\varphi_P\rp-c^L\lp\varphi_P\rp c^R\lp\varphi_Q\rp\rp T+\O_{\Psi^{\theta}}\lp1\rp
    \end{align*}
    where $(-1)^H$ is the horizontal form degree operator on $\wedge\F^*$. Once again we run into the same issue, which is that both terms
    \[
        c^L\lp \varphi_Q\rp c^L\lp\varphi_P\rp+c^R\lp\varphi_P\rp c^R\lp \varphi_Q\rp,\quad c^L\lp\varphi_Q\rp c^R\lp\varphi_P\rp-c^L\lp\varphi_P\rp c^R\lp\varphi_Q\rp
    \]
    vanish on odd elements of $\wedge\F^*$, leaving up to highest order the term $-Q^2-P^2+c^L_{\F}\lp \varphi_Q\varphi_P\rp T$ which isn't Rockland.\par\qed
\end{example}
From the above example, we see that we may rule out the Levi-Civita connection on $\wedge\F^*\otimes\wedge\T^*\simeq\wedge T^*M$ as a choice of Clifford compatible connection $\nabla^{\bE,u}$. 
\begin{remark}
    When one selects $\nabla^{\wedge\F^*,u}$ as the Clifford connection and adds in $-\frac{1}{2}\omega^R_u$, then the Dirac operator $D_u$ is the de Rham-Dirac operator $d_M+d_M^{*,u}$ on $\Omega^*\lp M\rp$, where the adjoint is taken relative to $g_{M,u}$. One may be a bit suspicious that the construction does not behave well for this one-parameter family of Dirac operators. However, within the framework of the families index theorem, when one takes the vertical Dirac operators $D^{\V}$ to be the de Rham-Dirac operators of the fibers, then the one-parameter family of Dirac operators $D_u\circlearrowright\Gamma\lp\wedge\V^*\hat{\otimes}\pi^*\wedge T^*B\rp$ are \emph{not} the de Rham-Dirac operators corresponding to the metrics $g_{\Mc,u}$ on $\Mc$. As we have not done the calculation for ourselves, it is unclear how the one-parameter family of de Rham-Dirac operators $d_{\Mc}+d_{\Mc}^{*,u}$ will behave with regards to local index theoretic considerations. 
\end{remark}
\subsection{The Graded Rockland Condition}
One may find it a bit bizarre that the operator $\bD^E$ is \emph{not} Rockland within the Heisenberg calculus but its square $\Deltab^E$ \emph{is}. There is a remedy here which is to use the notion of \emph{graded bundles} and \emph{graded symbol calculus} as described in \cite{DH}. If one is equipped with a graded vector bundle $G=\bigoplus _{j=0}^kG^k\rightarrow M$, one can express a differential operator $L:\Gamma\lp G\rp\rightarrow\Gamma\lp G\rp$ in matrix form as
\[
    L=\begin{bmatrix}
        L_{00} & \cdots & L_{0k} \\
        \vdots & & \vdots \\
        L_{k0} & \cdots & L_{kk}
    \end{bmatrix}
\]
where $L_{pq}:\Gamma\lp G^q\rp\rightarrow \Gamma\lp G^p\rp$. The operator $L$ is said to have \textbf{\emph{graded Heisenberg order $l$}} if each $L_{pq}$ has Heisenberg order $l+p-q$ and the \textbf{\emph{graded Heisenberg cosymbol of $L$ at $x\in M$}} is defined to be the matrix of co-symbols
\[
    \widecheck{\sigma}^G_l\lp L\rp_x:=\begin{bmatrix}
        \widecheck{\sigma}_l\lp L_{00}\rp_x &\cdots & \widecheck{\sigma}_{l-k}\lp L_{0k}\rp_x \\
        \vdots & & \vdots \\
        \widecheck{\sigma}_{l+k}\lp L_{k0}\rp_x & \cdots & \widecheck{\sigma}_l\lp L_{kk}\rp_x
    \end{bmatrix}.
\]
Note that at each $x\in M$ the graded cosymbol gives an element $\widecheck{\sigma}^G_l\lp L\rp_x\in\U\lp\t_{\F}M_x\rp\otimes \End G_x$. Moreover the graded Heiseneberg order satisfies the important composition property: If $L_1$ has graded Heisenberg order $l_1$ and $L_2$ has graded Heisenberg order $l_2$ then the composition $L_1L_2$ has graded Heisenberg order $l_1+l_2$ and at each $x\in M$ the cosymbol satisfies
\[
    \widecheck{\sigma}^G_{l_1+l_2}\lp L_1L_2\rp_x=\widecheck{\sigma}^G_{l_1}\lp L_1\rp_x\widecheck{\sigma}^G_{l_2}\lp L_2\rp_x
\]
where when taking products of cosymbols we are using the tensor product of the natural product structures on $\U\lp\t_{\F}M_x\rp$ and $\End G_x$.\par
For a given irreducible unitary representation $\Pi_x:T_{\F}M_x\rightarrow U\lp H\rp$ one can inflate by $G_x$ and obtain an operator
\[
    \pi_x^{G_x}\lp \widecheck{\sigma}^G_l\lp L\rp_x\rp:H^{\infty}\otimes G_x\rightarrow H^{\infty}\otimes G_x.
\]
\begin{theorem}[Dave-Haller \cite{DH}]
\label{GRC}
    Suppose that $\lp M,\F\rp$ is a $m$-step filtered manifold and $G=\bigoplus_{j=1}^kG^j\rightarrow M$ a graded complex vector bundle with differential operator $L$ of graded Heisenberg order $l$. Then $L$ is hypoelliptic if for every $x\in M$ and every irreducible unitary representation $\Pi_x:T_{\F}M_x\rightarrow U(H)$, the induced operator on smooth vectors of the inflated representation $\pi^{E_x}_x:\mathfrak{t}_{\F}M_x\otimes \gl\lp E_x\rp\rightarrow \gl\lp H^{\infty}\otimes E_x\rp$ 
    \[
        \pi^{E_x}_x\lp\widecheck{\sigma}^G_l\lp L\rp_x\rp:H^{\infty}\otimes E_x\rightarrow H^{\infty}\otimes E_x
    \]
    is injective. If $M$ is closed, then $\Ker L$ is a finite dimensional space of smooth sections.
\end{theorem}
It will be a benefit for us to slightly generalize the setting from contact manifolds $\lp M,\theta\rp$ to filtered manifolds $\lp M,\F\rp$ for which $\Gamma\lp\F\rp+[\Gamma\lp\F\rp,\Gamma\lp\F\rp]=\Gamma\lp TM\rp$, i.e. $2$-step filtrations $\F$ of $TM$. For a given $\bZ_2$-graded covariant Hermitian Clifford module
\[
    \lp C\ell\lp\F^*,g^*_{\F}\rp,\nabla^{\F,0}\rp\circlearrowright\lp E,h_E,\nabla^E\rp
\]
we assign elements of $E\otimes \wedge^k\T^*$ to have degree $k$. At a given $x\in M$ the osculating Lie algebra $\t_{\F}M_x$ is a two-step nilpotent Lie algebra $\t_{\F}M_x=\t^1_{\F}M_x\oplus\t^2_{\F}M_x$ and the splitting $TM=\F\oplus\T$ allows us to make the identifications $\F_x\simeq\t^1_{\F}M_x$ and $\T_x\simeq\t^2_{\F}M_x$ and equip $\t^1_{\F}M_x$ and $\t^2_{\F}M_x$ with the metrics $g_{\F_x}$ and $g_{\T_x}$, respectively. One can view $T_{\F}M_x$ as a manifold in its own right and use the subspaces $\t^1_{\F}M_x$ and $\t^2_{\F}M_x$, along with their metrics $g_{\F_x}$ and $g_{\T_x}$, to generate left-invariant transverse Euclidean distributions $TT_{\F}M_x=\t^1_{\F}M^L_x\oplus \t^2_{\F}M_x^L$. One can compute that under a left-invariant global framing of $\t^1_{\F}M^L_x$ and $\t^2_{\F}M^L_x$ the connections $\nabla^{\F,0}$ and $\nabla^{\T,0}$ will be the canonical flat connections. We can left translate the Clifford action 
\[
    C\ell\lp \F^*_x,g^*_{\F_x}\rp\circlearrowright\lp E_x,h_{E_x}\rp
\]
and denoting $E_x^L\rightarrow T_{\F}M_x$ the trivial bundle, the flat connection $d^{E_x^L}$ gives rise to a $\bZ_2$-graded covariant Hermitian Clifford module
\[
    \lp C\ell\lp\t^1_{\F}M^*_x,g^*_{\t^1_{\F}M^L_x}\rp,\nabla^{\F,0}\rp\circlearrowright\lp E_x^L,h_{E_x},d^{E^L_x}\rp
\]
and thus we obtain an associated asymptotic Bismut superconnection which we denote by $\bD^{E^L_x}$.
\begin{proposition}
\label{LFD}
    At a given $x\in M$, let $\{e_i\}_{i=1}^{n_1}$ be an orthonormal basis of $\lp\F_x,g_{\F_x}\rp$ with dual framing $\{\varphi^i\}_{i=1}^{n_1}$ and left translation $\{ E_i\}_{i=1}^{n_1}\subseteq\Gamma^L\lp \t^1_{\F}M_x^L\rp$, and let $\{f_{\mu}\}_{\mu=1}^{n_2}$ be an orthonormal basis of $\lp\T_x,g_{\T_x}\rp$ with dual framing $\{\psi^{\mu}\}_{\mu=1}^{n_2}$ and left translation $\{ F_{\mu}\}_{\mu=1}^{n_2}\subseteq \Gamma\lp\t^2_{\F}M_x^L\rp$. Then the asymptotic Bismut superconnection $\bD^{E^L_x}$ acting on $C^{\infty}\lp T_{\F}M_x,E_x\otimes\wedge\T^*_x\rp$ is given by
    \[
        \bD^{E^L_x}=c\lp\varphi^i\rp E_i+\varepsilon^{\T_x}_{\psi^{\mu}}F_{\mu}+\frac{1}{4}\iota^{\T_x}_{f_{\mu}}c\lp\iota_{F_{\mu}}\Omega_{\t^1_{\F}M_x^L}\rp
    \]
\end{proposition}
\begin{proof}
    Taking a look at the global form of the asymptotic Bismut superconnection provided in Proposition \ref{GFD}, the terms involving Lie derivatives of $g_{\F_x}$ and $g_{\T_x}$ do not appear since they are left-invariant. The tensor $\Omega_{\t^2_{\F}M^L_x}$ also vanishes due to the two-step nilpotency of $\t_{\F}M_x$. Thus, the global formula takes the form
    \[
        \bD^{E^L_x}=D^{\t^1_{\F}M_x^L\otimes\wedge\t^2_{\F}M^L_x}+d^{\t^2_{\F}M^L_x}_{d^{E^L_x}}+\frac{1}{4}\iota^{\t^2_{\F}M^L_x}\circ c\lp\Omega_{\t^1_{\F}M_x^L}\rp.
    \]
    Taking the global framing as described and using the fact that both $\nabla^{\F,0}$ and $\nabla^{\T,0}$ are flat when trivialized along a left-invariant framing of $\F$ and $\T$ proves the proposition.
\end{proof}
Using the graded Heisenberg calculus of $E\otimes\wedge\T^*$ we are able to capture the correct micro-local model operator of $\bD^E$.
\begin{theorem}
\label{MLMD}
    Suppose that $\lp M,\F,g_{\F}\rp$ is a $2$-step filtered sub-Riemannian manifold and $\lp\T,g_{\T}\rp$ is a Euclidean distribution which is transverse $TM=\F\oplus\T$. Then for a given $\bZ_2$-graded covariant Hermitian Clifford module
    \[
        \lp C\ell\lp\F^*,g^*_{\F}\rp,\nabla^{\F,0}\rp\circlearrowright\lp E,h_E,\nabla^E\rp
    \]
    the asymptotic Bismut superconnection $\bD^E$ has graded Heisenberg order $1$. Moreover, at any $x\in M$ one has
    \[
        \widecheck{\sigma}^G_1\lp\bD^E\rp_x=\bD^{E^L_x}.
    \]
\end{theorem}
\begin{proof}
The global formula of Proposition \ref{GFD} gives us
\[
        \bD^E=D^{\F,E\otimes\wedge\T^*}+d^{\T}_{\nabla^E}+\frac{1}{4}\iota^{\T}\circ c\lp\Omega_{\F}\rp+\frac{1}{4}\varepsilon^{\T}\circ\tr Lg_{\F}-\frac{1}{2}c\circ\varepsilon^{\T}_{\Omega_{\T}}+\frac{1}{4}c\circ\tr Lg_{\T}.
\]
We will analyze each operator separately. Fix a local orthonormal framing $\{e_i\}_{i=1}^{n_1}$ of $\lp\F,g_{\F}\rp$ with dual framing $\{\varphi^i\}_{i=1}^{n_1}$ and an orthonormal framing $\{ f_{\mu}\}_{\mu=1}^{n_2}$ of $\lp\T,g_{\T}\rp$ with dual framing $\{\psi^{\mu}\}_{\mu=1}^{n_2}$ and note that $D^{\F,E\otimes\wedge\T^*}$ takes the local form
\[
    D^{\F,E\otimes\wedge\T^*}=c\lp\varphi^i\rp\lp\nabla^E_{e_i}\otimes I+I\otimes\nabla^{\T,0}_{e_i}\rp.
\]
All operators in the above expansion preserve the grading degree of $E\otimes\wedge\T^*$ and the differential operator components have Heisenberg order $1$ whereas the $c\lp\varphi^i\rp$ have Heisenberg order $0$. Thus $D^{\F,E\otimes\wedge\T^*}$ has graded Heisenberg order $1$ and at $x\in M$ one has
\[
    \widecheck{\sigma}^G_1\lp D^{\F,E\otimes\wedge\T^*}\rp_x=c\lp \varphi^i_x\rp [e_i]_x
\]
where $[e_i]_x\in\t^1_{\F}M_x$ denotes the projection of $e_i$. But when identifying $[e_i]_x$ with the induced left-invariant vector field on $T_{\F}M_x$ we obtain
\[
    \widecheck{\sigma}^G_1\lp D^{\F,E\otimes\wedge\T^*}\rp_x=c\lp\varphi^i_x\rp E_i.
\]
The operator $d^{\T}_{\nabla^E}$ takes the local form
\[
    d^{\T}_{\nabla^E}=\varepsilon^{\T}_{\psi^{\mu}}\lp\nabla^E_{f_{\mu}}\otimes I+I\otimes\nabla^{\T,0}_{f_{\mu}}\rp.
\]
For each $0\leq k\leq n_2$, one has
\[
    d^{\T}_{\nabla^E}\Gamma\lp E\otimes\wedge^k\T^*\rp\subseteq\Gamma\lp E\otimes\wedge^{k+1}\T^*\rp
\]
and each covariant derivative operator in $d^{\T}_{\nabla^E}$ has Heisenberg order $2$. Thus since $d^{\T}_{\nabla^E}$ has Heisenberg differential order $2$ and shifts the grading degree of each factor within $E\otimes\wedge\T^*$ by $1$ we see that $d^{\T}_{\nabla^E}$ has graded Heisenberg order $1$ and that its graded Heisenberg symbol at $x\in M$ is
\[
    \widecheck{\sigma}^G_1\lp d^{\T}_{\nabla^E}\rp_x=\varepsilon^{\T_x}_{\psi_x^{\mu}}[f_{\mu}]_x
\]
where $[f_{\mu}]_x\in\t^2_{\F}M_x$ denotes the projection of $f_{\mu}$. But when identifying $[f_{\mu}]_x$ with the induced left-invariant vector field on $T_{\F}M_x$ we obtain
\[
    \widecheck{\sigma}\lp d^{\T}_{\nabla^E}\rp_x=\varepsilon^{\T_x}_{\psi^{\mu}_x}F_{\mu}.
\]
The operator $\iota^{\T}\circ c\lp\Omega_{\F}\rp$ takes the local form
\[
    \iota^{\T}\circ c\lp\Omega_{\F}\rp=\iota^{\T}_{f_{\mu}}\frac{1}{2}c\lp\varphi^i\rp c\lp\varphi^j\rp\iota_{f_{\mu}}\Omega_{\F}\lp e_i,e_j\rp.
\]
which demonstrates that for each $0\leq k\leq n_2$ one has
\[
    \iota^{\T}\circ c\lp \Omega_{\F}\rp\Gamma\lp E\otimes\wedge^k\T^*\rp\subseteq\Gamma\lp E\otimes\wedge^{k-1}\T^*\rp
\]
and all operators within the local form of $\iota^{\T}\circ c\lp\Omega_{\F}\rp$ are have Heisenberg order $0$. Thus since $\iota^{\T}\circ c\lp\Omega_{\F}\rp$ has Heisenberg order $0$ and shifts the grading degree of each factor within $E\otimes\wedge\T^*$ by $-1$ we see that $\iota^{\T}\circ c\lp\Omega_{\F}\rp$ has graded Heisenberg order $1$ and that its graded Heisenberg symbol at $x\in M$ is
\[
    \widecheck{\sigma}^G_1\lp\iota^{\T}\circ c\lp\Omega_{\F}\rp\rp_x=\iota^{\T_x}_{f_{\mu,x}}\frac{1}{2}c\lp\varphi^i_x\rp c\lp\varphi^j_x\rp\iota_{f_{\mu,x}}\Omega_{\F}\lp e_{i,x},e_{j,x}\rp.
\]
However, by definition one has that
\begin{align*}
    \iota_{f_{\mu,x}}\Omega_{\F}\lp e_{i,x},e_{j,x}\rp&=g_{\T_x}\lp f_{\mu,x},[e_j,e_i]^{\T}_x\rp\\
    &=g_{\t^2_{\F}M_x}\lp f_{\mu,x},[e_{j,x},e_{i,x}]\rp\\
    &=\iota_{f_{\mu,x}}\Omega_{\t^1_{\F}M_x}\lp e_{i,x},e_{j,x}\rp
\end{align*}
where the bracket in the second line is the Lie bracket within the Lie algebra $\t_{\F}M_x$. Lastly, the operators $\varepsilon^{\T}\circ \tr Lg_{\F}$, $c\circ\varepsilon^{\T}_{\Omega_{\T}}$, and $c\circ\tr L g_{\T}$ are all zeroth order Heisenberg differential operators which either preserve the degree or positively shift the degree of $E\otimes\wedge\T^*$. Therefore their graded Heisenberg symbol of degree $1$ vanishes. Putting all of the calculations together we have found
\begin{align*}
    \widecheck{\sigma}^G_1\lp\bD^E\rp_x&=\widecheck{\sigma}^G_1\lp D^{\F,E\otimes\wedge\T^*}+d^{\T}_{\nabla^E}+\frac{1}{4}\iota^{\T}\circ c\lp\Omega_{\F}\rp\rp_x\\
    &=c\lp\varphi^i_x\rp E_i+\varepsilon^{\T_x}_{\psi^{\mu}_x}F_{\mu}+\frac{1}{4}\iota^{\T_x}_{f_{\mu}}c\lp\iota_{F_{\mu}}\Omega_{\t^1_{\F}M_x^L}\rp\\
    &=\bD^{E^L_x}.
\end{align*}
\end{proof}
The question of whether or not $\bD^E$ is a graded Rockland operator now becomes a question of how the simpler operator $\bD^{E^L_x}$ behaves on the simply connected nilpotent Lie group $T_{\F}M_x$ for each $x\in M$. The fact that we have assumed that the manifold $\lp M,\F\rp$ is two-step filtered implies that for each $x\in M$ the osculating Lie group $T_{\F}M_x$ will be a two-step graded simply connected nilpotent Lie group, and as $x\in M$ varies, the group structure of $T_{\F}M_x$ can vary. However, in the two-step case this variation can be simply understood. First observe that at each $x\in M$ the bracket
\[
    B_x:\t^1_{\F}M_x\times \t^1_{\F}M_x\rightarrow\t^2_{\F}M_x
\]
induces a linear map $B_x:\wedge^2\t^1_{\F}M_x\rightarrow\t^2_{\F}M_x$ which must be \emph{surjective} due to the two-step non-integrability of $\F$, and in fact the smooth variation of the Lie algebraic structure of $\t_{\F}M$ is determined entirely by the smooth variation of the field of bilinear forms $B\in\Gamma\lp\wedge^2\t^1_{\F}M^*\otimes\t^2_{\F}M\rp$. Using Kirillov's orbit method one can compute all irreducible unitary representations of $T_{\F}M_x$ for each $x\in M$ and we record the representations below. To set up notation, if $W$ is a finite dimensional $\bR$ vector space and $w\in W$ we let $E_w:W^*\rightarrow \bR$ be the evaluation functional on $W^*$ and we let $\partial{\tau}$ denote the directional derivative operator on $C^{\infty}\lp W^*,\bR\rp$ determined by the vector $\tau\in W^*$.
\begin{proposition}
    \label{TNG}
    L $\n=\n^1\oplus\n^2$ be a two-step nilpotent Lie algebra, let 
    \[
        B:\n^1\times\n^1\rightarrow\n^2
    \]
    be the restriction of the Lie bracket on $\n$ to the subspace $\n^1$, and suppose that the induced linear map $B:\wedge^2\n^2\rightarrow\n^2$ is surjective. Then the irreducible unitary representations of the simply connected Lie group $N=\exp\n$ fall into two classes.
    \begin{itemize}
        \item[(I)] For a non-zero $\eta\in\n^{1,*}$ one has the representation
        \begin{align*}
            \pi_{\eta}:\n&\rightarrow \bC\\
            \lp x,a\rp&\mapsto \eta\lp x\rp
        \end{align*}
        and these exhaust all possible finite dimensional irreducible unitary representations of $N$.
        \item[(II)] For a non-zero $\tau\in\n^{2,*}$ the form 
        \[
            \tau B:\n^1\times\n^1\rightarrow\bR
        \]
        gives a splitting $\n^1=I_{\tau B}\oplus Q_{\tau B}\oplus P_{\tau B}$ for which $\tau B\big|_{I_{\tau B}}\equiv 0$, $\tau B\big|_{Q_{\tau B}\oplus P_{\tau B}}$ is symplectic, and $Q_{\tau B}$ is Lagrangian. Let $\eta\in I^*_{\tau B}$ and $\S\lp Q^*_{\tau B}\rp$ denote the Schwartz class functions on $Q^*_{\tau B}$. Expressing a vector $x\in\n^1$ as $x=n+q+p$ consistent with the splitting $\n^1=I_{\tau B}\oplus Q_{\tau B}\oplus P_{\tau B}$, one has the representation
        \begin{align*}
            \pi_{\eta+\tau}:\n&\rightarrow \gl\lp \S\lp Q^*_{\tau B}\rp\rp\\
            \lp n+q+p,a\rp&\mapsto i\eta\lp n\rp+ \frac{1}{i}E_q+\partial_{\tau B\lp\cdot,p\rp}+i\tau\lp a\rp
        \end{align*}
        and these exhaust all possible infinite dimensional unitary representations of $N$.
    \end{itemize}
\end{proposition}
Combining Theorem \ref{MLMD} with Proposition \ref{TNG} we have a method of checking the pointwise graded Rockland condition and apply Theorem \ref{GRC}.
\begin{theorem}
\label{TSRC}
    Suppose that $\lp M,\F,g_{\F}\rp$ is a two-step filtered sub-Riemannian manifold and $\lp\T,g_{\T}\rp$ is a Euclidean distribution which is transverse $TM=\F\oplus\T$. Then for a given $\bZ_2$-graded covariant Hermitian Clifford module
    \[
        \lp C\ell\lp\F^*,g^*_{\F}\rp,\nabla^{\F,0}\rp\circlearrowright\lp E,h_E,\nabla^E\rp
    \]
    the asymptotic Bismut superconnection $\bD^E$ is a graded Rockland operator of order $1$. Moreover, if $M$ is closed then $\Ker\bD^E$ is a finite dimensional subspace of $\Gamma\lp E\otimes\wedge\T^*\rp$.
\end{theorem}
\begin{proof}
    By Theorem \ref{MLMD} we need only check the graded Rockland condition for the operators $\bD^{E^L_x}$ for each $x\in M$. For a non-zero $\eta\in\t^1_{\F}M^*_x$ with $g_{\F_x}$-dual $n\in\t^1_{\F}M_x$, we let $\pi_{\eta}$ denote the corresponding irreducible representation of $\t_{\F}M_x$ as given in Proposition \ref{TNG}. If we inflate the representation by $E_x\otimes\wedge\T^*_x$, then the global formula in Proposition \ref{GFD} evaluates
    \[
        \pi^{E_x\otimes\wedge\T^*_x}_{\eta}\lp\bD^{E^L_x}\rp=c\lp n\rp+\frac{1}{4}\iota^{\T}\circ c\lp \Omega_{\F_x}\rp\circlearrowright E_x\otimes\wedge\T^*_x.
    \]
    The operator $c\lp n\rp$ is invertible and the operator $\iota^{\T}\circ c\lp\Omega_{\F_x}\rp$ is nilpotent and the product $c\lp n\rp\iota^{\T}\circ c\lp\Omega_{\F_x}\rp$ is nilpotent, thus the sum is invertible for any non-zero $\eta\in\t^1_{\F}M_x^*$.\par
    For a given $\tau\in\t^2_{\F}M_x^*$ with $g_{\T_x}$-dual $t\in\t^2_{\F}M_x$, note that
    \[
        \tau B_x=\iota_t\Omega_{\F_x}.
    \]
    Let $\{ n_j\}_{j=1}^m$, $\{ q_k,p_k\}_{k=1}^{\lp n_1-m\rp/2}$ be a $g_{\F_x}$-orthonormal basis which diagonalizes the form $\iota_t\Omega_{\F_x}$
    \[
        \iota_t\Omega_{\F_x}\big|_{\text{Span}\{ n_j\}}\equiv 0,\quad \iota_t\Omega_{\F_x}\lp q_k,p_l\rp=\delta_{kl}\alpha_k,\quad \alpha_k>0.
    \]
    By defining $I_{\tau B}:=\text{Span}\{ n_j\}_{j=1}^m$, $Q_{\tau B}:=\text{Span}\{ q_k\}_{k=1}^{\lp n_1-m\rp/2}$, and $P_{\tau B}:=\text{Span}\{ p_k\}_{k=1}^{\lp n_1-m\rp/2}$ we obtain a splitting
    \[
        \t^1_{\F}M_x=I_{\tau B}\oplus Q_{\tau B}\oplus P_{\tau B}
    \]
    which is consistent with the splitting required of $\t^1_{\F}M_x$ in Proposition \ref{TNG} (II). Let $\eta\in I^*_{\tau B}$ and let $\pi_{\eta+\tau}:\t_{\F}M_x\rightarrow \S\lp Q^*_{\tau B}\rp$ denote the corresponding irreducible representation as given in Proposition \ref{TNG}. Let $\{ N_j\}_{j=1}^m$ and $\{ Q_k,P_k\}_{k=1}^{\lp n_1-m\rp/2}$ denote the induced left-invariant of $\t^1_{\F}M^L_x$ and choose an arbitrary left-invariant orthonormal framing $\{ F_{\mu}\}_{\mu=1}^{n_2}$ of $\t^2_{\F}M^L_x$. Then by Proposition \ref{LFD} one has 
    \[
        \bD^{E^L_x}=c\lp\varphi_{n_j}\rp N_j+c\lp\varphi_{q_k}\rp Q_k+c\lp\varphi_{p_l}\rp P_l+\varepsilon^{\T_x}_{\psi^{\mu}}F_{\mu}+\frac{1}{4}\iota^{\T_x}_{f_{\mu}}c\lp\iota_{F_{\mu}}\Omega_{\t^1_{\F}M_x^L}\rp
    \]
    Denoting $\Deltab^{E^L_x}=\lp\bD^{E^L_x}\rp^2$ we have
    \begin{align*}
        &\Deltab^{E^L_x}=-\delta^{ij}\lp N_iN_j+Q_iQ_j+P_iP_j\rp +c\lp\varphi_{n_j}\rp c\lp \varphi_{q_k}\rp[N_j,Q_k]+c\lp\varphi_{n_j}\rp c\lp\varphi_{p_l}\rp[N_j,P_l]\\
        &\hspace{5cm}+c\lp\varphi_{q_k}\rp c\lp\varphi_{p_l}\rp[Q_k,P_l]+\frac{1}{4}F_{\mu}c\lp\iota_{F_{\mu}}\Omega_{\t^1_{\F}M^L_x}\rp.
    \end{align*}
    Note that one has
    \begin{align*}
        [Q_k,P_l]&=F_{\mu}g_{\T_x}\lp F_{\mu},[Q_k,P_l]\rp\\
        &=-F_{\mu}\iota_{F_{\mu}}\Omega_{\t^1_{\F}M^L_x}\lp Q_k,P_l\rp.
    \end{align*}
    thus giving the evaluation
    \begin{align*}
        &\pi^{E\otimes\wedge\T^*_x}_{\eta+\tau}\lp \Deltab^{E^L_x}\rp=\delta^{ij}\eta\lp N_i\rp\eta\lp N_j\rp+\delta^{kl}\lp- \partial_{\beta B\lp\cdot,p_k\rp}\partial_{\beta B\lp\cdot,p_l\rp}+E_{q_k}E_{q_l}\rp-\frac{3}{4} c\lp\iota_t\Omega_{\F_x}\rp\\
        &=\delta^{ij}\eta\lp N_i\rp\eta\lp N_j\rp+\delta^{kl}\lp- \partial_{\beta B\lp\cdot,p_k\rp}\partial_{\beta B\lp\cdot,p_l\rp}+E_{q_k}E_{q_l}+\frac{3i}{4}\alpha_k c\lp\varphi_{q_k}\rp c\lp\varphi_{p_l}\rp\rp\circlearrowright\S\lp Q^*\rp\otimes E_x\otimes\wedge\T^*_x.
    \end{align*}
    The same analysis in Theorem \ref{OGH} demonstrates that each operator in the above sum is formally self-adjoint, and positive with explicit lower bound 
    \[
        -\partial^2_{\beta B\lp\cdot,p_k\rp}+E^2_{q_k}+\frac{3i}{4}\alpha_kc\lp\varphi_{q_k}\rp c\lp\varphi_{p_k}\rp\geq \frac{\alpha_k}{4}>0
    \]
    hence one has that the lower bound
    \[
        \pi^{E_x\otimes\wedge\T^*_x}_{\eta+\tau}\lp\Deltab^{E^L_x} \rp\geq \delta^{ij}\eta\lp N_i\rp\eta\lp N_j\rp+\frac{\sum_{j=1}^{n_2-m}\alpha_j}{4}>0
    \]
    which by the composition rule for graded Heisenberg symbols
    \[
        \pi^{E_x\otimes\wedge\T^*_x}_{\eta+\tau}\lp\Deltab^{E^L_x} \rp=\lp\pi^{E_x\otimes\wedge\T^*_x}_{\eta+\tau}\lp\bD^{E^L_x}\rp\rp^2
    \]
    we see that $\pi^{E_x\otimes\wedge\T^*_x}_{\beta}\lp\bD^{E^L_x}\rp$ is injective. Since we have verified injectivity of $\pi^{E_x\otimes\wedge\T^*_x}_x\lp\bD^{E^L_x}\rp$ for every irreducible unitary representation $\pi_x$ of $T_{\F}M_x$ at every $x\in M$, we can apply Theorem \ref{GRC} which implies that $\bD^E$ is a graded Rockland operator of Heisenberg order $1$.
\end{proof}
The above theorem can be applied in the contact manifold setting $\lp M,\theta\rp$ as well as the more general poly-contact setting given in \cite{Van2}. In the presence of a Lie group $G$ acting on $M$ and preserving the geometric structures $\lp\F,g_{\F}\rp$ and $\lp\T,g_{\T}\rp$ a direct application of Theorem \ref{TSRC} above in the setting of two-step filtered manifolds gives the following theorem.
\begin{theorem}
\label{GET}
    Suppose a Lie group $G$ acts simultaneously by isometries on a two-step subRiemannian manifold $\lp M,\F,g_{\F}\rp$ and also prserves a Euclidean distribution $\lp\T,g_{\T}\rp$ which is transvers $TM=\F\oplus\T$. Then for a given $G$-equivariant $\bZ_2$-graded covariant Hermitian Clifford module
    \[
        G\circlearrowright\lp C\ell\lp\F^*,g^*_{\F}\rp,\nabla^{\F,0}\rp\circlearrowright G\circlearrowright\lp E,h_E,\nabla^E\rp
    \]
    we have that $\bD^{E}$ is a $G$-equivariant hypoelliptic operator and if $M$ is closed the kernel is a finite dimensional $\bZ_2$-graded representations of $G$ 
    \[
        G\circlearrowright\Ker\bD^{E}
    \]
\end{theorem}
\subsection{The General Case}
We now call into question the generality of our construction, namely whether $\bD^E$ is hypoelliptic on general filtered manifolds with $m$-step filtrations
\[
    \F=\F^1\leqslant\F^2\leqslant\cdots\leqslant\F^m=TM
\]
satisfying $[\Gamma\lp\F^i\rp,\Gamma\lp\F^j\rp]\subseteq\Gamma\lp\F^{i+j}\rp$. If one were to take an arbitrary Euclidean distribution $\lp\T,g_{\T}\rp$ which is transverse $\F^1\oplus\T=TM$, then we would immediately begin to have issues with the Heisenberg order of $\Deltab^E$ as it would contain derivative terms in the $\F^2,\cdots,\F^m$-directions and each $\F^k$ derivative has order at least $k$. The notion of graded Heisenberg order in the case of two-step filtrations was able to provide the correct micro-local model of $\bD^E$ as seen in Theorem \ref{MLMD}. Thus we first seek to put a grading on $\wedge\T^*$ that will force $\bD^E$ to have Heisenberg order $1$. There is a bit of experimentation that one needs to do in order to arrive at a definition of $\bD^E$ which guarantees a graded Heisenberg differential operator of order $1$. What we have found is a working definition, but it is by no means the only approach to defining $\bD^E$ in the general case.\par
In the case of a general filtration of $TM$, we can choose a sequence of transverse distributions 
\[
    \F^j=\F^{j-1}\oplus\T^{j},\quad \T:=\bigoplus_{j=2}^{m}\T^j\quad\Rightarrow\quad TM=\F\oplus\T
\]
along with a sequence of transverse metrics $g_{\T^j}$ for each factor $\T^j$ and define the full transverse metric via
\[
    g_{\T}:=g_{\T^2}\oplus\cdots\oplus g_{\T^m}.
\]
We obtain the global isomorphism of vector bundles 
\[
    \wedge\T^*\simeq \wedge\T^{2,*}\otimes\cdots\wedge\T^{m,*}.
\]
For each factor $\wedge\T^{j,*}$ we assign the grading 
\[
    |\wedge^{k_j}\T^{j,*}|:=(j-1)\cdot k_j,\quad 0\leq k_j\leq j,\quad 2\leq j\leq m.
\]
We will apply a new adiabatic scaling in this case and define
\[
    g_{M,u}:=g_{\F}\oplus\bigoplus_{j=2}^{m}u^{-j+1}g_{\T^j}.
\]
We will need a small generalization of our curvature forms given above. For $1\leq i\neq j, k\leq m$ we define
\[
    \iota_X\Omega^{\T^i}_{\T^j,\T^k}\lp U_j,V_k\rp:=g_{\T^i}\lp X^{\T^i},[V_k,U_j]^{\T^i}\rp,\quad X\in\Gamma\lp TM\rp,\quad U_j\in\Gamma\lp\T^j\rp,\quad V_k\in\Gamma\lp\T^k\rp
\]
where we are taking $\T^1=\F$.
The connection $\nabla^{\F,u}$ is defined as projection of the Levi-Civita connection to $\F$, and the relation
\[
    \nabla^{\F,u}=\nabla^{\F,0}+\frac{u^{-1}}{2}\so\Omega^{\T^2}_{\F,\F}
\]
remains valid. The reason for restricting the above brackets to $\T^2$ is due to the filtration condition $[\Gamma\lp\F^1\rp,\Gamma\lp\F^1\rp]\subseteq\Gamma\lp\F^2\rp$. Thus the only interesting $\T$ component that a bracket of vector fields tangent to $\F^1$ can have is in the $\T^2$-component. In the general case of $m$-step filtered manifolds we will take $\nabla^{\T,u}$ to be defined as
\[
    \nabla^{\T,u}:=\bigoplus_{j=1}^{m-1}P^{\T^j}\nabla^{g_{M,u}}.
\]
This connection will have the desired property of preserving the filtration degree on $\Gamma\lp\wedge\T^*\rp$, however the relationship between $\nabla^{\T,u}$ and $\nabla^{\T,0}$ will be slightly more complicated. We define $\Omega_{\T^j}\in\Omega^1\lp M,\wedge^2\T^{j,*}\rp$ by
\[
    \Omega_{\T^j,u}:=\frac{u^{j-1}}{2}\Omega^{\F}_{\T^j,\T^j}+\sum_{j\neq k}\frac{u^{j-k}}{2}\Omega^{\T^k}_{\T^j,\T^j}
\]
and using $g_{\T^j}$ we obtain an element $\so\Omega_{\T^j,u}\in\Omega^1\lp M,\so\lp\T^j,g_{\T^j}\rp\rp$. Observe that all terms in $\so\Omega_{\T^j,u}$ come with a non-zero power of $u$.

\begin{proposition}
    There is a $u$-independent $g_{\T^j}$-compatible connection $\nabla^{\T^j,0}$ for which
    \[
        \nabla^{\T^j,u}=\nabla^{\T^j,0}+\so\Omega_{\T^j,u}
    \]
\end{proposition}
\begin{proof}
    This follows directly by the Koszul identity. 
\end{proof}
We now study the dual of the shape form
\[
    \nabla^{g_{M,u}}=\nabla^{\F,u}\oplus\bigoplus_{j=2}^{m}\nabla^{\T^j,u}+S_u
\]
given at $u>0$ by
\[
    \iota_X\omega_u\lp Y,Z\rp:=g_{M,u}\lp \iota_XS_uY,Z\rp,\quad X,Y,Z\in\Gamma\lp TM\rp.
\]
Firstly observe that by definition one has
\[
    \omega_u\big|_{\wedge^2\F}\equiv0,\quad \omega_u\big|_{\wedge^2\T^{j}}\equiv 0,\quad 2\leq j\leq m.
\]
The Koszul identity gives the following proposition.
\begin{proposition}
    For $T_k\in\Gamma\lp \T^k\rp$, $U_l\in\Gamma\lp\T^l\rp$, and $V_j\in\Gamma\lp \T^j\rp$ with $l\neq j$ and $1\leq j,k,l\leq m$ one has 
    \begin{itemize}
        \item For $k=l$ the dual shape-form is
        \[
            \iota_{T_k}\omega_u\lp U_k,V_l\rp=-\frac{u^{-l+1}}{2}L_{V_j}g_{\T^l}\lp T_l,U_l\rp-\frac{u^{-j+1}}{2}\iota_{V_j}\Omega^{\T^j}_{\T^l,\T^l}\lp T_l,U_l\rp
        \]
        \item For $k\neq l\neq j$ the dual shape-form is
        \[
            \iota_{T_k}\omega_u\lp U_l,V_j\rp=-\frac{u^{-l+1}}{2}\iota_{U_l}\Omega^{\T^l}_{\T^j,\T^k}\lp V_j,T_k\rp+\frac{u^{-k+1}}{2}\iota_{T_k}\Omega^{\T^k}_{\T^l,\T^j}\lp U_l,V_j\rp-\frac{u^{-j+1}}{2}\iota_{V_j}\Omega^{\T^j}_{\T^k,\T^l}\lp T_k,U_l\rp
        \]
    \end{itemize}
\end{proposition}
For a given $\bZ_2$-graded covariant Hermitian Clifford module
\[
    \lp C\ell\lp\F^*,g^*_{\F}\rp\rp\circlearrowright\lp E,h_E,\nabla^E\rp
\]
we take 
\[
    \nabla^{E,u}:=\nabla^E+\frac{u^{-1}}{4}c\lp\Omega^{\T^2}_{\F,\F}\rp
\]
and define the asymptotic Bismut superconnection associated to the one-parameter family of connection $\nabla^{\T,u}$ constructed above as
\[
    \bD^E:=\lim^{F.P.}_{u\rightarrow 0^+}m_u\lp\nabla^{E,u}\otimes I+I\otimes \nabla^{\T,u}+\frac{1}{2}m_u\omega_u\rp.
\]
We assign elements of $E\otimes \bigotimes_{j=2}^m\wedge^{k_j}\T^{j,*}$ degree $k_2+2\cdot k_3+\cdots+(m-1)\cdot k_m$. At a given $x\in M$ the osculating Lie algebra $\t_{\F}M_x$ is a $m$-step nilpotent Lie algebra and the splitting $TM=\F\oplus\bigoplus^m_{j=2}\T^j$ allows us to make the identifications $\F_x\simeq\t^1_{\F}M_x$ and $\T^j_x\simeq \t^j_{\F}M_x$ and equip $\t^j_{\F}M_x$ with the metrics $g_{\F_x}$ and $g_{\T^j_x}$. One can view $T_{\F}M_x$ as a manifold in its own right and use the subspaces $\t^j_{\F}M_x$ along with their metrics to generate left-invariant transverse Euclidean distributions $TT_{\F}M_x=\bigoplus_{j=1}^m\t^j_{\F}M^L_x$. One can compute that under a left-invariant framing of $\t^j_{\F}M^L_x$ the connections $\nabla^{\T^j,0}$ will be the canonical flat connections. We can left translate the Clifford action
\[
    C\ell\lp\T^1,g^*_{\T^1}\rp\circlearrowright\lp E_x,h_{E_x}\rp
\]
and denoting $E^L_x\rightarrow T_{\F}M_x$ the trivial bundle, the flat connection $d^{E^L_x}$ gives rise to a $\bZ_2$-graded covariant Hermitian Clifford module
\[
   \lp C\ell\lp \t^1_{\F}M^L_x,g^*_{\t^1_{\F}M_x^L}\rp,\nabla^{\T^1,0}\rp\circlearrowright\lp E^L_x,h_{E^L_x},d^{E^L_x}\rp
\]
and thus obtain an associated asymptotic Bismut superconnection which we denote by $\bD^{E^l_x}$. We are now in a position to collect the relevant micro-local information regarding this definition of the asymptotic Bismut superconnection $\bD^E$ associated to the one-parameter family of connections $\bigoplus_{j=1}^m\nabla^{\T^j,u}$.
\begin{theorem}
\label{FABS}
    Suppose that $\lp M,\F,g_{\F}\rp$ is an $m$-step filtered subRiemannian manifold and $\lp\T^j,g_{\T^j}\rp$ is a sequence of Euclidean distributions which are transverse $\F^j=\F^{j-1}\oplus\T^j$. Then for a $\bZ_2$-graded covariant Hermitian Clifford module
    \[
        \lp C\ell\lp\F^*,g^*_{\F}\rp,\nabla^{\F,0}\rp\circlearrowright\lp E,h_E,\nabla^E\rp
    \]
    the asymptotic superconnection $\bD^E$ for the one-parameter family of connections $\bigoplus_{j=1}^m\nabla^{\T^j,u}$ has graded Heisenberg order $1$ and at $x\in M$, the graded Heisenberg cosymbol is given by 
    \[
        \widecheck{\sigma}^G_1\lp\bD^E\rp_x=\bD^{E^L_x}.
    \]
\end{theorem}
\begin{proof}
    We split the Clifford contraction into its $\T^j$-components
    \[
        D_u=\sum_{j=1}^mm_u\Big|_{\T^j}\circ \nabla^{\bE,u}
    \]
    and analyze each component separately. For $\T^1=\F$ let $\{ e_i\}_{i=1}^{n_1}$ be a local $g_{\F}$-orthonormal framing and let $\{\varphi^i\}_{i=1}^{n_1}$ denote the dual framing and for each $2\leq j\leq m$ we let $\{f^j_{\mu_j}\}_{\mu_j=1}^{n_j}$ be a local $g_{\T_j}$-orthonormal framing and $\{\psi^{\mu_j}_j\}_{\mu_j=1}^{n_j}$ denote the corresponding dual framing. Then we have
    \[
        m_u\Big|_{\T^1}\circ\nabla^{\bE,u}=c\lp \varphi^i\rp\lp \nabla^{E,u}_{e_i}\otimes I+I\otimes\nabla^{\T,u}_{e_i}\rp+\frac{1}{2}c\lp \varphi^i\rp m_u\lp\iota_{e_i}\omega_u\rp
    \]
    First observe that the $u$-dependent portion of $\nabla^{E,u}$ is completely supported on $\T^1$, hence $\nabla^{E,u}_{e_i}=\nabla^E_{e_i}$. The covariant derivative $\nabla^{\T,u}_{e_i}$ will contain non-trivial $u$-dependent terms, however since we are only interested in the $u$-independent terms, we see that the $u$-independent term is given by $\nabla^{\T,0}_{e_i}$. Moreover, each covariant derivative $\nabla^E_{e_i}$ and $\nabla^{\T,0}_{e_i}$ has Heisenberg differential order $1$ and each $c\lp\varphi^i\rp$ has Heisenberg differential order $0$ and all operators preserve the grading degree of $E\otimes\wedge\T^*$. Thus we have
    \[
        \widecheck{\sigma}^G_1\lp\lim^{F.P.}_{u\rightarrow 0^+}c\lp \varphi^i\rp\lp \nabla^{E,u}_{e_i}\otimes I+I\otimes\nabla^{\T,u}_{e_i}\rp\rp=c\lp\varphi^i_x\rp[e_i]_x
    \]
    We now analyze the $\F$-Clifford contracted term involving $\omega_u$. Note that by the above proposition we have
    \begin{align*}
        &m_u\lp\iota_{e_i}\omega_u\rp=c\lp\varphi^k\rp\lp\varepsilon^{\T^j}_{\psi^{\mu_j}_j}-u^{j-1}\iota^{\T^j}_{f^j_{\mu_j}}\rp\lp-\frac{1}{2}L_{f^j_{\mu_j}}g_{\F}\lp e_i,e_k\rp-\frac{u^{-j+1}}{2}\iota_{f^j_{\mu_j}}\Omega^{\T^j}_{\F,\F}\lp e_i,e_k\rp\rp\\
        &\lp \varepsilon^{\T^l}_{\psi^{\mu_l}_l}-u^{l-1}\iota^{\T^l}_{f^l_{\mu_l}}\rp\lp\varepsilon^{\T^j}_{\psi^{\mu_j}_j}-u^{j-1}\iota^{\T^j}_{f^j_{\mu_j}}\rp \cdot\\
        &\lp -\frac{u^{-l+1}}{2}\iota_{f^l_{\mu_l}}\Omega^{\T^l}_{\T^j,\F}\lp f^j_{\mu_j},e_i\rp+\frac{1}{2}\iota_{e_i}\Omega^{\F}_{\T^l,\T^j}\lp f^l_{\mu_l},f^j_{\mu_j}\rp-\frac{u^{-j+1}}{2}\iota_{f^j_{\mu_j}}\Omega^{\T^j}_{\F,\T^l}\lp e_i,f^l_{\mu_l}\rp\rp
    \end{align*}
    Thus the constant $u$-term is given by
    \begin{align*}
        &m_u\lp\iota_{e_i}\omega_u\rp_0=c\lp\varphi^k\rp\varepsilon^{\T^j}_{\psi^{\mu_j}_j}\lp-\frac{1}{2}L_{f^j_{\mu_j}}g_{\F}\lp e_i,e_k\rp\rp+c\lp\varphi^k\rp\iota^{\T^1}_{f^1_{\mu_1}}\lp\frac{1}{2}\iota_{f^1_{\mu_1}}\Omega^{\T^1}_{\F,\F}\lp e_i,e_k\rp\rp\\
        &+\varepsilon^{\T^l}_{\psi^{\mu_l}_l}\varepsilon^{\T^j}_{\psi^{\mu_j}_j}\frac{1}{2}\Omega^{\F}_{\T^l,\T^j}\lp f^l_{\mu_l},f^j_{\mu_j}\rp+\iota^{\T^l}_{f^l_{\mu_l}}\varepsilon^{\T^l}_{\psi^{\mu_l}_l}\lp\frac{1}{2}\iota_{f^l_{\mu_l}}\Omega^{\T^l}_{\T^j,\F}\lp f^j_{\mu_j},e_i\rp\rp+\varepsilon^{\T^l}_{\psi^{\mu_l}_l}\iota^{\T^j}_{f^j_{\mu_j}}\frac{1}{2}\iota_{f^j_{\mu_j}}\Omega^{\T^j}_{\F,\T^l}\lp e_i,f^l_{\mu_l}\rp.
    \end{align*}
    The operators $\varepsilon^{\T^j}_{\psi^{\mu_j}_j}$ have shift the grading degree by $j-1$ and the operators $\iota^{\T^j}_{f^j_{\mu_j}}$ shift the grading degree by $1-j$. Moreover when $j+1<l$ one has $\Omega^{\T^l}_{\T^j,\F}\equiv 0$ and when $l+1<j$ one has $\Omega^{\T^j}_{\F,\T^l}\equiv 0$. Thus all terms have graded Heisenberg differential order $\leq 1$ and one has
    \begin{align*}
        \widecheck{\sigma}^G_1\lp\lim^{F.P.}_{u\rightarrow0^+} c\lp\varphi^i\rp m_u\lp\iota_{e_i}\omega_u\rp\rp_x&=c\lp\varphi^i_x\rp c\lp\varphi^k_x\rp\iota^{\T^1}_{f^1_{\mu_1,x}}\lp\frac{1}{2}\iota_{f^1_{\mu_1,x}}\Omega^{\T^1}_{\F,\F}\lp e_{i,x},e_{k,x}\rp\rp\\
        &+c\lp\varphi^i_x\rp\iota^{\T^{j+1}}_{f^{j+1}_{\mu_{j+1},x}}\varepsilon^{\T^j}_{\psi^{\mu_j}_{j,x}}\lp\iota_{f^j_{\mu_j,x}}\Omega^{\T^{j+1}}_{\T^j,\F}\lp f^j_{\mu_j,x},e_{i,x}\rp\rp 
    \end{align*}
    We now analyze a generic Clifford contraction restricted to $\T^j$ for $2\leq j\leq m$
    \[
        m_u\Big|_{\T^j}\circ\nabla^{\bE,u}=\lp\varepsilon^{\T^j}_{\psi^{\mu_j}_j}-u^{j-1}\iota^{\T^j}_{f^j_{\mu_j}}\rp\lp\nabla^{E,u}_{f^j_{\mu_j}}\otimes I+I\otimes\nabla^{\T,u}_{f^j_{\mu_j}}\rp+\frac{1}{2}\lp\varepsilon^{\T^j}_{\psi^{\mu_j}_j}-u^{j-1}\iota^{\T^j}_{f^j_{\mu_j}}\rp m_u\lp\iota_{f^j_{\mu_j}}\omega_u\rp
    \]
By definition we have
\[
    \nabla^{E,u}_{f^j_{\mu_j}}=\nabla^E_{f^j_{\mu_j}}+\frac{u^{-1}}{4}c\lp\iota_{f^j_{\mu_j}}\Omega_{\F}\rp
\]
and the singular $u$-term is non-vanishing only when $j=1$. Thus the constant term within the $E$-covariant derivative term is
\[
    \varepsilon^{\T^j}_{f^j_{\mu_j}}\nabla^E_{f^j_{\mu_j}}-\frac{1}{4}\delta_{2j}\iota^{\T^j}_{f^j_{\mu_j}}c\lp \iota_{f^j_{\mu_j}}\Omega^{\T^j}_{\T^1,\T^1}\rp.
\]
As for the $\T$-covariant derivative term, the derivation property of $\nabla^{\T,u}$ on $\wedge\T^*$ implies that it is enough to study the restriction of $\nabla^{\T,u}$ to $\T^2\oplus\cdots\oplus\T^m$ in which case the connection breaks into a sum of the $\nabla^{\T^k,u}$ which allows us to study each $\T^k$-factor separately. For a given $2\leq k\leq m$ we have
\[
    \lp\varepsilon^{\T^j}_{\psi^{\mu_j}_j}-u^{j-1}\iota^{\T^j}_{f^j_{\mu_j}}\rp\nabla^{\T^k,u}_{f^j_{\mu_j}}=\lp\varepsilon^{\T^j}_{\psi^{\mu_j}_j}-u^{j-1}\iota^{\T^j}_{f^j_{\mu_j}}\rp\lp\nabla^{\T^k,0}_{f^j_{\mu_j}}+\sum_{l\neq k}\frac{u^{k-l}}{2}\so\iota_{f^j_{\mu_j}}\Omega^{\T^l}_{\T^k,\T^k}\rp
\]
We see that a $\so\Omega^{\T^l}_{\T^k,\T^k}$ will contribute only if $j-1+k-l=0$. Moreover, $\iota_{f^j_{\mu_j}}\Omega^{\T^l}_{\T^k,\T^k}$ is non-vanishing only if $j=l$. Thus one must have $k=1$ but we are assuming that $2\leq k$, hence none of the $\so\Omega^{\T^l}_{\T^k,\T^k}$ terms make the final cut when $u\rightarrow 0^+$. Hence what we have found so far is that 
\begin{align*}
    &\lim^{F.P.}_{u\rightarrow 0^+}\lp\varepsilon^{\T^j}_{\psi^{\mu_j}_j}-u^{j-1}\iota^{\T^j}_{f^j_{\mu_j}}\rp\lp\nabla^{E,u}\otimes I+I\otimes\nabla^{\T,u}_{f^j_{\mu_j}}\rp=\\
    &\hspace{6cm}\varepsilon^{\T^j}_{f^j_{\mu_j}}\lp\nabla^E_{f^j_{\mu_j}}\otimes I+I\otimes\nabla^{T,0}_{f^j_{\mu_j}}\rp-\frac{1}{4}\delta_{2j}\iota^{\T^j}_{f^j_{\mu_j}}c\lp\iota_{f^j_{\mu_j}}\Omega^{\T^j}_{\T^1,\T^1}\rp
\end{align*}
All terms above have graded Heisenberg order $\leq 1$ with
\[
    \widecheck{\sigma}^G_1\lp \lim_{u\rightarrow 0^+}\lp \varepsilon^{\T^j}_{\psi^{\mu_j}_j}-u^{j-1}\iota^{\T^j}_{f^j_{\mu_j}}\rp\lp\nabla^{E,u}_{f^j_{\mu_j}}\otimes I+I\otimes\nabla^{\T,u}_{f^j_{\mu_j}}\rp\rp_x=\varepsilon^{\T^j}_{\psi^{\mu_j}_{j,x}}[f^j_{\mu_j}]_x-\frac{1}{4}\delta_{2j}\iota^{\T^j}_{f^j_{\mu_j,x}}c\lp\iota_{f^j_{\mu_j,x}}\Omega^{\T^j}_{\T^1,\T^1}\rp 
\]
We now analyze the terms involving $\omega_u$. Observe that $m_u\lp\iota_{f^j_{\mu_j}}\omega_u\rp$ takes the form
\begin{align*}
    &m_u\lp\iota_{f^j_{\mu_j}}\omega_u\rp=c\lp\varphi^k\rp\lp\varepsilon^{\T^j}_{\psi^{\nu_j}_j}-u^{j-1}\iota^{\T^i}_{f^j_{\nu_j}}\rp\lp \frac{u^{-j+1}}{2}L_{e_k}g_{\T^j}\lp f^j_{\nu_j},f^j_{\mu_j}\rp-\frac{1}{2}\iota_{e_k}\Omega^{\T^1}_{\T^j,\T^j}\lp f^j_{\nu_j},f^j_{\mu_j}\rp\rp\\
    &+\lp\varepsilon^{\T^k}_{\psi^{\mu_k}_k}-u^{k-1}\iota^{\T^k}_{f^k_{\mu_k}}\rp\lp\varepsilon^{\T^j}_{\psi^{\nu_j}_j}-\iota^{\T^j}_{f^j_{\nu_j}}\rp\lp \frac{u^{-j+1}}{2}L_{f^k_{\mu_k}}g_{\T^j}\lp f^j_{\nu_j},f^j_{\mu_j}\rp-\frac{u^{-k+1}}{2}\iota_{f^k_{\mu_k}}\Omega^{\T^k}_{\T^j,\T^j}\lp f^j_{\nu_j},f^j_{\mu_j}\rp\rp\\
    &\lp\varepsilon^{\T^k}_{\psi^{\mu_k}_k}-u^{k-1}\iota^{\T^k}_{f^k_{\mu_k}}\rp\lp\varepsilon^{\T^l}_{\psi^{\mu_l}_l}-u^{l-1}\iota^{\T^l}_{f^l_{\mu_l}}\rp_{j\neq k\neq l}\cdot\\
    &\lp -\frac{u^{-k+1}}{2}\iota_{f^k_{\mu_k}}\Omega^{\T^k}_{\T^l,\T^j}\lp f^l_{\mu_l},f^j_{\mu_j}\rp+\frac{u^{-j+1}}{2}\iota_{f^j_{\mu_j}}\Omega^{\T^j}_{\T^k,\T^l}\lp f^k_{\mu_k},f^l_{\mu_l}\rp-\frac{u^{-l+1}}{2}\iota_{f^l_{\mu_l}}\Omega^{\T^l}_{\T^j,\T^k}\lp f^j_{\mu_j},f^k_{\mu_k}\rp\rp\\
    &+c\lp\varphi^k\rp\lp \varepsilon^{\T^l}_{\psi^{\mu_l}_l}-u^{l-1}\iota^{\T^l}_{f^l_{\mu_l}}\rp_{l\neq j}\cdot\\
    &\lp -\frac{1}{2}\iota_{e_i}\Omega^{\F}_{\T^l,\T^j}\lp f^l_{\mu_l},f^j_{\mu_j}\rp+\frac{u^{-j+1}}{2}\iota_{f^j_{\mu_j}}\Omega^{\T^j}_{\F,\T^l}\lp e_k,f^l_{\mu_l}\rp-\frac{u^{-l+1}}{2}\iota_{f^l_{\mu_l}}\Omega^{\T^l}_{\T^j,\F}\lp f^j_{\mu_j},e_k\rp\rp
\end{align*}
where the $\cdot$ term indicates that the last two terms should be multiplied. We analyze each term separately. First note that in
\begin{align*}
    &\lp\varepsilon^{\T^j}_{\psi^{\mu_j}_j}-u^{j-1}\iota^{\T^j}_{f^j_{\mu_j}}\rp c\lp\varphi^k\rp\lp \varepsilon^{\T^j}_{\psi^{\nu_j}_j}-u^{j-1}\iota^{\T^j}_{f^j_{\nu_j}}\rp\cdot\lp \frac{u^{-j+1}}{2}L_{e_k}g_{\T^j}\lp f^j_{\nu_j},f^j_{\mu_j}\rp-\frac{1}{2}\iota_{e_k}\Omega^{\F}_{\T^j,\T^j}\lp f^j_{\nu_j},f^j_{\mu_j}\rp\rp
\end{align*}
all constant $u$-terms have graded Heisenberg order $<1$. Thus none of the above terms make the cut when computing $\widecheck{\sigma}^G_1$. In the second sum we have
\begin{align*}
    &\lp\varepsilon^{\T^j}_{\psi^{\mu_j}_j}-u^{j-1}\iota^{\T^j}_{f^j_{\mu_j}}\rp\lp \varepsilon^{\T^k}_{\psi^{\mu_k}_k}-u^{k-1}\iota^{\T^k}_{f^k_{\mu_k}}\rp\lp\varepsilon^{\T^j}_{\psi^{\nu_j}_j}-u^{j-1}\iota^{\T^j}_{f^j_{\nu_j}}\rp\cdot\\
    &\hspace{5cm}\lp \frac{u^{-j+1}}{2}L_{f^k_{\mu_k}}g_{\T^j}\lp f^j_{\nu_j},f^j_{\mu_j}\rp-\frac{u^{-k+1}}{2}\iota_{f^k_{\mu_k}}\Omega^{\T^k}_{\T^j,\T^j}\lp f^j_{\nu_j},f^j_{\mu_j}\rp\rp
\end{align*}
All constant $u$-terms have graded Heisenberg order $\leq 1$ and after taking $\widecheck{\sigma}^G_1$ one obtains
\[
    \varepsilon^{\T^j}_{\psi^{\mu_j}_{j,x}}\iota^{\T^{2j}}_{f^{2j}_{\mu_{2j},x}}\varepsilon^{\T^j}_{\psi^{\nu_j}_{j,x}}\lp\frac{1}{2}\iota_{f^{2j}_{\mu_{2j},x}}\Omega^{\T^{2j}}_{\T^j,\T^j}\lp f^j_{\nu_j,x},f^j_{\mu_j,x}\rp\rp
\]
The third term takes the form
\begin{align*}
    &\lp\varepsilon^{\T^j}_{\psi^{\mu_j}_j}-u^{j-1}\iota^{\T^j}_{f^j_{\mu_j}}\rp\lp \varepsilon^{\T^k}_{\psi^{\mu_k}_k}-u^{k-1}\iota^{\T^k}_{f^k_{\mu_k}}\rp\lp\varepsilon^{\T^l}_{\psi^{\mu_l}_l}-u^{l-1}\iota^{\T^l}_{f^l_{\mu_l}}\rp_{j\neq k\neq l}\cdot\\
    &\lp -\frac{u^{-k+1}}{2}\iota_{f^k_{\mu_k}}\Omega^{\T^k}_{\T^l,\T^j}\lp f^l_{\mu_l},f^j_{\mu_j}\rp+\frac{u^{-j+1}}{2}\iota_{f^j_{\mu_j}}\Omega^{\T^j}_{\T^k,\T^l}\lp f^k_{\mu_k},f^l_{\mu_l}\rp-\frac{u^{-l+1}}{2}\iota_{f^l_{\mu_l}}\Omega^{\T^l}_{\T^j,\T^k}\lp f^j_{\mu_j},f^k_{\mu_k}\rp\rp
\end{align*}
All constant $u$-terms have graded Heisenberg order $\leq 1$ and after taking $\widecheck{\sigma}^G_1$ one obtains
\begin{align*}
    &\varepsilon^{\T^j}_{\psi^{\mu_j}_{j,x}}\varepsilon^{\T^k}_{\psi^{\mu_k}_{k,x}}\iota^{\T^{j+k}}_{f^{j+k}_{\mu_{j+k},x}}\lp\frac{1}{2}\iota_{f^{j+k}_{\mu_{j+k},x}}\Omega^{\T^{j+k}}_{\T^j,\T^k}\lp f^j_{\mu_j,x},f^k_{\mu_k,x}\rp\rp_{j\neq k}\\
    &\quad-\iota^{\T^j}_{f^j_{\mu_j,x}}\varepsilon^{\T^k}_{\psi^{\mu_k}_{k,x}}\varepsilon^{\T^l}_{\psi^{\mu_l}_{l,x}}\lp\frac{1}{2}\iota_{f^j_{\mu_j,x}}\Omega^{\T^j}_{\T^k,\T^l}\lp f^k_{\mu_k,x},f^l_{\mu_l,x}\rp\rp_{j=k+l,k\neq l}\\
    &\quad\quad +\varepsilon^{\T^j}_{\psi^{\mu_j}_{j,x}}\iota^{\T^{l+j}}_{f^{l+j}_{\mu_{l+j},x}}\varepsilon^{\T^l}_{\psi^{\mu_l}_{l,x}}\lp\frac{1}{2}\iota_{f^{l+j}_{\mu_{l+j},x}}\Omega^{\T^{l+j}}_{\T^l,\T^j}\lp f^l_{\mu_l,x},f^j_{\mu_j,x}\rp\rp_{j\neq l}
\end{align*}
All terms above are signed multiple of one another, hence when summing over all $j\neq k\neq l$ one obtains the term
\[
    \varepsilon^{\T^j}_{\psi^{\mu_j}_{j,x}}\varepsilon^{\T^k}_{\psi^{\mu_k}_{k,x}}\iota^{\T^{j+k}}_{f^{j+k}_{\mu_{j+k},x}}\lp\frac{1}{2}\iota_{f^{j+k}_{\mu_{j+k},x}}\Omega^{\T^{j+k}}_{\T^j,\T^k}\lp f^j_{\mu_j,x},f^k_{\mu_k,x}\rp\rp_{j\neq k}
\]
The last term takes the form
\begin{align*}
    &\lp\varepsilon^{\T^j}_{\psi^{\mu_j}_j}-u^{j-1}\iota^{\T^j}_{f^j_{\mu_j}}\rp c\lp\varphi^k\rp\lp\varepsilon^{\T^l}_{\psi^{\mu_l}_l}-u^{l-1}\iota^{\T^l}_{f^l_{\mu_l}}\rp\cdot\\
    &\lp -\frac{1}{2}\iota_{e_i}\Omega^{\F}_{\T^l,\T^j}\lp f^l_{\mu_l},f^j_{\mu_j}\rp+\frac{u^{-j+1}}{2}\iota_{f^j_{\mu_j}}\Omega^{\T^j}_{\F,\T^l}\lp e_k,f^l_{\mu_l}\rp-\frac{u^{-l+1}}{2}\iota_{f^l_{\mu_l}}\Omega^{\T^l}_{\T^j,\F}\lp f^j_{\mu_j},e_k\rp\rp
\end{align*}
All terms in the above expansion have graded Heisenberg order $\leq 1$ and after taking $\widecheck{\sigma}^G_1$ one obtains 
\begin{align*}
    \varepsilon^{\T^j}_{\psi^{\mu_j}_{j,x}}c\lp\varphi^k_x\rp\iota^{\T^{j+1}}_{f^{j+1}_{\mu_{j+1},x}}\lp\frac{1}{2}\iota_{f^{j+1}_{\mu_{j+1}}}\Omega^{\T^{j+1}}_{\T^j,\F}\lp f^j_{\mu_j,x},e_{k,x}\rp\rp -\iota^{\T^j}_{f^j_{\mu_j,x}}c\lp\varphi^k_x\rp\varepsilon^{\T^{j-1}}_{\psi^{\mu_{j-1}}_{j-1,x}}\lp\frac{1}{2}\iota_{f^j_{\mu_j,x}}\Omega^{\T^{j}}_{\F,\T^{j-1}}\lp e_{k,x},f^{j-1}_{\mu_{j-1},x}\rp\rp
\end{align*}
However summing over all $j$ above the terms above cancel. Collecting our calculations, we have found that at any $x\in M$ $\bD^E$ has graded Heisenberg order $1$ and that graded cosymbol is given by
\begin{align*}
    \widecheck{\sigma}^G_1\lp\bD^E\rp_x=c\lp\varphi^i_x\rp[e_i]_x&+\varepsilon^{\T^j}_{\psi^{\mu_j}_{j,x}}[f^j_{\mu_j}]_x+\frac{1}{4}\iota^{\T^2}_{f^2_{\mu_2,x}}c\lp\varphi^i_x\rp c\lp\varphi^k_x\rp \iota_{f^2_{\mu_2,x}}\Omega^{\T^1}_{\F,\F}\lp e_{i,x},e_{k,x}\rp\\
    &\quad+\frac{1}{2}c\lp\varphi^i_x\rp\iota^{\T^{j+1}}_{f^{j+1}_{\mu_{j+1},x}}\varepsilon^{\T^j}_{\psi^{\mu_j}_{j,x}}\lp\iota_{f^{j+1}_{\mu_{j+1},x}}\Omega^{\T^{j+1}}_{\T^j,\F}\lp f^j_{\mu_j,x},e_{i,x}\rp\rp\\
    &\quad\quad+\frac{1}{2}\varepsilon^{\T^j}_{\psi^{\mu_j}_{j,x}}\varepsilon^{\T^k}_{\psi^{\mu_k}_{k,x}}\iota^{\T^{j+k}}_{f^{j+k}_{\mu_{j+k},x}}\lp\frac{1}{2}\iota_{f^{j+k}_{\mu_{j+k},x}}\Omega^{\T^{j+k}}_{\T^j,\T^k}\lp f^j_{\mu_j,x},f^k_{\mu_k,x}\rp\rp.
\end{align*}
In the above calculation, the cosymbol converted all covariant derivatives to ordinary derivatives, got rid of all terms involving $Lg_{\T^j}$ and $\Omega^{\T^l}_{\T^j,\T^k}$ when $l\neq j+k$, and kept all other terms. However, this is precisely the recipe for computing $\bD^{E^L_x}$ relative to a left-invariant framing of $\t_{\F}M^L_x$. Thus we have found that
\[
    \widecheck{\sigma}^G_1\lp\bD^E\rp=\bD^{E^L_x}.
\]
\end{proof}
The hypoellipticity of the above operator is now left to checking injectivity of $\pi^{E_x\otimes\wedge\T^*_x}_x\lp\bD^{E^L_x}\rp$ on all non-trivial irreducible unitary representations of $T_{\F}M_x$ for all $x\in M$. Given the complexity of the local model operators, the use of Kirillovs method to compute the operators $\pi^{E_x\otimes\wedge\T^*_x}_x\lp\bD^{E^L_x}\rp$ and prove their injectivity seems out of reach. We show however that general hypoellipticity can be achieved if one allows for a modification within the definition of the asymptotic Bismut superconnection. 
\begin{theorem}
    Suppose that $\lp M,\F,g_{\F}\rp$ is an $m$-step filtered subRiemannian manifold and $\lp\T^j,g_{\T^j}\rp$ is a sequence of Euclidean distributions which are transverse $\F^j=\F^{j-1}\oplus\T^j$. Suppose that $\F^1$ is bracket generating. Then for a $\bZ_2$-graded covariant Hermitian Clifford module
    \[
        \lp C\ell\lp\F^*,g^*_{\F}\rp,\nabla^{\F,0}\rp\circlearrowright\lp E,h_E,\nabla^E\rp
    \]
    the operator
    \[
        \widetilde{\bD}^E:=D^{\F,E\otimes\wedge\T^*}+d^{\T}_{\nabla^E}+\lim^{F.P.}_{u\rightarrow 0^+}m_u\circ m_u\lp\omega_u\rp
    \]
   has graded Heisenberg order $1$ and satisfies the graded Rockland condition.
\end{theorem}
\begin{proof}
    For $\T^1=\F$ let $\{ e_i\}_{i=1}^{n_1}$ be a local $g_{\F}$-orthonormal framing and let $\{\varphi^i\}_{i=1}^{n_1}$ denote the dual framing and for each $2\leq j\leq m$ we let $\{f^j_{\mu_j}\}_{\mu_j=1}^{n_j}$ be a local $g_{\T_j}$-orthonormal framing and $\{\psi^{\mu_j}_j\}_{\mu_j=1}^{n_j}$ denote the corresponding dual framing. A repeat of the calculations presented in the proof of Theorem \ref{FABS} show that the operator $\widetilde{\bD}^E$ has graded Heisenberg order $1$ and at $x\in M$ the cosymbol is
    \begin{align*}
    \widecheck{\sigma}^G_1\lp\widetilde{\bD}^E\rp_x=c\lp\varphi^i_x\rp[e_i]_x&+\varepsilon^{\T^j}_{\psi^{\mu_j}_{j,x}}[f^j_{\mu_j}]_x+\iota^{\T^2}_{f^2_{\mu_2,x}}c\lp\varphi^i_x\rp c\lp\varphi^k_x\rp \lp\frac{1}{2}\iota_{f^2_{\mu_2,x}}\Omega^{\T^1}_{\F,\F}\lp e_{i,x},e_{k,x}\rp\rp\\
    &\quad+c\lp\varphi^i_x\rp\iota^{\T^{j+1}}_{f^{j+1}_{\mu_{j+1},x}}\varepsilon^{\T^j}_{\psi^{\mu_j}_{j,x}}\lp\iota_{f^{j+1}_{\mu_{j+1},x}}\Omega^{\T^{j+1}}_{\T^j,\F}\lp f^j_{\mu_j,x},e_{i,x}\rp\rp\\
    &\quad\quad+\varepsilon^{\T^j}_{\psi^{\mu_j}_{j,x}}\varepsilon^{\T^k}_{\psi^{\mu_k}_{k,x}}\iota^{\T^{j+k}}_{f^{j+k}_{\mu_{j+k},x}}\lp\frac{1}{2}\iota_{f^{j+k}_{\mu_{j+k},x}}\Omega^{\T^{j+k}}_{\T^j,\T^k}\lp f^j_{\mu_j,x},f^k_{\mu_k,x}\rp\rp.
\end{align*}
If we square the above cosymol we obtain
\begin{align*}
    \lp\widecheck{\sigma}^G_1\lp\widetilde{\bD}^E\rp_x\rp^2=-\delta^{ij}[e_i]_x[e_j]_x+[c\lp\varphi^i_x\rp[e_i]_x,A]+A^2
\end{align*}
where
\begin{align*}
    A&=\iota^{\T^2}_{f^2_{\mu_2,x}}c\lp\varphi^i_x\rp c\lp\varphi^k_x\rp \lp\frac{1}{2}\iota_{f^2_{\mu_2,x}}\Omega^{\T^1}_{\F,\F}\lp e_{i,x},e_{k,x}\rp\rp+c\lp\varphi^i_x\rp\iota^{\T^{j+1}}_{f^{j+1}_{\mu_{j+1},x}}\varepsilon^{\T^j}_{\psi^{\mu_j}_{j,x}}\lp\iota_{f^{j+1}_{\mu_{j+1},x}}\Omega^{\T^{j+1}}_{\T^j,\F}\lp f^j_{\mu_j,x},e_{i,x}\rp\rp\\
    &\quad\quad+\varepsilon^{\T^j}_{\psi^{\mu_j}_{j,x}}\varepsilon^{\T^k}_{\psi^{\mu_k}_{k,x}}\iota^{\T^{j+k}}_{f^{j+k}_{\mu_{j+k},x}}\lp\frac{1}{2}\iota_{f^{j+k}_{\mu_{j+k},x}}\Omega^{\T^{j+k}}_{\T^j,\T^k}\lp f^j_{\mu_j,x},f^k_{\mu_k,x}\rp\rp.
\end{align*}
All terms in $[c\lp\varphi^i_x\rp[e_i]_x,A]+A^2$ will negatively shift the degree by $-1$ or $-2$. Moreover, within any irreducible unitary representation $\pi_x$ of $T_{\F}M_x$ the operator $\pi^{E_x\otimes\wedge\T^*_x}_x\lp -\delta^{ij}[e_i]_x[e_j]_x\rp$ is non-negative and injective, due to the bracket generating assumption on $\F^1$. Thus, if one were to construct a basis of $E_x\otimes\wedge\T^*_x$ which was compatible with the grading, then $\pi_x^{E_x\otimes\wedge\T^*_x}\lp\widecheck{\sigma}^G_1\lp\widetilde{\bD}^E\rp_x^2\rp$ would be triangular with $\pi^{E_x\otimes\wedge\T^*_x}_x\lp-\delta^{ij}[e_i]_x[e_j]_x\rp$ on the diagonal. The injectivity of the diagonal proves the injectivity of $\pi_x^{E_x\otimes\wedge\T^*_x}\lp\widecheck{\sigma}^G_1\lp\widetilde{\bD}^E\rp_x^2\rp$.
\end{proof}
\subsection{An Example: Principal $\bS^1$-bundles With Non-Vanishing Curvature}
Consider a closed manifold $\Sigma$ along with a covariant principal $\bS^1$-bundle $\lp P,\theta\rp\circlearrowleft\bS^1\xrightarrow{\pi} \Sigma$ for which the curvature $d\theta=\pi^*\bOmega$ is non-vanishing along $\Sigma$. This implies that the distribution $\F:=\Ker\theta$ gives rise to a two-step filtered manifold $\lp P,\F\rp$. We define the Reeb field $T$ as the generator of the $\bS^1$ action
\[
    T_p=\frac{d}{dt}\Big|_{t=0}p\cdot e^{it}.
\]
and take $\T$ to be the smooth subbundle of $TP$ spanned by $T$ along with transverse metric $g_{\T}\lp T,T\rp\equiv 1$. If $g_{\Sigma}$ is an arbitrary Riemannian metric on $\Sigma$, then using the isomorphism $\F\simeq\pi^*T\Sigma$ we may define $g_{\F}:=\pi^*g_{\Sigma}$ and compute that 
\[
    \nabla^{\F,0}=\pi^*\nabla^{g_{\Sigma}}.
\]
Note that one has $L_Tg_{\F}=0$.
We will consider a generic $\bZ_2$-graded covariant Hermitian Clifford module
\[
    \lp C\ell\lp T^*\Sigma,g^*_{\Sigma}\rp,\nabla^{g_{\Sigma}}\rp\circlearrowright_c\lp E,h_E,\nabla^E \rp
\]
and lift along $\pi:P\rightarrow \Sigma$ to form the $u$-dependent family of $\bZ_2$-graded covariant Hermitian Clifford modules
\[
    \lp C\ell\lp \F^*,g^*_{\F}\rp,\nabla^{\F,u}\rp\circlearrowright_c\lp\pi^*E,\pi^*h_E,\nabla^{E,u}\rp
\]
where
\[
    \nabla^{E,u}:=\pi^*\nabla^E+\frac{u^{-1}}{4}\theta\otimes c\lp d\theta\rp
\]
From here one obtains the asymptotic Bismut superconnection
\[
    \bD^{\pi^*E}=D^{\F,E\otimes\wedge[\theta]}+\pi^*\nabla^{E\otimes\wedge[\theta]}_T\varepsilon^{\T}_{\theta_{\T}}+\frac{1}{4}c\lp d\theta\rp\iota^{\T}_T.
\]
Observe that under the isomorphism $\pi^*E\otimes\wedge\T^*\simeq \pi^*E\otimes\wedge[\theta]$ we have the global Fourier expansion
\[
    \Gamma\lp P,\pi^* E\otimes\wedge[\theta]\rp\simeq\bigoplus_{m\in\bZ}\Gamma\lp\Sigma, E\otimes L_m\otimes\wedge[\theta]\rp
\]
where $\lp L_m,h_m,\nabla^{L_m}\rp:=\lp P,\theta\rp\times_{\rho_m}\lp\bC,h_{\bC}\rp\rightarrow \Sigma$ is the covariant Hermitian line bundle associated to the unitary representation $\rho_m\lp z\rp=z^m$ of $\bS^1$. Given $s\in\Gamma\lp\pi^* E\otimes\wedge[\theta]\rp$ the Fourier expansion takes the form
\[
    s=\sum_{m\in\bZ}s_m,\quad s_m:=\frac{1}{2\pi}\int_0^{2\pi}s\lp p\cdot e^{it}\rp e^{-imt}dt\in \Gamma\lp P,\pi^* E\otimes\wedge[\theta]\rp^{\rho_m}\simeq \Gamma\lp\Sigma, E\otimes L_m\otimes\wedge[\theta]\rp
\]
The operator $\bD^{\pi^*E}$ leaves each Fourier component invariant and on a given component takes the form
\[
    \bD^{\pi^*E}\Big|_{\Gamma\lp\Sigma, E\otimes L_m\otimes \wedge[\theta]\rp}\simeq D^{E\otimes L_m\otimes\wedge[\theta] }+im\varepsilon^{\T}_{\theta_{\T}}+\frac{1}{4}c\lp\bOmega\rp\iota^{\T}_T
\]
where $D^{E\otimes L_m\otimes\wedge[\theta]}$ is the Dirac operator corresponding to the twisted covariant Hermitian Clifford module
\[
    \lp C\ell\lp T^*\Sigma,g^*_{\Sigma}\rp,\nabla^{g_{\Sigma}}\rp\circlearrowright_c\lp E,h_E,\nabla^E\rp\otimes\lp L_m,h_m,\nabla^{L_m}\rp\otimes\lp \wedge[\theta],g_{\wedge[\theta]},d^{\T}\rp
\]
Given our general facts concerning the hypoellipticity of $\bD$ we have that $\Ker\bD$ is a finite dimensional $\bZ_2$-graded subspace of $\Gamma\lp P,\pi^* E\otimes\wedge\T^*\rp$. Using the Fourier theoretic description given above we can compute the kernel precisely. Thus suppose that $\eta+\tau\theta\in\Gamma^0\lp P,\pi^*E\otimes\wedge[\theta]\rp$ satisfies $\bD^{\pi^*E}\lp\eta+\tau\theta\rp=0$. Then if $\eta_m$ and $\tau_m$ denote the Fourier components of sections then we have that 
\[
  D^{E\otimes L_m}\eta_m-\frac{1}{4}c\lp\bOmega\rp\tau_m=0,\quad D^{E\otimes L_m}\tau_m+im\eta_m=0,\quad \forall m\in\bZ  
\]
When $m\neq0$ we have 
\[
    \eta_m=\frac{i}{m}D^{E\otimes L_m}\tau_m\quad\Rightarrow\quad \frac{1}{m}\lp D^{E\otimes L_m}\rp^2\tau_m+\frac{1}{4}ic\lp\bOmega\rp\tau_m=0.
\]
For $m=0$ we obtain the system
\[
    D^E\tau_0=0,\quad D^E\eta_0=\frac{1}{4}c\lp\bOmega\rp\tau_0
\]
The latter equation is satisfies if and only if 
\[
    \la c\lp\bOmega\rp\tau_0,\xi\ra_{L^2\lp E\rp}=0,\quad \text{for all $\xi\in\Ker^1 D^{E}$}
\]
Considering now the case of $\sigma+\xi\theta\in\Gamma^1\lp P,\pi^*E\otimes\wedge[\theta]\rp$ satisfying $\bD^{\pi^*E}\lp\sigma+\xi\theta\rp=0$ we have that the Fourier components satisfy
\[
    D^{E\otimes L_m}\sigma_m+\frac{1}{4}c\lp\bOmega\rp\xi_m=0,\quad -im\sigma_m+D^{E\otimes L_m}\xi_m=0.
\]
When $m\neq0$ have
\[
    \sigma_m=\frac{1}{im}D^{E\otimes L_m}\xi_m\quad\Rightarrow\quad \frac{1}{m}\lp D^{E\otimes L_m}\rp^2\xi_m+\frac{1}{4}ic\lp\bOmega\rp\xi_m=0
\]
For $m=0$ we obtain the system
\[
    D^E\xi_0=0,\quad D^E\sigma_0=-\frac{1}{4}c\lp\bOmega\rp\xi_0.
\]
Collecting our calculations, we obtain the following proposition.
\begin{proposition}
\label{S1C}
    Suppose that $\Sigma$ is a closed manifold, $\lp P,\theta\rp\circlearrowleft\bS^1\xrightarrow{\pi}\Sigma$ is a covariant principal $\bS^1$-bundle with non-vanishing curvature $d\theta=\pi^*\bOmega$, and $g_{\Sigma}$ is a Riemannian metric of $\Sigma$. Then if we have a $\bZ_2$-graded covariant Hermitian Clifford module
    \[
        \lp C\ell\lp T^*\Sigma,g^*_{\Sigma}\rp,\nabla^{g_{\Sigma}}\rp\circlearrowright_c\lp E,h_E,\nabla^E \rp
    \]
    the kernel of the asymptotic Bismut superconnection $\bD^{\pi^*E}$ for the metric $g_{\F}:=\pi^*g_{\Sigma}$ defined above is given by 
    \begin{itemize}
        \item The even component of $\Ker\bD^{\pi^*E}$ is isomorphic to the direct summand
        \[
            \hspace{-.8cm}\Ker^0D^E\oplus\left\{ \tau\in\Ker^1D^E\Big|\text{ $c\lp\bOmega\rp\tau\perp \Ker^1D^E$}\right\}\bigoplus_{m\neq0}\left\{\tau_m\in\Gamma^1\lp E\otimes L_m\rp\Big|\text{ $c\lp \bOmega\rp\tau_m=\frac{4i}{m}\lp D^{E\otimes L_m}\rp^2\tau_m$}\right\}
        \]
        \item The odd component of $\Ker\bD^{\pi^*E}$ is isomorphic to the direct summand
        \[
            \hspace{-.8cm}\Ker^1D^E\oplus\left\{\xi\in\Ker^0D^E\Big|\text{ $c\lp\bOmega\rp\xi\perp\Ker^0D^E$}\right\}\bigoplus_{m\neq0}\left\{ \xi_m\in\Gamma^0\lp E\otimes L_m\rp\Big|\text{ $c\lp\bOmega\rp\xi_m=\frac{4i}{m}\lp D^{E\otimes L_m}\rp^2\xi_m$}\right\}
        \]
    \end{itemize}
\end{proposition}
An interesting feature of the above result is that the hypoellipticity of $\bD^{\pi^*E}$ forces the above summands, although indexed by the integers, to be \emph{finite dimensional}. An alternative proof that the above summands are always finite dimensional that does not use non-commutative micro-local calculus would be a great sanity check. In the case of a closed, connected surface $\Sigma$, one can obtain a more precise description when $\bOmega$ is \emph{conformally compatible} with $g_{\Sigma}$, i.e. when there exists $\alpha\in\bR$ such that $\lp\alpha\cdot g_{\Sigma},\bOmega\rp$ gives rise to a Kähler structure. Note that this directly implies that $\bOmega=\alpha dA$, where $dA$ is the oriented area element of $g_{\Sigma}$. In this case the above theorem instructs us to count the dimension of solutions to the equation
\[
    \alpha c\lp dA\rp s_m=\frac{4i}{m}\lp D^{E\otimes L_m}\rp^2s_m,\quad s_m\in\Gamma\lp E\otimes L_m\rp.
\]
which one can check is effectively asking whether the operator $\lp D^{E\otimes L_m}\rp^2$ has $\pm m\alpha/4$ as an eigenvalue. We now turn to an even simpler example where the full kernel can be computed, the case of $\Sigma\simeq \bS^2$. We will use the Peter-Weyl theorem in this case to diagonalize the operators and obtain a precise description of the kernel. 
\begin{theorem}
\label{S1S2}
    Suppose we equip $\bS^2$ with the round metric $g$ of radius $1/2$ and a conformally compatible symplectic form $\bOmega$, and consider a covariant principal $\bS^1$-bundle $\lp P,\theta\rp\circlearrowleft\bS^1\xrightarrow{\pi} \bS^2$ with $d\theta=\pi^*\bOmega$. Then for the $\bZ_2$-graded covariant Clifford module
    \[
        \lp C\ell\lp T^*\bS^2,g^*\rp,\nabla^g\rp\circlearrowright\lp\wedge T^*\bS^2,\wedge g^*,\nabla^g\rp
    \]
    one has
    \[
        \Ker\bD^{\pi^*\wedge T^*\bS^2}\simeq H^*_{dR}\lp\bS^2\rp
    \]
    where the isomorphism is at the level of $\bZ_2$-graded vector spaces.
\end{theorem}
\begin{proof}
    We first consider the case of the Hopf fibration $\bS^3\circlearrowleft\bS^1\rightarrow \bS^2$ and proceed to compute the even part of the kernel of $\bD^{\pi^*\wedge T^*\bS^2}$. Suppose we have an even section $\eta+\tau\theta$ of $\pi^*\wedge T^*\bS^2\otimes\wedge[\theta]$ satisfying
    \[
        \bD^{\pi^*\wedge T^*\bS^2}\lp\eta+\tau\theta\rp=0.
    \]
    Then we obtain the system
    \[
        D^{\F,\pi^*\wedge T^*\bS^2}\eta-\frac{1}{4}c\lp d\theta\rp \tau=0,\quad D^{\pi^*\wedge T^*\bS^2}\tau+\pi^*\nabla^{g_{\Sigma}}_T\eta=0.
    \]
    Because $\bOmega_{\bS^2}=-2dV_g$ we see that $D^{\pi^*\wedge T^*\bS^2}$ \emph{anti-commutes} with $c\lp d\theta\rp$. Applying $D^{\pi^*\wedge T^*\bS^2}$ to the first equation and denoting $\Delta^{\pi^*\wedge T^*\bS^2}=\lp D^{\pi^*\wedge T^*\bS^2}\rp^2$ gives
    \[
        \Delta^{\pi^*\wedge T^*\bS^2}\eta+\frac{1}{4}c\lp d\theta\rp D^{\pi^*\wedge T^*\bS^2}\tau=0
    \]
and using the second equation we obtain
\[
    \Delta^{\pi^*\wedge T^*\bS^2}\eta-\frac{1}{4}c\lp d\theta\rp\pi^*\nabla^{g_{\Sigma}}_T\eta=0.
\]
Observe that solutions to the above equation determine solutions to the original system since we can solve for $\tau$ in the first equation by using $c\lp d\theta\rp^2=-4$ giving
\[
    \tau=-c\lp d\theta\rp D^{\pi^*\wedge T^*\bS^2}\eta.
\]
The choice of metric of radius $1/2$ on $\bS^2$ guarantees that the induced metric
\[
    g_{\bS^3}=g_{\F}\oplus g_{\T}
\]
is the round metric of radius $1$. Furthermore recall that for the Hopf fibration $\bS^3\circlearrowleft\bS^1\rightarrow\bS^2$ there exists a global orthonormal framing $X,Y\in\Gamma\lp \F\rp$ such that 
\[
    [X,Y]=2T,\quad [Y,T]=2X,\quad [T,X]=2Y.
\]
The above bracket relations imply, via the Koszul identity, that we also have that
\[
    \nabla^{\F,0}_XY=\nabla^{\F,0}_YX=\nabla^{\F,0}_XX=\nabla^{\F,0}_YY=0.
\]
Lastly, we have $\pi^*\wedge^2 T^*\bS^2\simeq \underline{\bC}_{\bS^3}$ where the isomorphism is given by the global non-vanishing section $\pi^*dV_g$ of $\pi^*\wedge^2T^*\bS^2$. Thus under the isomorphism $\wedge^{ev}\F^*\simeq\underline{\bC^2}_{\bS^3}$ the connection $\nabla^{\F,0}$ becomes the canonical flat connection on $\underline{\bC^2}_{\bS^3}$.
Thus we have the following global formula 
\[
    \Delta^{\pi^*\wedge T^*\bS^2}=-X^2-Y^2-c\lp d\theta\rp T,
\]
which when combined with $-\frac{1}{4}c_{\F}\lp\Omega_{\F}\rp T$ gives
\[
    \Delta^{\pi^*\wedge T^*\bS^2}-\frac{1}{4}c\lp d\theta\rp T=-X^2-Y^2-\frac{5}{4}c\lp d\theta\rp T.
\]  
If we use $\bOmega=-2dV_g$ and a hidden zero we obtain
\[
    \Delta^{\pi^*\wedge T^*\bS^2}+\frac{1}{4}c\lp d\theta\rp T= \Delta_{\bS^3}+T^2+\frac{5}{2}c\lp \pi^*dV_g\rp T.
\]
Here we $\Delta_{\bS^3}$ denotes the Laplacian on $\bS^3$ with the round metric of radius $1$. With the respect to the global trivialization $\pi^*\wedge^{ev}T^*\bS^2\simeq\underline{\bC^2}_{\bS^3}$ given above, the operator section $c\lp \pi^*dV_g\rp$ is the constant matrix
\[
    c\lp\pi^*dV_g\rp=\begin{bmatrix}
        0 & -1 \\
        1 & 0
    \end{bmatrix}
\]
which squares to $-1$ and the eigenspaces for this matrix will be denoted $\E_{\pm i}$. If we use the homogeneous polynomials as a basis for $L^2\lp\bS^3,\E_{\pm i}\rp$ and if $V^{\E_{\pm i}}_m$ denotes the space of homogeneous polynomials of degree $m\geq 0$ with coefficients in $\E_{\pm i}$, the operator $\Delta^{\pi^*\wedge T^*\bS^2}-\frac{1}{4}c\lp d\theta\rp T$ leaves each of these subspaces invariant and is diagonalizable on each such subspace. The spectrum is given as
\[
    \sigma\lp\Delta^{\pi^*\wedge T^*\bS^2}-\frac{1}{4}c\lp d\theta\rp T\big|_{V^{\E_{\pm i}}_m}\rp=\{ m(m+2)-k^2\pm \frac{5}{2}k\}_{k=-m}^m.
\]
Thus the only possible subspaces where $\Delta^{\pi^*\wedge T^*\bS^2}+\frac{1}{4}c\lp d\theta\rp T$ has non-trivial kernel are those subspaces for which there is an $m\geq0$ and $0\leq|k|\leq m$ such that
\[
    m(m+2)-k^2\pm \frac{5}{2}k=0.
\]
Since $k$ is already assumed to lie within $-m\leq k\leq m$ we can absorb the $\pm$ into $k$ and simply solve
\[
    m(m+2)-k^2+\frac{5}{2}k=0.
\]
The integer solutions to the above quadratic equation are
\[
    \lp m,k\rp=(-2,0),(-1,2),(0,0).
\]
Because $m\geq0$, we see that the kernel of $\Delta^{\pi^*\wedge T^*\bS^2}-\frac{1}{4}c\lp d\theta\rp T$ is the constant functions of $\bS^3$ with values in $\bC^2$. But these are precisely the functions on $\bS^3$ which are pullbacks of constant functions on $\bS^2$, hence
\[
    \Ker^+\bD^{\pi^*\wedge T^*\bS^2}\simeq H^{ev}_{dR}\lp\bS^2\rp.
\]
Suppose now that we have an odd element $\chi+\xi\theta$ satisfying
\[
    \bD^{\pi^*\wedge T^*\bS^2}\lp \chi+\xi\theta\rp=0.
\]
Once again we obtain a system
\[
    D^{\pi^*\wedge T^*\bS^2}\chi+\frac{1}{4}c\lp d\theta\rp\xi=0,\quad D^{\pi^*\wedge T^*\bS^2}\xi-\pi^*\nabla^{g_{\Sigma}}_T\chi=0,
\]
for which the same trick gives
\[
    \Delta^{\pi\wedge T^*\bS^2}\chi-\frac{1}{4}c\lp d\theta\rp\pi^*\nabla^{g_{\Sigma}}_T\chi=0,\quad \xi=c\lp d\theta\rp D^{\pi^*\wedge T^*\bS^2}\chi.
\]
Now observe that if we apply $D^{\pi^*\wedge T^*\bS^2}$ to the first equation and let $\alpha=D^{\pi^*\wedge T^*\bS^2}\chi\in\Gamma\lp\pi^*\wedge^{ev} T^*\bS^2\rp$, then
\[
    \Delta^{\pi^*\wedge T^*\bS^2}\alpha+\frac{1}{4}c\lp d\theta\rp \pi^*\nabla^{g_{\Sigma}}_T\alpha=0.
\]
Thus we must now compute the kernel of the operator
\[  
    \Delta^{\pi^*\wedge T^*\bS^2}+\frac{1}{4}c\lp d\theta\rp \pi^*\nabla^{g_{\Sigma}}_T=\Delta_{\bS^3}+T^2-\frac{3}{4}c\lp \pi^*dV_g\rp T.
\]
When we restrict to $V^{\E_{\pm i}}_m$ and compute the spectrum we find
\[
    \sigma\lp\Delta^{\pi^*\wedge T^*\bS^2}+\frac{1}{4}c\lp d\theta\rp T\big|_{V^{\E_{\pm i}}_m}\rp=\{ m(m+2)-k^2\pm \frac{3}{2}k\}_{k=-m}^m.
\]
In this case one obtains the lower bound
\[
    \lp\Delta^{\pi^*\wedge T^*\bS^2}+\frac{1}{4}c\lp d\theta\rp\rp\Big|_{V^{\E_{\pm i}}_m}\geq \frac{m}{2}.
\]
Thus we have that $\alpha$ must be the pullback of a constant function on $\bS^2$. Hence 
\[
    0=D^{\pi^*\wedge T^*\bS^2}\alpha=\Delta^{\pi^*\wedge T^*\bS^2}\chi.
\]
which implies $\pi^*\nabla^{g_{\Sigma}}_T\chi=0$, i.e. that $\chi$ is the pullback of a $1$-form on $\bS^2$ and that this $1$-form is harmonic on $\bS^2$ and thus is zero by the Hodge Theorem. All together, what we have found is that $\chi_{\F}=0$ and hence so is $\xi_{\F}$. Thus
\[
    \Ker^-\bD^{\wedge\F^*}=0.
\]
Now recall that if $\pi_P:P\circlearrowleft\bS^1\rightarrow \bS^1$ is a principal $\bS^1$-bundle then there is a covering map $w_m:\bS^1\rightarrow \bS^1$, given my $w_m\lp z\rp=z^m$ for $m\in\bZ$, and a covering map $\rho:\bS^3\rightarrow P$ which fit into a commutative diagram
\begin{displaymath} 
\begin{tikzcd}[column sep=8em]
\bS^3\times \bS^1 \arrow[r, "R"] \arrow[d, "\rho\times w_m"]
\& \bS^3 \arrow[r, "\pi_{\bS^3}"] \arrow[d, "\rho" ] \& \bS^2 \arrow[d, "Id"] \\
P\times \bS^1 \arrow[r, "R"]
\& P \arrow[r, "\pi_P"] \& \bS^2
\end{tikzcd}
\end{displaymath}
where $R$ denotes the right multiplication map. The diagram above implies that if $T_P$ is the generator of the $\bS^1$-action on $P$, then $\rho_*T=mT_P$. Thus if choose the contact form $\theta_P$ so that $p^*\theta_P=m\theta_{\bS^3}$, the extra factor of $m$ ensure that $\theta_P$ will serve as a connection one-form for $P$. This also directly implies that if $\pi^*_P\bOmega_P=d\theta_P$, then $\bOmega_P=m\bOmega$, hence $\bOmega_P$ is conformally compatible with $g$. The sub-Riemannian metric $g_{\F_P}$ is chosen so that $\rho^*g_{\F_P}=g_{\F_{\bS^3}}$. In this case, we obtain an injective linear map
\[
    \rho^*:\Gamma\lp\pi^*_P\wedge T^*\bS^2\otimes\wedge[\theta_P]\rp\rightarrow \Gamma\lp\pi^*\wedge T^*\bS^2\otimes\wedge[\theta]\rp
\]
which commutes with action of the asymptotic Bismut superconnection
\[
    \rho^*\bD^{\pi_P^*\wedge T^*\bS^2}=\bD^{\pi^*\wedge T^*\bS^2}\rho^*.
\]
By Proposition \ref{S1C} we have that $\pi_P^*$ injectively places $H^*_{dR}\lp\bS^2\rp$ within $\Ker\bD^{\pi^*_P\wedge T^*\bS^2}$, and $\rho^*$ injectively places $\Ker\bD^{\pi^*_P\wedge T^*\bS^2}$ within $\Ker\bD^{\pi^*\wedge T^*\bS^2}$. Because we have already checked that $\Ker\bD^{\pi^*\wedge T^*\bS^2}\simeq H^*_{dR}\lp\bS^2\rp$, we have that 
\[
    \Ker\bD^{\pi^*_P\wedge T^*\bS^2}\simeq H^*_{dR}\lp\bS^2\rp.
\]
Now suppose that $\theta'_P$ is an alternative connection one-form for $P$ whose curvature is conformally compatible with $g$. This immediately implies that $\theta_P$ and $\theta_P'$ have the same curvature and because $\bS^2$ is simply connected, that $\theta_P$ and $\theta_P'$ must be related by a gauge transformation 
\[
    u:P\circlearrowleft\bS^1\xrightarrow{\sim}P\circlearrowleft\bS^1
\]
satisfying $u^*\theta_P=\theta'_P$ and $\pi_Pu=u\pi_P$. The gauge transformation gives rise to a linear isomorphism
\[
    u^*:\Gamma\lp \pi^*_P\wedge T^*\bS^2\otimes\wedge[\theta_P]\rp\xrightarrow{\sim}\Gamma\lp\pi^*_P\wedge T^*\bS^2\otimes\wedge[\theta_P']\rp
\]
which intertwines the asymptotic Bismut superconnections. The theorem follows.
\end{proof}
\section{Directions for Further Study}
There are a few questions that this constructions demands be answered. The first has already been hinted at above which is the general hypoellipticity on arbitrary filtered manifolds. The Kirillov orbit method for computing representations quickly becomes difficult to practically implement once one leaves the realm of two-step filtrations thus leaving the general hypoellipticity an open problem. The second is the question of whether $\bD^E$ has good index theoretic properties. Examples coming from principal $\bS^1$-bundles over $\bS^2$ demonstrate that it has rather odd behavior but also suggests a redefinition of index. If one allows Clifford relations with matrix coefficients, we provide examples for which the asymptotic Bismut superconnection will have non-vanishing index. In our last subsection, we provide a direct link between the asymptotic Bismut superconnection and 1-dimensional non-linear sigma models.
\subsection{Relations to Index Theory}
Given that the asymptotic Bismut superconnection is hypoelliptic on contact manifolds, and in fact more generally on arbitrary two-step filtered manifolds, it is natural to inquire on the index theoretic properties of these operators. We confine ourselves to the setting of contact subRiemannian manifolds in this section to simplify calculations and demonstrate that simple modifications of our construction of $\bD^E$ can give first order hypoelliptic differential operators with non-vanishing Fredholm index on odd-dimensional manifolds, which to our knowledge is the first example of such an operator.
\begin{theorem}
\label{NI}
    Let $\lp M,\theta,g_{\F}\rp$ be a contact subRiemannian manifold which is closed. Then for any $\bZ_2$-graded covariant Hermitian Clifford module
    \[  
    \lp C\ell\lp\F^*,g^*_{\F}\rp,\nabla^{\F,0}\rp\circlearrowright\lp E,h_E,\nabla^E\rp.
    \]
    the restriction of the asymptotic Bismut superconnection to the even sections 
    \[
        \bD^{E,+}:\Gamma^0\lp E\otimes\wedge[\theta]\rp\rightarrow \Gamma^1\lp E\otimes\wedge[\theta]\rp
    \]
    is Fredholm and its Fredholm index is zero.
\end{theorem}
\begin{proof}
    We equip $M$ with the Riemannian metric $g_M=g_{\F}\oplus g_{\T}$ and $E\otimes\wedge[\theta]$ with the Hermitian fiber metric $h_E\otimes g_{\wedge[\theta]}$. If we use the global formula for the asymptotic Bismut superconnection on contact manifolds, we see that the $L^2$-formal adjoint of $\bD^E$ with respect to the metric structures given is 
    \[
        \lp\bD^E\rp^*=D^{\F,E\otimes\wedge[\theta]}-\iota^{\T}_T\lp\nabla^{E\otimes\wedge[\theta]}_T+\text{div}T\rp-\frac{1}{4}c\lp d\theta\rp\varepsilon^{\T}_{\theta_{\T}}+\frac{1}{4}\iota^{\T}\circ\text{tr}Lg_{\F}.
    \]
    If we reverse the grading on $E\otimes\wedge[\theta]$ by defining $E\otimes\wedge^0[\theta]$ to have degree $1$ and $E\otimes\wedge^1[\theta]$ to have degree $0$, then $\lp\bD^E\rp^*$ is a differential operator of graded Heisenberg order $1$. At $x\in M$ it's graded cosymbol is given by
    \[
        \widecheck{\sigma}^G_1\lp\lp\bD^E\rp^*\rp=\lp\bD^{E^L_x}\rp^*.
    \]
    At a given $x\in M$, we can use the splittings
    \[
        \Gamma^0\lp E^L_x\otimes\wedge[\theta]\rp=\Gamma\lp \lp E^0_x\rp^L\rp\oplus\Gamma\lp \lp E^1_x\rp^L\rp\cdot\theta_x,\quad\Gamma^1\lp E^L_x\otimes\wedge[\theta]\rp=\Gamma\lp\lp E^1_x\rp^L\rp\oplus\Gamma\lp \lp E^0_x\rp^L\rp\cdot\theta_x
    \]
    to express the cosymbol of the asymptotic Bismut superconnection in matrix form
    \[
        \bD^{E^L_x,+}=\begin{bmatrix}
            D^{\t_{\F}M^L_x,E_x^+} & \frac{1}{4}c\lp d\theta_x\rp\epsilon_x \\
            T_x\epsilon_x & D^{\t_{\F}M^L_x,E_x^-}
        \end{bmatrix}
    \]
    where $\epsilon_x$ is the grading operator of $E_x$ and the cosymbol of the adjoint
    \[
        \lp \bD^{E^L_x}\rp^*=\begin{bmatrix}
            D^{\t_{\F}M^L_x,E_x^-} & -T_x\epsilon_x \\
            -\frac{1}{4}c\lp d\theta_x\rp\epsilon_x & D^{\t_{\F}M^L_x,E^+_x}
        \end{bmatrix}
    \]
    We see then that $\lp\bD^E\rp^*$ will satisfy the graded Rockland condition and thus is also hypoelliptic. By Corollary 5.5 in \cite{DH}, the operator $\bD^{E,+}$ is \emph{Fredholm}. Note that the one-parameter family of cosymbols
    \[
        \bD_{x,t}:=\begin{bmatrix}
            D^{\t_{\F}M^L_x,E^-_x} & -e^{-it}T_x\epsilon_x \\
            -e^{it}\frac{1}{4}c\lp d\theta_x\rp\epsilon_x & D^{\t_{\F}M_x^L,E^+_x}
        \end{bmatrix}
    \]
    gives a one-parameter family of cosymbols which are Rockland for all $t\in\bR$ and whose value at $t=0$ is the cosymbol $\lp\bD^{E^L_x}\rp^*$ and whose value at $t=\pi$ is
    \[
        \bD_{x,t=\pi}=\begin{bmatrix}
            D^{\t_{\F}M^L_x,E^-_x} & \epsilon_xT_x \\
            \frac{1}{4}c\lp d\theta_x\rp\epsilon_x & D^{\t_{\F}M_x^L,E^+_x}
        \end{bmatrix}
    \]
    for which a simple swap of $E^0$ and $E^1$ recovers the symbol of $\bD^{E,+}$. Thus we have that the cosymbol of $\lp\bD^{E,+}\rp^*$, up to a swap of $E^0$ and $E^1$, is homotopic to the symbol of $\bD^{E,+}$. By Corollary 5.5 in \cite{DH} we have that 
    \[
        \text{Ind}\lp\bD^{E,+}\rp^*=\text{Ind} \bD^{E,+}
    \]
    but this directly implies that $\text{Ind} \bD^{E,+}=0$.
\end{proof}
The above theorem should be contrasted with the fact that the $\bZ_2$-invariant
\begin{align*}
    \chi\lp\bD^E\rp&:=\text{dim}\Ker^+\bD^E-\text{dim}\Ker^-\bD^E
\end{align*}
is \emph{non-zero} in general, as Theorem \ref{S1S2} demonstrates in the case of principal $\bS^1$-bundles over $\bS^2$. This deviation between $\chi\lp\bD^E\rp$ and the Fredholm index can be attributed to the non-self-adjointness of $\bD^E$, but also suggests that $\chi\lp\bD^E\rp$, along with it's dependence on subRiemannian structure and filtration $\lp\F,g_{\F}\rp$, is the relevant index one should study.\par
If one is interested in cooking up examples of first order differential operators with non-zero index, the van-Erp trick \cite{Van3} can be employed within this context to produce such examples. Suppose that we have an even section $\gamma\in\Gamma\lp E\rp$ and consider the operator
\[  
    \bD^E_{\gamma}:=D^{\F,E\otimes\wedge[\theta]}+\varepsilon^{\T}_{\theta_{\T}}\nabla^{E\otimes\wedge[\theta]}_T+\gamma c\lp d\theta\rp\iota^{\T}_T+\frac{1}{4}\varepsilon^{\T}\circ \tr Lg_{\F}.
\]
Note that the operator above can be built using the asymptotic Bismut superconnection construction and allowing Clifford relations with matrix coefficients, i.e. for $u>0$ one uses the $\gamma$-twisted Clifford relation
\[
    c^{\T}_{\gamma,u}\lp\alpha\theta\rp^2=-4u\alpha^2\gamma,\quad \alpha\in\bR.
\]
The above relation can be realized by assigning
\[
    c^{\T}_{\gamma,u}\lp\theta\rp_{\gamma}:=\varepsilon^{\T}_{\theta^{\T}}-4u\gamma\iota^{\T}_T\circlearrowright E\otimes\wedge\T^*
\]
and the same computation done in Proposition \ref{GFD} implies that if $D_{\gamma,u}$ denotes the corresponding Dirac operator using the $\gamma$-twisted Clifford relation, then
\[
    \bD^E_{\gamma}=\lim^{F.P.}_{u\rightarrow 0^+}D_{\gamma,u}.
\]
\begin{theorem}
\label{NIT}
    Let $\lp M,\theta,g_{\F}\rp$ be a contact subRiemannian manifold and consider a $\bZ_2$-graded covariant Hermitian Clifford module
    \[  
    \lp C\ell\lp\F^*,g^*_{\F}\rp,\nabla^{\F,0}\rp\circlearrowright_c\lp E,h_E,\nabla^E\rp.
    \]
    Let $\gamma\in\Gamma\lp\End E\rp$ be an even element such that $[\gamma,c\lp d\theta\rp]=0$ and for each $x\in M$ the spectrum of the operator section $\gamma c\lp d\theta\rp$ satisfies
    \[
        \textup{Spec}\lp\pm i\lp\gamma_x-1\rp c\lp d\theta_x\rp\rp \bigcap\textup{Spec}\lp\pi_{1}\lp\sigma_2\lp\nabla^E\big|_{\F}\rp^*\nabla^E\big|_{\F}\rp_x\rp=\emptyset
    \]
    where $\pi_{ 1}$ is the $\lambda= 1$ Schrodinger representation of $\t_{\F}M_x$. Then the operator
    \[
        \bD^E_{\gamma}\circlearrowright\Gamma\lp E\otimes\wedge[\theta]\rp
    \]
    is graded Rockland operator of order $1$ and is Fredholm if $M$ is closed. Thus $\bD^{E,+}_{\gamma}$ and $\bD^{E,-}_{\gamma}$ are Fredholm and one has
    \[
        \textup{Ind}\bD^E_{\gamma}=\textup{Ind}\bD^{E,+}_{\gamma}+\textup{Ind}\bD^{E,-}_{\gamma}.
    \]
\end{theorem}
\begin{proof}
    At $x\in M$, the graded Rockland cosymbol is given by
    \[
        \widecheck{\sigma}^G_1\lp\bD^E_{\gamma}\rp=D^{\t_{\F}M_X^L,E^L_x\otimes\wedge[\theta]}+\varepsilon^{\T_x}_{\theta_x}T_x+\gamma_xc\lp d\theta_x\rp\iota^{\T_x}_{T_x}
    \]
    If we square the above cosymbol and use Proposition \ref{WF} we obtain
    \begin{align*}
        \widecheck{\sigma}^G_1\lp\bD^E_{\gamma}\rp^2&=\sigma_2\lp\lp\nabla^E\big|_{\F}\rp^*\nabla^E\big|_{\F}\rp_x+\lp\gamma_x-1\rp c\lp d\theta_x\rp T_x+[D^{\t_{\F}M^L_x,E^L_x\otimes\wedge[\theta]},\gamma_x c\lp d\theta_x\rp\iota^{\T_x}_{T_x}].
    \end{align*}
    The assumed conditions on $\gamma$ imply that applying any irreducible representation of $\t_{\F}M_x$ to $\widecheck{\sigma}^G_1\lp\bD^E_{\gamma}\rp^2$ will result in an injective operator. Thus $\bD^{E}_{\gamma}$ is a graded Rockland operator of order $1$. Similarly, if we swap the grading of $E\otimes\wedge[\theta]$ by defining $E\otimes\wedge^0[\theta]$ to be order $1$ and $E\otimes\wedge^1[\theta]$ to be order $0$, then the graded cosymbol of the $L^2$-formal adjoint $\lp\bD^E_{\gamma}\rp^*$ is given at $x\in M$ by
    \[
        \widecheck{\sigma}^G_1\lp\lp\bD^E_{\gamma}\rp^*\rp_x=D^{\t_{\F}M^L_x,E^L_x\otimes\wedge[\theta]_x}-\iota^{\T}_{T_x}T_x-\gamma^*_xc\lp d\theta_x\rp\varepsilon^{\T_x}_{\theta^{\T}_x}
    \]
    where $\gamma^*_x$ is the adjoint of $\gamma_x$. If we square the cosymbol above then the same calculation for the cosymbol of $\bD^E_{\gamma}$ gives
    \begin{align*}
        \widecheck{\sigma}^G_1\lp\lp\bD^E_{\gamma}\rp^*\rp_x^2&=\sigma_2\lp \lp\nabla^E\big|_{\F}\rp^*\nabla^E\big|_{\F}\rp +\lp\gamma^*_x-1\rp c\lp d\theta_x\rp T_x-[D^{\t_{\F}M^L_x,E^L_x\otimes\wedge[\theta]_x},\gamma^*_xc\lp d\theta_x\rp\varepsilon^{\T_x}_{\theta^{\T}_x}].
    \end{align*}
    Note that $\text{Spec}\lp\pi_1\lp \sigma_2\lp\nabla^E\big|_{\F}\rp^*\nabla^E\big|_{\F}\rp\rp$ is contained within the real axis $\bR\subseteq\bC$. Thus if the pointwise spectrum of $\pm i\lp\gamma_x-1\rp c\lp d\theta_x\rp$ is disjoint from this subset then the spectrum of the adjoint, which is simply the conjugation of the eigenvalues, will also be disjoint from this subset. But the pointwise adjoint is given by
    \[
        \lp\pm i\lp\gamma_x-1\rp c\lp d\theta_x\rp\rp^*=\pm i\lp\gamma^*_x-1\rp c\lp d\theta_x\rp.
    \]
    Thus one has
    \[
        \text{Spec}\lp\pm i\lp\gamma^*_x-1\rp c\lp d\theta_x\rp\rp\bigcap \text{Spec}\lp\pi_1\lp\sigma_2\lp\nabla^E\big|_{\F}\rp^*\nabla^E\big|_{\F}\rp\rp=\emptyset
    \]
    which implies that $\lp\bD^{E}_{\gamma}\rp^*$ is a graded Rockland operator of order $1$. By Corollary 5.5 in \cite{DH}, $\bD^{E}_{\gamma}$ is Fredholm. If we split $E\otimes\wedge[\theta]$ into the respective even and odd components then $\bD^E_{\gamma}$ takes the block diagonal form
    \[
        \bD^{E}_{\gamma}=\begin{bmatrix}
            0 & \bD^{E,-}_{\gamma} \\
            \bD^{E,+}_{\gamma} & 0
        \end{bmatrix}
    \]
    which implies that $\bD^{E,+}_{\gamma}$ and $\bD^{\E,-}_{\gamma}$ are Fredholm and their respective Fredholm indices sum to the Fredholm index of $\bD^E_{\gamma}$.
\end{proof}
To give a simple example with non-zero index, suppose that we take $M$ to be three dimensional and $E=\wedge\F^*$. Choose an almost complex structure $J_{\F}$ of $\F$ that is compatible with $d\theta$ and take $g_{\F}\lp\cdot,\cdot\rp=d\theta\lp J_{\F}\cdot,\cdot\rp$, which is always possible due to the contractibility of the space of compatible almost complex structures for $d\theta$. In this case, for all $x\in M$ we have
\[
    \text{Spec}\lp\pi_1\lp\sigma_2\lp\nabla^E\big|_{\F}\rp^*\rp\nabla^E\big|_{\F}\rp=2\bZ_{\geq 0}+1
\]
Note then that $c\lp d\theta\rp$ leaves $\wedge^{ev}\F^*$ and $\wedge^{odd}\F^*$ invariant and when restricted to $\wedge^{ev}\F^*$, under the trivialization $\wedge^{ev}\F^*\simeq \underline{\bC^2}_M$, is the constant matrix
\[
    c\lp d\theta\rp\big|_{\wedge^{ev}\F^*}=\begin{bmatrix}
        0 & -1 \\
        1 & 0
    \end{bmatrix}.
\]
We let $\E_{\pm i}$ be the eigenbundles corresponding to the eigenvalues $\pm i$ of the matrix above. In this case we can take $\gamma:M\rightarrow \bC\setminus 2\bZ$ and define 
\[ \widetilde{\gamma}:=\begin{cases} 
      \gamma & \text{on } \Gamma\lp \E_i\rp \\
      1 & \text{on } \Gamma\lp\E_{-i}\rp \\
      Id & \text{on } \Gamma\lp \F^*\rp 
   \end{cases}.
\]
By Theorem \ref{NIT} the operator $\bD^{\wedge\F^*}_{\widetilde{\gamma}}$ is hypoelliptic and we wish to compute its index. First note that 
\[
    2\text{Ind}\bD^{\wedge\F^*}_{\widetilde{\gamma}}=\text{Ind}\lp\bD^{\wedge\F^*}_{\gamma}\rp^2.
\]
If we use the splitting $\wedge\F^*\otimes\wedge[\theta]\simeq \E_i\otimes\wedge[\theta]\oplus\E_{-i}\otimes\wedge[\theta]\oplus\F^*\otimes\wedge[\theta]$, then the cosymbol 
$\widecheck{\sigma}^G_2\lp\bD^{\wedge\F^*}_{\widetilde{\gamma}}\rp^2$ is given by
\[
    \begin{bmatrix}
        \lp\lp\nabla^{\F,0}\big|_{\F}\rp^*\nabla^{\F,0}\big|_{\F}+i\lp\gamma-1\rp \nabla^{\F,0}_T\rp\otimes I_{\wedge[\theta]} & UT_{12} & UT_{13} \\
        0 & \lp\nabla^{\F,0}\big|_{\F}\rp^*\nabla^{\F,0}\big|_{\F}\otimes I_{\wedge[\theta]} & UT_{23} \\
        0 & 0 & \lp\nabla^{\F,0}\big|_{\F}\rp^*\nabla^{\F,0}\big|_{\F}\otimes I_{\wedge[\theta]}
    \end{bmatrix}
\]
We can apply a homotopy to the symbol of $\lp\bD^{\wedge\F^*}_{\gamma}\rp^2$ to remove the upper triangular components by simply scaling these components by $t$ for $0\leq t\leq 1$, and note that the homotopy preserves the graded Rockland condition for all $0\leq t\leq 1$. Moreover, the last two diagonal blocks have index zero implying that
\begin{align*}
    2\text{Ind}\bD^{\wedge\F^*}_{\widetilde{\gamma}}&=\text{Ind}\lp \lp\lp \nabla^{\F,0}\big|_{\F}\rp^*\nabla^{\F,0}\big|_{\F}+i\lp\gamma-1 \rp\nabla^{\F,0}_T\rp\otimes I_{\wedge[\theta]}\rp\\
    &=2\text{Ind}\lp\lp \nabla^{\F,0}\big|_{\F}\rp^*\nabla^{\F,0}\big|_{\F}+i\lp\gamma-1 \rp\nabla^{\F,0}_T\rp
\end{align*}
By \cite{Van3}, the index formula for the above operator can be computed. Let $L\subset M$ be a link in $M$ that is half the Poincaré dual of the euler class $e\lp H\rp$ of $H$ in integral homology
\[
    P.D.\lp e\lp H\rp\rp=2[L].
\]
If $L=\sqcup_j L_j$, where each $L_j\subseteq M$ is a connected component of $L$, then we can restrict $\gamma$ to $L_j$ and obtain a map from $\bS^1$ to $\bC$. For $k\in \bZ$ we let $W\lp\gamma\big|_{L_j},k\rp$ denote the winding number of this loop around $k$. The index formula then reads
\[
    \text{Ind}\lp\lp\nabla^{\F,0}\big|_{\F}\rp^*\nabla^{\F,0}\big|_{\F}+i\lp\gamma-1\rp\nabla^{\F,0}_T\rp=\sum_{k\in 2\bZ+1}k\sum_jW\lp\lp\gamma-1\rp\big|_{L_j},k\rp
\]
which implies the index formula
\[  
    \text{Ind}\bD^{\wedge\F^*}_{\widetilde{\gamma}}=\sum_{k\in 2\bZ+1}k\sum_jW\lp\lp\gamma-1\rp\big|_{L_j},k\rp.
\]
Furthermore observe that $\bD^{E,+}_{\widetilde{\gamma}}$ only depends on the even component $\widetilde{\gamma}^+$ whereas $\bD^{E,-}$ only depends on the odd component $\widetilde{\gamma}^-$. A similar argument to Theorem \ref{NI} will show that for this partciular choice of operator section $\widetilde{\gamma}$, the cosymobl of $\bD^{E,-}_{\widetilde{\gamma}}$ will be homotopic to the cosymbol of its $L^2$-formal adjoint thus giving zero Fredholm index. Hence we also have the index computation
\[
    \text{Ind}\bD^{E,+}_{\widetilde{\gamma}}=\sum_{k\in 2\bZ+1}k\sum_jW\lp\lp\gamma-1\rp\big|_{L_j},k\rp.
\]
To produce examples with non-vanishing index, one can proceed as follows: Note that one has $[M,\bS^1]\simeq H^1\lp M,\bZ\rp$, thus if $e\lp \F\rp\in H^2\lp M,\bZ\rp$ has non-trivial image in $H^2\lp M,\bQ\rp$ then there must exist a curve $\gamma:M\rightarrow \bS^1$ for which 
\[
    e\lp \F\rp\cup[\gamma]\neq 0.
\]
By Poincaré duality, this would imply that 
\[
    \la [\gamma],2P.D\lp e\lp \F\rp\rp\ra\neq 0
\]
which using the relationship $P.D.\lp e\lp \F\rp\rp=2[L]$ would give
\[
    \la[\gamma],[L]\ra\neq 0.
\]
If we center the $\bS^1$ about a fixed even integer $k_0$ in $\bC$ then we would have
\[
    \sum_jW\lp \lp\gamma-1\rp\big|_{L_j},k_0-1\rp=\la[\gamma],[L]\ra\neq 0.
\]
If we further make the radius of $\bS^1$ sufficiently small then $\gamma$ will avoid all other even integers and have zero winding number about them
\[
    \sum_jW\lp\gamma\big|_{L_j},k\rp=0,\quad k\neq k_0.
\]
Using our index formula for $\bD^{\wedge\F^*}_{\widetilde{\gamma}}$ we obtain
\[
    \text{Ind}\bD^{\wedge\F^*}_{\widetilde{\gamma}}=\lp k_0-1\rp\la[\gamma],[L]\ra\neq 0.
\]
Of course everything hinges on then producing an example of a contact manifold $\lp M,\theta\rp$ for which $e\lp \F\rp\neq 0$ in $H^2\lp M,\bQ\rp$. These can be constructed using particular choices of surface diffeomorphisms $f\in \text{Diff}\lp\Sigma\rp$ for $\Sigma$ a closed oriented surface, and considering Dehn surgery on the corresponding mapping torus $M_f$. The interested reader can consult \cite{SY} and the references therein for producing such examples.
\subsection{The Dream: A Path Integral Approach}
In this last section, we will formally describe the ideas which led to the consideration of our operator. If one considers the loop space of a Riemannian manifold $\lp M,g\rp$, denoted $\L M:=C^{\infty}\lp\bR/\bZ,M\rp$, then $\L M$ inherits a natural action of $\bR/\bZ$ given by the rotation of a loop. In this case, the fixed point set of this action is seen to be the set of constant loops which are naturally identified as $M$. Witten has shown in \cite{Wit} that the integral of the reciprocal equivariant Euler characteristic of the normal bundle $\mathscr{N}_{\L M}M\rightarrow M$ can be regularized in an appropriate sense and computes the $\hat{\A}$-genus of $M$
\[
    \int_M\chi^{-1}_{\bR/\bZ}\lp \mathscr{N}_{\L M}M\rp=\int_M\hat{\A}\lp M\rp.
\]
Witten then applies a supersymmetry argument to formally identify the reciprocal of the regularized equivariant Euler characteristic to the Euclidean path integral of the $\mathcal{N}=1$ non-linear sigma model with periodic boundary conditions on the source $\bR^{1,1}$ and a Riemannian manifold target $\lp M,g\rp$
\[
    \int_M\chi^{-1}_{\bR/\bZ}\lp \mathscr{N}_{\L M}M\rp=\int_{\L M}e^{-S},
\]
where the Lagrangian action $S$ is locally given as
\[
    S=\int_{\bR/\bZ}dt\hspace{0.1cm} \frac{1}{2}g_{\mu\nu}\dot{x}^{\mu}\dot{x}^{\nu}-\frac{1}{2}g_{\mu\nu}\lp\frac{D}{dt}\psi^{\mu}\rp\psi^{\nu} .
\]
Here $\psi^{\mu}$ is a Grassmann odd vector field defined on the target space which is pulled back to the loop $\mathbb{R} / \mathbb{Z}$;  this additional field needed to ensure supersymmetry.  The supersymmetry transformation $Q$ that leaves the action $S$ invariant
\[
    QS=0
\]
is the \textbf{\emph{Cartan equivariant differential}} on $\L M$ given as
\[
    Q=d^{\Cc}_{\bR/\bZ}:=d_{\L M}-\iota_{\bX}\circlearrowright\Omega\lp\L M,\bR\rp,\quad \bX_{\sigma}=\dot{\sigma}.
\]
From this perspective, the spinor Dirac operator $\slashed{D}$ emerges naturally as the quantized Noether charge of the $\mathcal{N}=1$ supersymmetric sigma model action $S$, i.e. $\sqrt{2}\hat{Q}=\slashed{D}$. By using the correspondence between operator methods and the path integral formalism, Witten then relates the path integral to the supertrace of the heat flow of the corresponding quantum Hamiltonian $\hat{H}=\frac{1}{2}\slashed{D}^2$
\[
    \int_{\L M}e^{-S}=\operatorname{STr}e^{-\hat{H}}.
\]
The supersymmetry algebra $2\hat{H}=\{ \hat{Q},\hat{Q^{\dagger}}\}$ implies that the supertrace of the heat flow will compute the index of the Dirac operator thus giving
\[
    \int_{\L M}e^{-S}=\text{Ind }\slashed{D}^+.
\]
From here, the index theorem for the spinor Dirac operator follows. For quite some time, various authors have previously tried to extract a rigorous argument from the above construction along with a proper definition of supersymmetric path integrals. We recommend that the interested reader consult \cite{HL} along with the references therein for the current state of the art.\par
If one is now equipped with a filtered manifold $\lp M,\F\rp$, then one can still consider the $\bR/\bZ$-space of smooth loops which are tangent to $\F$
\[
    \L^{\F} M:=\left\{\text{$\sigma\in\L M$ }\big|\text{ $\dot{\sigma}\in\F$}\right\}
\]
and note that the fixed point set of the $\bR/\bZ$ action is still the set of constant loops. Thus one is confronted with the question of whether one can make sense of the quantity
\[
    \int_M\chi^{-1}_{\bR/\bZ}\lp \mathscr{N}_{\L^{\F}M}M\rp
\]
and equate it to the index of a corresponding operator $\bD$
\[
    \int_M\chi^{-1}_{\bR/\bZ}\lp \mathscr{N}_{\L^{\F}M}M\rp=\text{Ind $\bD$}.
\]
One could then interpret the regularized reciprocal equivariant Euler characteristic as the supersymmetric localization of an action functional $S_{\F}$ corresponding to a constrained supersymmetric non-linear sigma model with target the sub-Riemannian manifold $\lp M,\F,g_{\F}\rp$, the constraint corresponding to the condition that the bosonic velocity of a map $\sigma:\bR^{1,1}\rightarrow M$ be tangent to $\F$. The analytic difficulties here would be in giving a rigorous meaning of the constrained supersymmetric path integral
\[
    \int_{\L^{\F}M}e^{-S_{\F}}.
\]
However, one can approach this problem directly from the operator side and instead try to cook up the quantized Noether current corresponding to the constrained action functional $S_{\F}$. One can gain a bit of a clue for how this could be practically implemented by studying the work of Bismut on the local families index theorem \cite{BisF}, for there one considers a distribution $\F$ which is the vertical distribution of a locally trivial fibration $\pi:M\rightarrow B$ equipped with an Ehresmann connection $\nabla$ and vertical fiber metric $g_{\F}$
\begin{displaymath}
  \begin{tikzcd}[column sep=2em]
   F \arrow{r}{ } \&  \lp M,g_{\F},\nabla\rp  \arrow{d}{\pi}\\
  \& B
  \end{tikzcd}
\end{displaymath}
One can again consider the smooth manifold of loops which are tangent to the distribution $\L^{\F}M$ and note that it sits within the loop space $\L M$ as a smooth $\bR/\bZ$-submanifold. The loop space of $M$ maps to the loop space of $B$ and one obtains an equivariant diagram
\begin{displaymath}
  \begin{tikzcd}[column sep=2em]
   \bR/\bZ\circlearrowright\L F \arrow{r} \&  \bR/\bZ\circlearrowright\L^{\F}M \arrow[hookrightarrow]{r}  \arrow{d}{\L^{\F}\lp\pi\rp}\& \bR/\bZ\circlearrowright\L M \arrow{d}{\L\lp\pi\rp} \& \\
   \& B \arrow[hookrightarrow]{r} \& \bR/\bZ\circlearrowright\L B.
  \end{tikzcd}
\end{displaymath}

The action $S$ for the $\mathcal{N}=1$ sigma model with target $\lp M,g_{\F}\oplus\pi^* g_B\rp$ gives an extension of the constrained action $S_{\F}$, i.e. $S\big|_{\L_{\F}M}=S_{\F}$ and in fact one can take this as the \emph{definition} of $S_{\F}$ in the case of the families index theorem. This implies an \emph{expected} relationship between fibered path-integrals holds
\[
    \lp\int^!_{\L M\rightarrow\L B}e^{-S}\rp\Big|_B=\int^!_{\L^{\F}M\rightarrow B}e^{-S_{\F}}.
\]
Rather than attempt to give a rigorous meaning to the fibered path integral of action functionals on $\L M$ over $\L B$, Bismut gives indication in \cite{BisF} that the Chern character of the Bismut superconnection $\bA$ is the analogue of the path-integration along the fiber on the operator theory side
\[
    \int^!_{\L M\rightarrow \L B}e^{-S}\Big|_B=\text{STr }e^{-\bA^2}.
\]
Given the success of the Bismut superconnection in the local families index theorem, one is lead to explore its behavior in the non-integrable setting, in hopes that it will give a rigorous meaning to the quantized Noether current $\hat{Q}$ for a constrained supersymmetric action $S_{\F}$ and provide a proper definition of the reciprocal equivariant Euler characteristic of $\mathscr{N}_{\L^{\F}M}M$. The \emph{hypoellipticity} that arises within the asymptotic superconnection, which is \emph{not} present in the classic story given above, adds a new layer of interest in pursuing these ideas further.


\begin{thebibliography}{1}
\bibitem{Bis3} Bismut, J. M. (1985). The Atiyah-Singer index theorem for families of Dirac operators: two heat equation proofs. \emph{Université de Paris-Sud. Département de Mathématique}.
\bibitem{BisF} Bismut, J.M. (1986). Localization Formulas, Superconnections, and the Index Theorem for Families. \emph{Communications in Mathematical Physics}, Vol 103, 127-166
\bibitem{DH} Dave, S., Haller, S. (2022). Graded hypoellipticity of BGG sequences. \emph{Ann Glob Anal Geom} 62, 721–789 
\bibitem{Go} Goffeng, M. (2024). Solving the index problem for (curved) Bernstein-Gelfand-Gelfand sequences. \emph{arxiv preprint}
\bibitem{GH} Goffeng, M., Helffer, B. (2025). The index of sub-laplacians: beyond contact manifolds. \emph{arxiv preprint}
\bibitem{GK} Goffeng, M., Kuzmin, A. (2024). Index theory of hypoelliptic operators on Carnot manifolds. \emph{arxiv preprint}
\bibitem{HL} Hanisch, F., Ludewig, M. (2022). A Rigorous Construction of the Supersymmetric Path Integral Associated to a Compact Spin Manifold. \emph{Communications in Mathematical Physics} vol 391, 1209–1239
\bibitem{Has} Hasselmann, S. (2014). Spectral Triples on Carnot Manifolds, PhD Thesis, \emph{arxiv preprint}, S. Hasselmann, 
\bibitem{Moh} Mohsen, O. (2021). On the index of maximally hypoelliptic operators. \emph{arxiv preprint}.
\bibitem{SY} Sivek, S., Yazdi, M. (2023). Thurston norm and Euler classes of tight contact structures. \emph{Bulletin of the London Mathematical Society}, vol 55
\bibitem{S} Stadtmüller, C. (2017). Horizontal Dirac Operators in CR Geometry. \emph{Publikationsserver der Humboldt-Universität}

\bibitem{Van} van Erp, E. (2010). The Atiyah-Singer index formula for subelliptic operators on contact manifolds. Part I. \emph{Annals of Mathematics}, 171, 1647-1681.
\bibitem{Van2} van Erp, E. (2010). Contact structures of arbitrary codimension and idempotents in the Heisenberg algebra. \emph{arXiv preprint}
\bibitem{Van3} van Erp, E. (2010). Noncommutative topology and the world’s simplest index theorem, \emph{Proc. Natl. Acad. Sci. U.S.A.} vol 107, no. 19.
\bibitem{Van-Yun} van Erp, E., Yuncken, R. (2017). A groupoid approach to pseudodifferential operators. \emph{Crelles Journal}, vol 756, 151-182.
\bibitem{Wit} Witten, E. (1982). Constraints on Supersymmetry Breaking. \emph{Nuclear Phys. B.} 202. no. 2, 253-316.
\end{thebibliography}
\end{document}